\documentclass[reprint,aps,amsmath,amssymb,pra,10pt,superscriptaddress,floatfix,tightenlines,nofootinbib,nobibnotes]{revtex4-2}

\usepackage{graphicx}
\usepackage[colorlinks,linkcolor=blue,urlcolor=blue, citecolor=blue]{hyperref}
\usepackage{xurl}
\usepackage{booktabs}
\usepackage{subcaption}
\usepackage{comment}
\usepackage{nicematrix}
\usepackage{tikz}
\usetikzlibrary{positioning, shapes.multipart}

\usepackage{dsfont}
\usepackage{physics}
\usepackage{nicefrac}
\usepackage{mathtools}
\usepackage{amsmath, amssymb, amsfonts, amsthm}
\usepackage{thmtools, thm-restate}

\usepackage{algorithm}
\usepackage{algorithmic}

\newtheorem{theorem}{Theorem}
\newtheorem*{theorem*}{Theorem}

\newtheorem{lemma}[theorem]{Lemma}
\newtheorem{definition}[theorem]{Definition}
\newtheorem{proposition}[theorem]{Proposition}

\newcommand{\beginsupplement}{
  \setcounter{section}{0}
  \renewcommand{\thesection}{S\arabic{section}}
  \setcounter{subsection}{0}
  \renewcommand{\thesubsection}{S\arabic{subsection}}
  \setcounter{equation}{0}
  \renewcommand{\theequation}{S\arabic{equation}}
  \setcounter{figure}{0}
  \renewcommand{\thefigure}{S\arabic{figure}}
  \setcounter{table}{0}
  \renewcommand{\thetable}{S\arabic{table}}

  \setcounter{theorem}{0}
  \renewcommand{\thetheorem}{S\arabic{theorem}}
  \setcounter{lemma}{0}
  \renewcommand{\thelemma}{S\arabic{lemma}}
  \setcounter{proposition}{0}
  \renewcommand{\theproposition}{S\arabic{proposition}}
  \setcounter{definition}{0}
  \renewcommand{\thedefinition}{S\arabic{definition}}
}

\DeclareMathOperator{\polylog}{polylog}
\DeclareMathOperator{\poly}{poly}

\DeclareMathOperator{\Prob}{\mathbb{P}}
\DeclareMathOperator{\diam}{diam}

\allowdisplaybreaks

\begin{document}

\title{Sparse-Blossom Decoding in \texorpdfstring{$o(1)$}{o(1)} Time}
\author{Ryo Mikami}
\email{ryo-223606@g.ecc.u-tokyo.ac.jp}
\affiliation{
Department of Computer Science, Graduate School of Information Science and Technology, The University of Tokyo, 7--3--1 Hongo, Bunkyo-ku, Tokyo, 113--8656, Japan
}
\author{Hayata Yamasaki}
\email{hayata.yamasaki@gmail.com}
\affiliation{
Department of Computer Science, Graduate School of Information Science and Technology, The University of Tokyo, 7--3--1 Hongo, Bunkyo-ku, Tokyo, 113--8656, Japan
}

\begin{abstract}
Matching-based decoding is widely used in quantum error correction, and accelerating it is key to enabling fast and scalable fault-tolerant quantum computation. 
Minimum-weight perfect matching (MWPM) decoding provides rigorous guarantees for error suppression, while sparse blossom enables its practical implementation at modest problem sizes.
However, the runtime of existing sparse-blossom implementations unavoidably increases with problem size, motivating a rigorous parallelization framework that guarantees correctness and a runtime shorter than the syndrome-extraction timescale. 
Here, we present such a framework and prove that the resulting parallel sparse-blossom algorithm produces the same correction as the original, non-parallel sparse blossom. 
For the rotated surface code with code distance $d$ and physical error rates below a finite threshold, we prove that the average parallel runtime of decoding for $O(d)$ rounds of syndrome extraction is upper bounded by a quasi-polylogarithmic function of $d$. 
For a $d$-round decoding window, this implies that the average parallel runtime per round is $o(1)$. 
We also perform numerical simulation to identify conditions under which the parallel runtime per round decreases with increasing code distance.
These results suggest that increasing code distance need not lead to longer parallel decoding times, providing a foundation for scalable parallel matching-based decoding.
\end{abstract}

\maketitle

\section*{Introduction}
A central goal of quantum technology is to realize large-scale quantum computation, whose practical execution will require fault-tolerant quantum computation (FTQC).
The runtime and overall efficiency of FTQC can be substantially affected by the overhead of quantum error correction~\cite{Yamasaki2024,Wills2025,Tamiya2026,Takada2026}.
Matching-based decoders, particularly minimum-weight perfect matching (MWPM) decoders, are widely applicable across quantum error-correcting codes because they admit polynomial-time algorithms while providing rigorous guarantees for error suppression~\cite{dennis2002topological,fowler2012proof}.
For quantum memory, faster alternative decoders may be available~\cite{delfosse2021almost,yoshida2026prooffinitethresholdunionfind,lake2025fastofflinedecodinglocal,diego2025,varsamopoulos2017decoding}; however, MWPM decoding also plays an important role beyond quantum memory.
For example, magic-state preparation is a key ingredient of leading architectures for universal FTQC, and recent low-overhead protocols, including magic-state cultivation, use the complementary gap obtained from MWPM decoding as a confidence metric for post-selection~\cite{gidney2024magicstatecultivationgrowing,bombin2024,Gidney2025}.
Accelerating MWPM decoding is therefore a broadly applicable classical technology for scalable quantum computation, relevant not only to quantum-memory decoding but also to the efficiency of essential components of universal FTQC.

Conventional algorithms for finding an MWPM, which are based on Edmonds' blossom algorithm~\cite{edmonds1965paths,edmonds1965maximum}, have a polynomial worst-case runtime in the size of the input graph, motivating extensive efforts to improve their practical efficiency~\cite{micali1980v,vazirani2020proofmvmatchingalgorithm,kolmogorov2009blossom}.
In the context of topological-code decoding, representative implementations of MWPM decoders, such as sparse blossom~\cite{Higgott2025sparseblossom} and fusion blossom~\cite{wu2023fusion}, are also based on variants of the blossom algorithm. 
Sparse blossom solves the matching problem directly on the sparse detector graph and achieves high decoding throughput in practice~\cite{Higgott2025sparseblossom}. 
Sparse-blossom-based decoding has been deployed in a real-time multi-threaded system using a related fusion strategy~\cite{google2023suppressing,google2025quantum}.
Fusion blossom also exploits parallelism through graph partitioning and fusion~\cite{wu2023fusion}. 
Several other approaches have been developed to parallelize or accelerate MWPM decoding. 
Micro blossom exploits hardware parallelism at the level of the detector graph~\cite{wu2025micro}. 
A distinct algebraic approach achieves polylogarithmic parallel runtime for MWPM using determinant-based algorithms~\cite{Takada2026,Mikami2026}. 
Reference~\cite{fowler2013minimum} also argued that MWPM decoding for topological quantum error correction could achieve an average parallel runtime of $O(1)$ per syndrome-extraction round.

\begin{figure*}[t]
    \centering
    \includegraphics[
        width=\textwidth,
        keepaspectratio
    ]{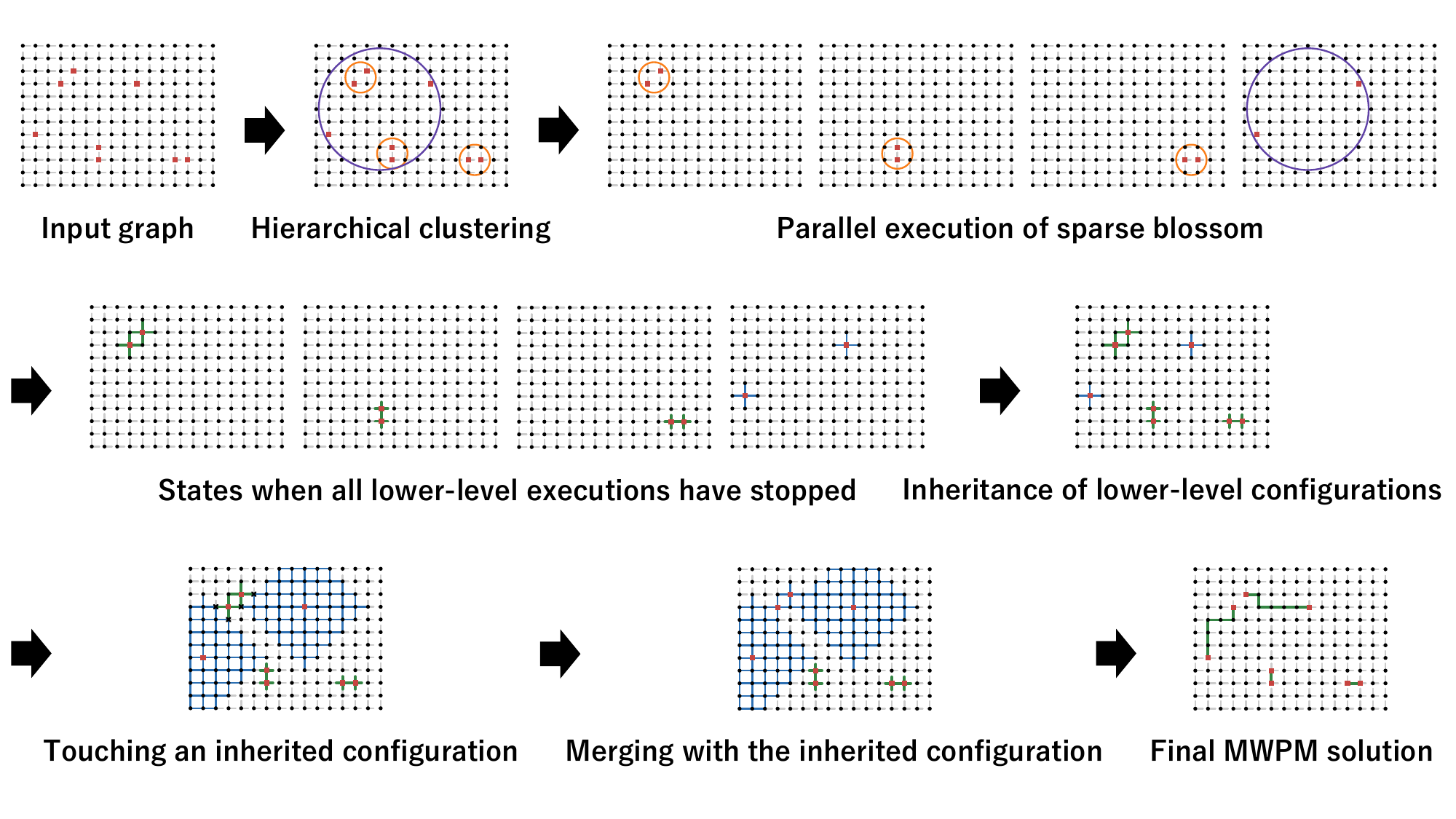}
    \captionsetup{justification=raggedright,singlelinecheck=false}
    \caption{
        Schematic illustration of the parallel sparse-blossom decoder.
        Starting from an input graph with active detectors, the decoder constructs a hierarchy of processing clusters and executes sparse blossom concurrently on the resulting cluster-wise subproblems.
        The orange circles denote lower-level processing clusters, whereas the purple circle denotes a higher-level processing cluster.
        Once the lower-level executions stop, their matching configurations are explicitly inherited by the higher-level subproblem.
        The green edge segments indicate portions of graph edges occupied by graph-fill regions associated with stopped lower-level executions, whereas the blue edge segments indicate portions occupied by graph-fill regions that are still growing.
        The higher-level sparse-blossom dynamics subsequently interact and merge with the inherited configurations, eventually producing the final MWPM solution.
    }
    \label{fig:parallel_sparse_blossom_overview}
\end{figure*}

In view of the broad applications of MWPM decoding in practice, it is important to accelerate sparse blossom, where sparse error configurations provide an opportunity for parallelization. 
For scalable FTQC, decoding imposes two distinct performance requirements: throughput and latency~\cite{Skoric2023,Caune2026}. 
As a representative setting, we consider the decoding problem of a spacetime window containing $O(d)$ rounds of syndrome extraction for a distance-$d$ surface code~\cite{fowler2012surface}.
For throughput, in parallel-window decoding, a long decoding latency for each individual window does not by itself cause a backlog, because multiple decoding windows can be processed concurrently with sufficient computational resources~\cite{Skoric2023}.
In sliding-window decoding, by contrast, decoding windows are processed sequentially, so the average decoding throughput must keep pace with the syndrome-generation rate to prevent a growing backlog~\cite{Skoric2023}.
The relevant quantity is therefore the decoding time per syndrome-extraction round. 
As for latency, adaptive operations such as gate teleportation require a decoding-dependent classical feed-forward before the subsequent computation can proceed. 
Increasing throughput by processing other windows in parallel does not reduce the waiting time for the required decoding result. 
In this case, the relevant quantity is the total decoding time of an individual $O(d)$-round window.
From the empirical fit in Fig.~10 of Ref.~\cite{Higgott2025sparseblossom}, the decoding time per round of sparse blossom scales as $O(d^{2.34})$ over the numerically investigated range, corresponding to a total decoding time of $O(d^{3.34})$ for a $d$-round decoding window. 
For scalable decoding, however, it is desirable that the decoding time per round scales as $o(1)$, corresponding to the total decoding latency of $o(d)$  for an $O(d)$-round window, to make decoding asymptotically faster than syndrome extraction itself.

However, constructing a parallel sparse-blossom decoder from the sparsity of error configurations presents two major challenges.
First, the dynamics of the sparse-blossom algorithm are structurally complex.
The algorithm dynamically creates and modifies matches, blossoms, and alternating trees, while their graph-fill regions may grow, freeze, or shrink~\cite{Higgott2025sparseblossom}. 
It is therefore nontrivial to determine when spatially separated parts of the decoding problem can be processed independently and in parallel. 
Second, the actual error configuration is hidden from the decoder, and only the resulting active detectors are available. 
An implementable parallel decoder must therefore identify independently processable regions solely from the observed active detectors while guaranteeing the same decoding output as the original, non-parallel sparse blossom, which we refer to as global sparse blossom.

In this work, we develop a parallel sparse-blossom decoder (see Fig.~\ref{fig:parallel_sparse_blossom_overview}).
We establish its correctness by proving that the decoder's output is equivalent to that of the global sparse-blossom algorithm, and hence to that of conventional MWPM decoding as shown in~\cite{Higgott2025sparseblossom}.
For a distance-$d$ rotated surface code~\cite{bombin2007optimal} under a circuit-level local stochastic error model with a physical error rate below a finite threshold, we prove that the average parallel runtime for an $O(d)$-round decoding window, whose detector graph has size $O(d^3)$, is upper bounded by a quasi-polylogarithmic function of $d$.
For a decoding window of $d$ rounds, this implies an average parallel runtime of $o(1)$ per syndrome-extraction round.
We also numerically evaluate the finite-size parallelism of the proposed decoder by comparing its parallel performance with that of global sparse blossom under the idealized assumption that the additional overhead of parallelization is negligible.
Under this idealized setting, this comparison shows that, in a low-error regime, the parallel runtime per syndrome-extraction round can decrease with increasing code distance for $d= 9, 13, ..., 49$.

To establish these results, we use an error-clustering framework building on Ref.~\cite{yoshida2026prooffinitethresholdunionfind}, which extends a line of approaches originating from the hierarchical clustering framework of Gács~\cite{GACS198615,Gacs2001} and subsequently developed in Refs.~\cite{PhysRevLett.111.200501,2004PhDT.......155H,Balasubramanian2026localautomatond,lake2025fastofflinedecodinglocal}.
In such a framework, an error configuration is decomposed into hierarchical clusters surrounded by error-free buffer regions. 
In a detector graph, an error configuration can be represented by a set of faulty edges. 
The detector vertices incident to an odd number of these edges correspond to the observed active detectors, while the edges may also terminate at a boundary. 
Our idea is to group these faulty edges into sufficiently well-separated error clusters so that the sparse-blossom decoding process associated with different clusters can be analyzed independently, while the cluster statistics can be used to analyze the parallel runtime.
However, a mere application of the analysis in Ref.~\cite{yoshida2026prooffinitethresholdunionfind} for the union-find (UF) decoder is insufficient for our results since the sparse-blossom algorithm has much more complex dynamics on weighted detector graphs than the UF decoder on unweighted detector graphs.
To control the complex decoding dynamics, we prove that the decoding process within a sufficiently isolated region terminates before it can interact with any distinct region at the same or a larger scale. 
We then show that the required hierarchy can be constructed directly from the observed detection events and that the local decoding states obtained at smaller scales can be consistently passed to larger-scale regions.

Therefore, our results suggest that increasing the total decoding work does not necessarily lead to longer parallel decoding times, providing a foundation for scalable parallel matching-based decoding.
More broadly, our results pave the way toward large-scale FTQC in which classical decoding need not become a scalability bottleneck.

\section*{Results}

\paragraph*{Decoding task}
We define the decoding task solved by the sparse blossom~\cite{Higgott2025sparseblossom}. 
For concreteness, we present it for an $[[n=d^2,k=1,d]]$ two-dimensional (rotated) surface code with open boundary conditions under circuit-level noise, while the algorithm naturally works for a broader class of quantum error-correcting codes compatible with matching-based decoders~\cite{Kubica2023efficientcolorcode,sahay2026matchingdecoderbivariatebicycle,tan2026generalizedmatchingdecoders2d}.
Error correction in the surface code is performed through syndrome extraction, which is represented as a \textit{syndrome-extraction circuit} consisting of state preparation, gates, measurement, and idle operations on physical qubits. 
Each operation in this circuit is referred to as a \textit{location}. 
The syndrome-extraction circuit defines a \textit{detector}, which is the parity of syndrome measurement outcomes along the time direction. Each syndrome measurement corresponds to the measurement of a stabilizer generator. 
Here, an $X$ or $Z$ error activates one or two detectors.

We define a \textit{detector graph} $G=(V, E)$, where $V$ is the set of vertices consisting of the detectors and \textit{boundary vertices}. 
Edges and boundary vertices are defined as follows:
\begin{itemize}
    \item if there exists an $X$ or $Z$ error activating two detectors, we add an edge between those two detectors;
    \item if there exists an $X$ or $Z$ error activating one detector, we add a corresponding boundary vertex and an edge between the detector and the boundary vertex.
\end{itemize}

Given the set $D\subseteq V$ of active detectors, we define the minimum-weight \textit{embedded} matching (MWEM) problem~\cite{Higgott2025sparseblossom}, which finds a minimum-weight set of edges $M\subseteq E$ such that every detector in $D$ has an odd degree in $M$, every detector outside $D$ has an even degree in $M$, and boundary vertices may have arbitrary degrees. 

Sparse blossom solves the MWEM problem directly on the detector graph, rather than first constructing the complete graph of shortest paths between all pairs of detection events. 
During the blossom search, it grows regions from active detectors, identifying the shortest path between a pair of active detectors only when growing regions collide with other regions or with a boundary, thereby recovering the relevant shortest-path information on demand (see Supplementary Information~\ref{sup:sparse_blossom} for more details).
Blossom-based MWEM decoding searches for a globally optimal solution, and it is challenging to decompose the problem into independent local subproblems that can be processed in parallel. 
To achieve this, we need a framework that enables parallelizable problem decomposition while still guaranteeing global optimality.

\begin{figure*}[t]
    \centering
    \includegraphics[
        width=\textwidth,
        keepaspectratio
    ]{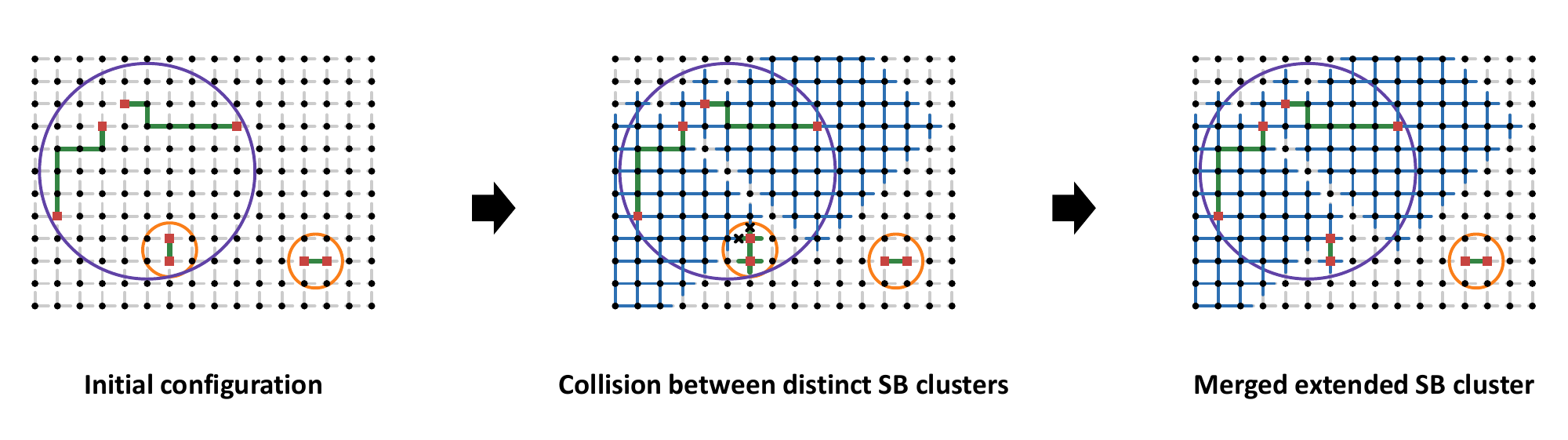}
    \captionsetup{justification=raggedright,singlelinecheck=false}
    \caption{
        Schematic illustration of the formation of an extended SB cluster.
        At $t=0$, each error cluster together with the graph-fill regions associated with its active detectors forms a separate extended SB cluster.
        As the graph-fill regions grow, regions associated with two initially distinct extended SB clusters can come into contact.
        Once the corresponding collision event occurs, the two extended SB clusters merge, and the error clusters connected through the resulting sparse-blossom dynamics belong to the same extended SB cluster.
    }
    \label{fig:extended_sb_cluster}
\end{figure*}

\paragraph*{Problem setting}
We assume that errors occur according to a circuit-level local stochastic Pauli error model~\cite{2685179.2685184,Tamiya2026}; i.e., each location $j$ has an error parameter $p_j$, and the probability that any set $S$ of locations is included in a set $F$ of faulty locations, in which Pauli errors occur, is upper-bounded by
\begin{align}
    \Prob[ F \supseteq S ] \le \prod_{j \in S}p_j.
\end{align}
For each edge $e\in E$, let $\mathcal{J}(e)$ be the set of locations $j$ such that a Pauli error at location $j$ activates the detectors (or the detector and the boundary vertex) corresponding to $e$.
We define
\begin{align}
    q_e := \sum_{j\in \mathcal{J}(e)} p_j.
\end{align}
The edge weights are defined as
\begin{align}
    w(e) := -\log q_e,
\end{align}
where we use natural logarithm.
In sparse blossom, edge weights are scaled and described as even integers to ensure that all events occur at integer times.

\begin{figure*}[t]
    \centering

    \begin{subfigure}{0.95\textwidth}
        \centering
        \includegraphics[
            width=\linewidth,
            keepaspectratio
        ]{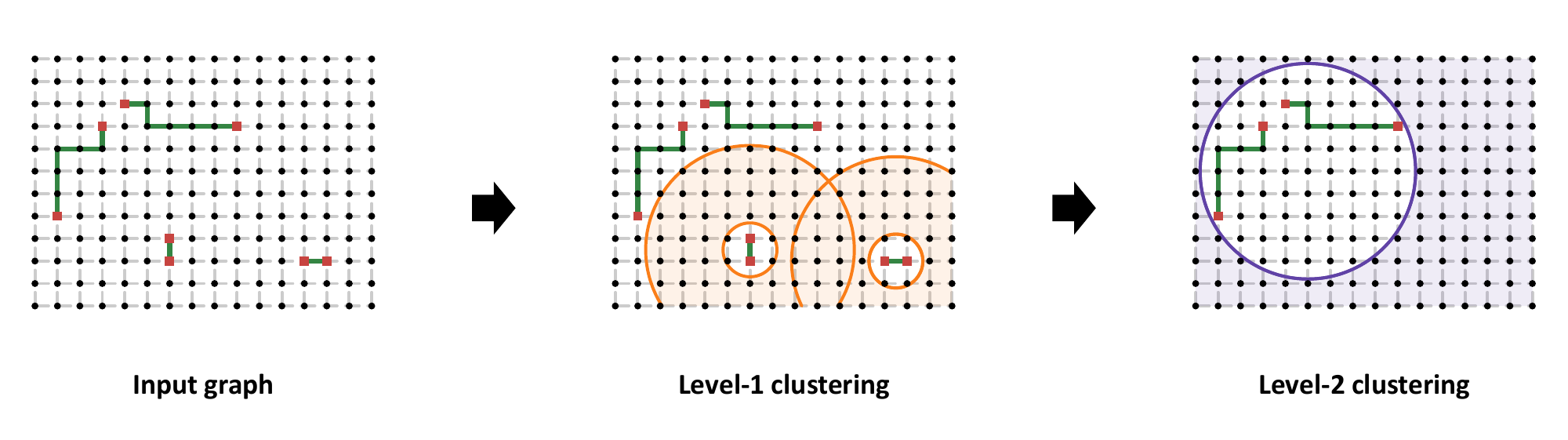}
        \caption{Construction of error clusters.}
        \label{fig:error_cluster_construction}
    \end{subfigure}

    \vspace{0.8em}

    \begin{subfigure}{0.95\textwidth}
        \centering
        \includegraphics[
            width=\linewidth,
            keepaspectratio
        ]{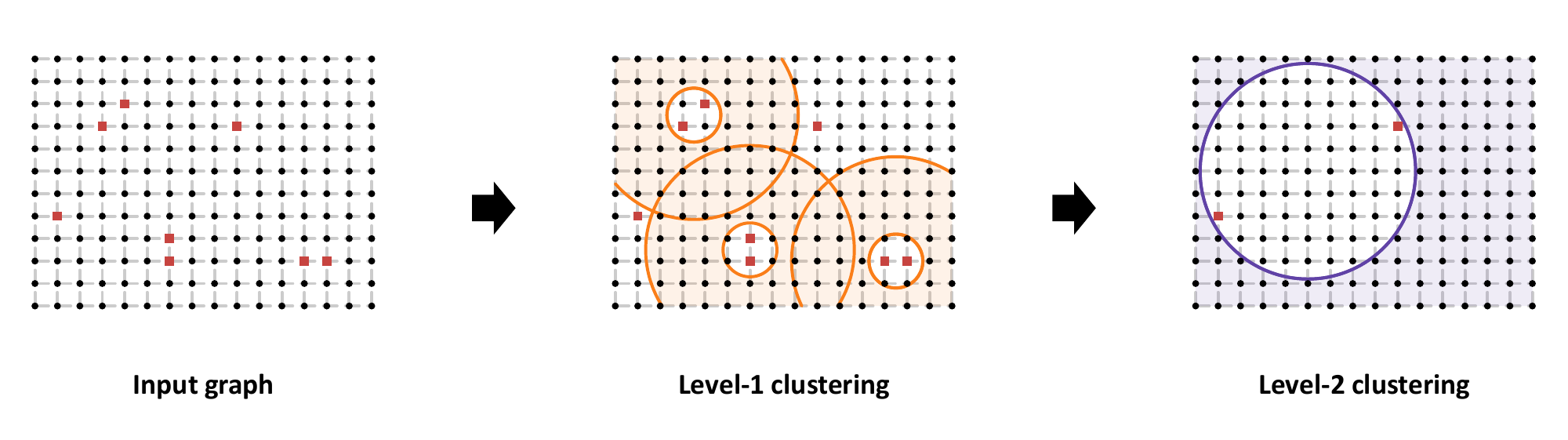}
        \caption{Construction of processing clusters.}
        \label{fig:processing_cluster_construction}
    \end{subfigure}

    \captionsetup{justification=raggedright,singlelinecheck=false}
    \caption{
        Schematic illustration of hierarchical cluster construction.
        (a) Error clusters are constructed from the hidden faulty-edge configuration.
        (b) Processing clusters are constructed from the observed active detectors.
        The regions enclosed by the orange and purple curves represent level-1 and level-2 clusters, respectively, while the lightly shaded regions represent the corresponding error-free buffer regions.
        In this example, the residual set becomes empty at level-2, so no higher-level clusters are constructed.
        If a nonempty residual set remains, the construction proceeds recursively to level-3 and beyond.
        }
    \label{fig:cluster_constructions}
\end{figure*}

\paragraph*{The algorithm for the decoder based on parallel sparse blossom}

Figure~\ref{fig:parallel_sparse_blossom_overview} provides a schematic overview of our parallel sparse-blossom decoder.
In the following, we review how the sparse blossom works. 
We then identify the main challenges in parallelizing sparse blossom, explain the key ideas used to overcome them, and present a parallel sparse-blossom decoder.

In sparse blossom~\cite{Higgott2025sparseblossom}, each active detector initially defines a graph-fill region.
Each region grows through the detector graph.
Collisions between growing regions or with a boundary trigger updates of matches, alternating trees, and blossoms.
The process continues until no unresolved alternating tree remains.
We parallelize these dynamics by exploiting the geometric localization of errors.
An actual error configuration is represented by a set of faulty edges $N\subseteq E$, whereas only the active detectors induced by these edges are observed.
Under the local stochastic error model, the faulty-edge set can be analyzed in terms of well-separated error clusters.
When such clusters are sufficiently separated, we show that their sparse-blossom dynamics terminate before interacting, allowing them to be processed independently and in parallel.
We first establish this property for a hierarchy defined from the hidden faulty-edge configuration then replace it with the clusters constructed solely from the observed active detectors, called \textit{processing clusters}, while keeping the decoding output unchanged.
In the resulting parallel decoder, sparse blossom is executed concurrently on separated processing clusters.
Once a low-level cluster in the hierarchy stops, its matching configuration is retained through subsequent levels and remains unchanged until a higher-level processing-cluster execution modifies the corresponding sparse-blossom state.
When such an interaction occurs, the sparse-blossom dynamics continue from that stopped state.
The matching configurations remaining after all processing-cluster executions have stopped are finally combined to obtain the decoding output.

To guarantee this property rigorously, we use the error-clustering framework developed for the threshold analysis of the UF decoder in Ref.~\cite{yoshida2026prooffinitethresholdunionfind}, building on the hierarchical clustering framework introduced by Gács~\cite{GACS198615,Gacs2001}.
In this framework, errors are decomposed into well-separated hierarchical clusters surrounded by error-free buffer regions.
A level-$k$ error cluster has diameter at most $d_k^E$ and is separated from every distinct error cluster of level $k'\ge k$ by a buffer of width at least $b_k^E$, with $d_k^E,b_k^E=\exp[\Theta(k\log k)]$ for the parameter choice used below.
To track interactions generated by sparse-blossom dynamics, we group error clusters and the associated growing sparse-blossom regions into a composite region called an \textit{extended sparse-blossom (SB) cluster}, as illustrated in Fig.~\ref{fig:extended_sb_cluster}.
In particular, level-$k$ error clusters and sparse-blossom regions connecting them are called a level-k extended SB cluster.
We prove that each level-$k$ extended SB cluster stops growing before it collides with any level-$k'$ extended SB cluster with $k'\ge k$.
Consequently, spatially separated clusters can be processed independently and in parallel until the level-$k$ extended SB clusters have stopped.
This stopping property enables the global sparse-blossom execution to be reproduced by retaining stopped lower-level matching configurations until a subsequent higher-level execution modifies the corresponding sparse-blossom state, at which point the dynamics continue from that stopped state.

The error clusters discussed above are defined in terms of the faulty-edge set $N$ induced by the actual errors. 
However, these errors are not directly accessible to the decoder and are introduced only for the mathematical analysis of its dynamics. 
We therefore need to ensure that the decoder remains correct when the clusters are instead constructed from the active detectors available to the decoder, from which the decoder must infer the underlying error clusters.
Since $N$ is hidden from a decoder, the decoder cannot construct these clusters directly from the syndrome measurement. 
Instead of using the error-clustering framework based on faulty edges, we implement parallel sparse blossom using error clusters inferred solely from the active detectors.

However, a challenge arises from a potential mismatch between the cluster of errors and the cluster inferred from the active detectors, as shown in Fig.~\ref{fig:cluster_constructions}.
Since the sparse-blossom dynamics are seeded by the active detectors, the assignment of active detectors to isolated regions determines whether the corresponding cluster-wise executions can proceed independently.
Errors occur on edges, and active detectors appear at vertices on the boundaries of error clusters.
By focusing on the diameter and the error-free buffer width of a level-$k$ error cluster, the detection events can be clustered using vertex distances while preserving the same assignment. 
This correspondence is obtained by translating the edge-based diameter and buffer parameters into their vertex-based counterparts, with adjustments for the lengths of the detector-graph edges.
While this extension allows the same assignment by intentionally choosing the corresponding active detectors, the sets of active detectors produced by the algorithmic clustering procedure do not necessarily coincide with that assignment. 
Consequently, applying vertex-based clustering to the observed detection events yields a certain set of clusters, but there is no guarantee that the resulting clusters output the correct matching configuration.
Figure~\ref{fig:cluster_constructions} shows the differences between edge-based clustering and vertex-based clustering.

Nevertheless, we resolve this issue by proving that the clusters constructed algorithmically from the observed detection events correctly reproduce the execution associated with the hidden faulty-edge configuration.
The key ingredient is a decomposition lemma showing that the active detectors associated with every level-$k$ error cluster can be partitioned into active-detector-induced clusters of levels at most $k$.
By retaining stopped matching configurations through subsequent levels and allowing their sparse-blossom states to execute only when they are modified by subsequent processing-cluster executions, we show that the active-detector-based execution produces the same decoding result as the faulty-edge-based execution, even though mutually independent parts may be processed in a different order.
Hence, it produces the same decoding result as the original global sparse-blossom decoder.

These ingredients yield the correctness of the parallel decoder.
\begin{theorem}[Correctness of the decoder based on parallel sparse blossom]
\label{thm:parallel_correctness}
    Fix a detector graph $G=(V,E)$, edge weights $w$, and a nonempty set of observed detection events $D\subseteq V$.
    Then, there is an algorithm that constructs a hierarchical clustering of $D$ such that sparse blossom can be run independently and in parallel on the clusters at each level, and the resulting decoder produces the same decoding result as the original sparse-blossom decoder run globally.
\end{theorem}

The corresponding decoder consists of two procedures, as presented in Algorithms~\ref{alg:clustering} and \ref{alg:parallel_sparse_blossom} in Methods. Algorithm~\ref{alg:clustering} constructs hierarchical \textit{processing clusters} solely from the observed detection events, and Algorithm~\ref{alg:parallel_sparse_blossom} executes sparse blossom on these clusters while propagating stopped matching configurations between levels.
The overall decoding workflow, including the hierarchical clustering and subsequent parallel sparse-blossom execution, is illustrated schematically in Fig.~\ref{fig:parallel_sparse_blossom_overview}.
The explicit procedures, presented in a form suited to the analysis of their average parallel runtime, are given in Methods.

\paragraph*{Quasi-polylogarithmic average runtime of the decoder based on parallel sparse blossom}
For our parallel sparse-blossom decoder, we establish the following bound on the average parallel runtime of the proposed decoder.
To place the resulting scaling in context, polynomial and polylogarithmic functions in $d$ can be written respectively as
\begin{align}
    \poly (d)
    &= d^{O(1)} = \exp[O(\log d)], \\
    \polylog (d)
    &= (\log d)^{O(1)} = \exp[O(\log\log d)].
\end{align}
The runtime bound lying between these two scales is referred to as quasi-polylogarithmic.

We measure parallel runtime in a classical circuit-depth model.
Each elementary computational operation acts on $O(1)$ bits, and independent operations are executed concurrently.
We assume a sufficiently large yet polynomial amount of parallel computational resources so that the runtime is determined by the longest chain of dependent operations rather than by the total number of operations.
We do not impose geometric locality constraints on the classical computation and do not separately account for communication, synchronization, or memory-access overheads associated with realizing nonlocal operations.
We call the longest path of dependent operations in this circuit the critical execution path, whose length defines the parallel runtime.

We note that, for the rotated surface-code setting considered below, the required number of parallel processors is $O(pd^3)$ on average with physical error rate $p$ for an $O(d)$-round decoding window, since the number of processors required to fully exploit the available parallelism scales with the number of processing clusters, which in turn scales with the number of faults.

\begin{theorem}[Quasi-polylogarithmic average runtime of the decoder based on parallel sparse blossom]\label{thm:parallel_ave_time}
For the $[[n=d^2,1,d]]$ rotated surface code under the local stochastic error model with a physical error rate $p$ below a finite threshold $p_{\mathrm{th}}$, the average runtime of the parallel sparse-blossom decoder for $O(d)$ rounds of syndrome extraction is upper bounded by
\begin{align}
    \exp\bigl[O((\log\log d)(\log\log\log d))\bigr].
\end{align}
\end{theorem}

Theorem~\ref{thm:parallel_ave_time} shows that, under the parallel model above, the runtime to decode an $O(d)$-round window grows only quasi-polylogarithmically in the code distance.
In particular, for a decoding window consisting of $d$ syndrome-extraction rounds, the average parallel runtime per round is upper bounded by
\begin{align}
    \frac{\exp\bigl[O((\log\log d)(\log\log\log d))\bigr]}{d} 
    &=d^{-1+o(1)} \nonumber\\
    &=o(1).
\end{align}

The proof combines a bound on the processing time of a level-$k$ cluster with a probabilistic bound on the occurrence of such clusters.
For $p<p_{\mathrm{th}}$, the probability of reaching level $k$ decreases rapidly with $k$, so only levels up to $k=O(\log\log d)$ contribute non-negligibly to the average runtime.
Recall that the edge-based cluster parameters satisfy $d_k^E,b_k^E=\exp[\Theta(k\log k)]$.
Through the edge-to-vertex conversion, the corresponding processing-cluster parameters inherit the same scaling, yielding a processing time of $\exp[O(k\log k)]$ at level $k$.
Substituting $k=O(\log\log d)$ gives the quasi-polylogarithmic average parallel-runtime bound
\begin{align}
    \exp\bigl[O((\log\log d)(\log\log\log d))\bigr].
\end{align}
The construction of processing clusters from active detectors obeys the same average-runtime bound because the probability of reaching high clustering levels likewise decreases rapidly with the level (see Methods for details).

\paragraph*{Decreasing parallel event counts in a low-error regime}

The error-clustering framework developed above establishes the correctness and asymptotic parallel runtime of our parallel sparse-blossom algorithm.
To assess whether this framework also yields a meaningful finite-size advantage, we numerically evaluate parallel sparse blossom at code distances up to $d=49$.
The closed-form parameter schedule used in the asymptotic analysis is deliberately conservative and produces clusters that are too large to reveal the benefits of clustering within the accessible range of code distances.
We therefore introduce a less conservative schedule that generates more realistic parameters satisfying the required clustering conditions.
We define the event count as the number of events that occur during the execution of sparse blossom.
To quantify the resulting parallelism, we define the parallel event count as the number of events along the critical execution path and use it as a proxy for the parallel runtime, comparing it with the event count of global sparse blossom (see Methods for details).
Under the classical circuit-depth model used in this work, the parallel runtime scales with the parallel event count.
Consequently, the parallel event count provides a direct finite-size proxy for the parallel runtime.

\begin{figure*}[t]
    \centering
    \includegraphics[width=0.95\textwidth]{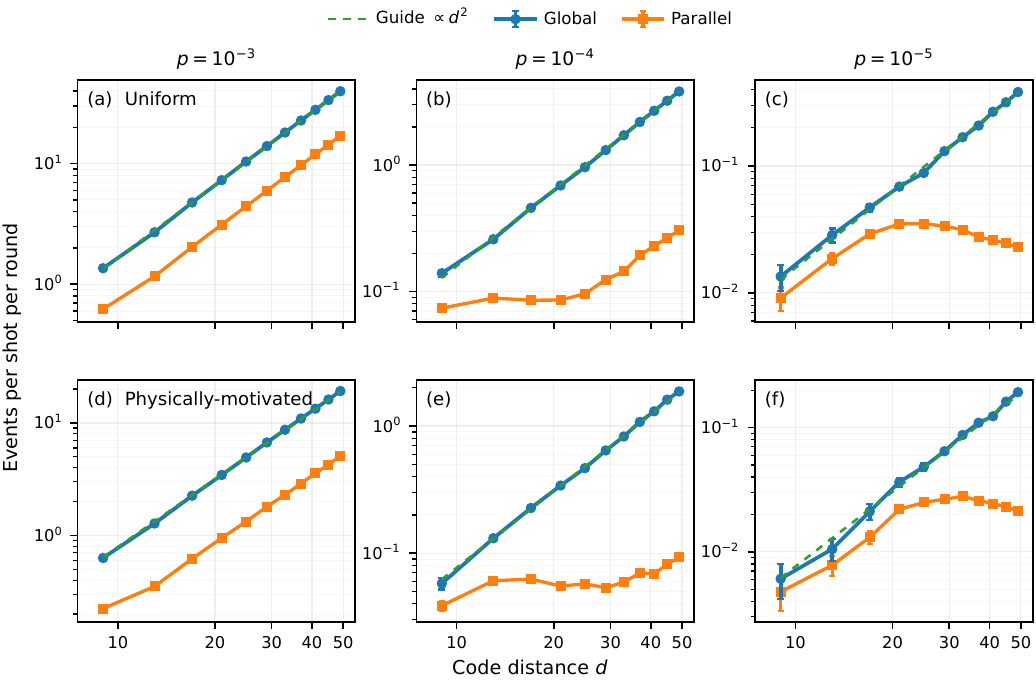}
    \captionsetup{justification=raggedright,singlelinecheck=false}
    \caption{
        Event counts per shot per syndrome-extraction round for global sparse blossom and parallel sparse blossom in rotated surface-code memory experiments as a function of the code distance $d$.
        We consider code distances $d\in\{9,13,\ldots,49\}$.
        The upper three panels correspond to physical error rates (a) $p=10^{-3}$, (b) $p=10^{-4}$, and (c) $p=10^{-5}$ under the uniform noise model.
        The lower three panels correspond to physical error rates (d) $p=10^{-3}$, (e) $p=10^{-4}$, and (f) $p=10^{-5}$ under the physically motivated noise model.
        The blue lines denote the event counts of global sparse blossom.
        The orange lines denote the event counts of parallel sparse blossom using the fixed parameter schedules.
        The dotted lines are guides proportional to $d^2$.
        Error bars indicate standard errors of the mean.
        Under both noise models, the parallel event counts per shot per syndrome-extraction round decrease at larger code distances for $p=10^{-5}$, indicating that a larger code distance increases the total amount of work but does not necessarily increase the critical-path length.
        For $p=10^{-4}$, the $d$ dependence of the parallel event count is weaker under the physically motivated noise model than under the uniform noise model, indicating that the effective noise strength relevant to the decoder depends on the details of the noise model even at the same nominal physical error rate $p$.
    }
    \label{fig:fixed_vs_global}
\end{figure*}

To evaluate this finite-size behavior, we performed numerical simulations of a quantum-memory experiment using the $[[n=d^2,k=1,d]]$ two-dimensional rotated surface code under two circuit-level independent and identically distributed (i.i.d.) depolarizing noise models with different error-parameter assignments: a uniform noise model and a physically motivated noise model with operation-dependent error parameters~\cite{sahay2026foldtransversalsurfacecodecultivation} (see Methods for details).
For each noise model and for each combination of physical error rate $p\in\{10^{-3}, 10^{-4}, 10^{-5}\}$ and code distance $d\in\{9,13,\ldots,49\}$, we generated a batch of 256 shots and decoded each shot using both global sparse blossom and parallel sparse blossom.
Figure~\ref{fig:fixed_vs_global} compares the parallel event count with the event count of global sparse blossom.
For all simulated shots, parallel sparse blossom produced the same decoding result as global sparse blossom, consistent with the correctness guarantee of Theorem~\ref{thm:parallel_correctness}.

For global sparse blossom, the event count increases with the code distance for all physical error rates considered, and the data closely follow the guide proportional to $d^2$ over the investigated range, which is consistent with the existing numerical evaluation in Fig.~10 of Ref.~\cite{Higgott2025sparseblossom} and our theoretical analysis in Sec.~\ref{sup:nonparallel_average_runtime}.
At $p=10^{-3}$, the parallel event counts increase with $d$ under both noise models over the entire range of code distances considered.
In contrast, for $p\le 10^{-4}$, the dependence of the parallel event count on $d$ becomes substantially weaker.
At $p=10^{-4}$, the parallel event count exhibits a non-monotonic dependence on $d$, first increasing, then decreasing slightly, and eventually increasing again.
This dependence is weaker under the physically motivated noise model than under the uniform noise model.
At $p=10^{-5}$, a decreasing regime appears at larger code distances within the investigated range under both noise models.

We interpret the non-monotonic behavior and the decreasing regime as finite-size effects arising from the hierarchical clustering structure.
We first consider how the clustering behavior depends on the physical error rate $p$.
At sufficiently low $p$, errors are sparse enough that active detectors can be successively removed by spatially localized processing clusters over multiple hierarchy levels.
In contrast, when the errors are not sufficiently sparse, although some active detectors can still be removed at lower levels, a substantial residual set can persist to higher levels.
Because the clustering diameter and the required error-free buffer increase monotonically with the hierarchy level, the remaining active detectors are eventually absorbed into a single graph-spanning cluster once the clustering scale exceeds the detector-graph size.
The workload of this cluster then grows with the number of residual active detectors and can therefore increase with $d$.
For $p<p_{\mathrm{th}}$, our error-clustering analysis guarantees that the probability of reaching high hierarchy levels is rapidly suppressed, preventing such large clusters from dominating the average asymptotic runtime.

For a fixed hierarchy level, the clustering parameters are independent of the code distance $d$.
When the processing-cluster size is comparable to or larger than the detector-graph size, increasing $d$ can increase the number of active detectors contained in the cluster and hence its sparse-blossom event count.
Once the detector graph becomes sufficiently larger than the processing cluster, however, the number of active detectors contained in the cluster, and hence its event count, approximately saturates with increasing $d$.
Normalization by the $d$ rounds of syndrome extraction then produces a decreasing event count per round.

This picture provides a possible interpretation of the three regimes observed in Fig.~\ref{fig:fixed_vs_global}.
For $p=10^{-3}$, the increasing event count is consistent with a regime in which a large, graph-spanning cluster contains an increasing number of residual active detectors as $d$ increases, resulting in behavior similar to that of global sparse blossom.
For $p=10^{-4}$, the observed non-monotonic behavior can be interpreted as the saturation of the event count associated with one hierarchy level, followed by a renewed increase as the next higher-level cluster becomes relevant.
For $p=10^{-5}$, the decreasing behavior over the investigated range is consistent with the event counts of the currently relevant clusters having approximately saturated, while the next higher-level cluster remains sufficiently unlikely to dominate the critical path within this range.

Based on this numerical evidence, the advantage of parallel sparse blossom becomes more pronounced in a low-error regime and at large code distances.
The weak dependence of the parallel event count on $d$, including the decreasing regimes observed for $p=10^{-5}$, suggests that increasing the decoding problem size by increasing the code distance does not necessarily increase the length of the critical execution path, even though the total amount of decoding work becomes larger.
Overall, the numerical results indicate that a low-error, large-distance regime is particularly favorable for our clustering-based parallel execution.

Note that the regime of $p$ in which our parallelization is effective depends strongly on the details of the noise model.
In particular, at $p=10^{-4}$, the physically motivated noise model exhibits a weaker dependence of the parallel event count on $d$ than the uniform noise model, suggesting that more realistic noise characteristics may extend the practically favorable regime of the proposed parallelization.
Further optimization of the parallelization scheme and a systematic investigation of its dependence on the noise model are interesting directions for future work; however, our present contribution is to demonstrate that such a parallelization strategy can be fundamentally effective for matching-based decoding.

\section*{Discussion}
In this work, we developed a parallel sparse-blossom decoder formulated at the level of detector-graph dynamics.
By relating visible processing clusters to well-separated clusters in the underlying faulty-edge configuration, we prove that its hierarchical parallel execution reproduces the same decoding result as global sparse blossom.
Using the rotated surface code as a representative setting, we prove that, for physical error rates below a finite threshold and code distance $d$, the average parallel runtime for $O(d)$ rounds of syndrome extraction is upper bounded by a quasi-polylogarithmic function of $d$.
Our finite-size numerical simulations further suggest that, in a low-error regime, increasing the code distance need not lead to longer decoding times under our parallelization, even though the computational runtime of the original non-parallel sparse blossom increases.
These numerical results also suggest that the favorable error-rate regime depends on the underlying noise model, indicating that the regime in which the parallelization is effective depends on the detailed distribution of physical errors.
Our asymptotic runtime guarantees have direct implications for the throughput and latency requirements of scalable FTQC.

Sufficiently high decoding throughput is required to prevent a growing backlog of syndrome data~\cite{terhal2015quantum}, whereas low decoding latency is necessary for logical branching in adaptive operations, such as non-Clifford gates implemented through gate teleportation~\cite{Skoric2023,Caune2026}.
Under our parallel computation model, we show that our decoder can achieve an average decoding time of $o(1)$ per syndrome-extraction round and a total decoding latency of $o(d)$ for a window containing $d$ rounds of syndrome extraction.
The former implies that the per-round decoding time need not increase with the code distance, provided sufficient parallel computational resources are available when sparse-blossom-based MWPM decoding is used within sliding-window decoding.
The latter implies that sparse-blossom decoding can become asymptotically faster than the $O(d)$ timescale required to acquire the syndrome data in the window, suggesting that decoding need not become the dominant source of latency in decoding-dependent logical operations.

Now we compare our result with the $O(1)$-average
parallel-time result of Ref.~\cite{fowler2013minimum}.
Reference~\cite{fowler2013minimum} argues that MWPM decoding can be performed with $O(1)$ average parallel processing time per incoming syndrome-extraction round in an online setting with continuously arriving syndrome data.
However, Ref.~\cite{fowler2013minimum} does not establish a formal correctness guarantee for the finite-window decoding problem considered here.
In contrast, we consider decoding over a finite spacetime window of $O(d)$ syndrome-extraction rounds and prove that our parallel decoding algorithm produces exactly the same decoding result as global sparse blossom, which provably guarantees error suppression~\cite{dennis2002topological,fowler2012proof}.
We further prove that the average parallel runtime for such an $O(d)$-round window is quasi-polylogarithmic in $d$, which implies an average parallel runtime per round of $o(1)$ rather than $O(1)$ for a window consisting of $d$ rounds.

A limitation of the present work is that our numerical evaluation concerns a parallel event count rather than end-to-end wall-clock decoding time.
Although the event count captures the critical execution path of sparse-blossom dynamics under our idealized parallel model, a practical parallel implementation must additionally account for scheduling, synchronization, communication, memory access, and processing-cluster construction.
Especially, developing efficient strategies for assigning processing clusters to a finite number of processors and scheduling their execution, analyzing the resulting runtime-resource tradeoff, and benchmarking throughput and latency on realistic computing architectures are important directions for future work.
Nevertheless, the present work provides the missing algorithmic and theoretical foundation for investigating whether this parallelism can be translated into end-to-end decoding advantages on realistic computing architectures.

Finally, although the concrete asymptotic and numerical results in this work are established for the rotated surface code, the parallelization framework itself is formulated in terms of sparse-blossom dynamics on a detector graph.
This suggests a possible extension to other quantum error-correcting codes that admit matching-based decoding.
The essential question is therefore whether the detector-graph structure and fault statistics of a given code permit the localization needed for independent sparse-blossom execution.
Establishing code-specific clustering bounds and parallel-runtime guarantees for these settings is an interesting direction for future work.

\section*{Methods}
We first recall the results of the edge-based error-clustering framework required for the runtime analysis. 
We then formalize the hierarchical execution of parallel sparse blossom on processing clusters and prove its equivalence to global sparse blossom. 
Next, we establish the quasi-polylogarithmic average-runtime bound by combining the processing time at each cluster level with the doubly exponentially decaying probability of reaching that level. 
Finally, we introduce the finite-size parameter schedule and the parallel event count used in the numerical evaluation.

\paragraph*{Error-clustering framework}
We use the edge-based error-clustering framework of Ref.~\cite{yoshida2026prooffinitethresholdunionfind}, with modifications described below.
Here, we recall only the notation and results needed for the runtime analysis; the full statements and proofs are given in Supplementary Information~\ref{sup:error_clustering}.

Under the circuit-level local stochastic error model, errors at different locations may activate the same pair of vertices associated with a given edge.
Therefore, the probability that a given detector-graph edge is faulty is bounded by a constant multiple of the physical error rate $p$.
That is, there exists an integer constant $\xi\geq 1$ such that
\begin{align}
    \tilde p \leq \xi p.
\end{align}
This induced edge-error model is the one used in the error-clustering analysis below, following the framework of Ref.~\cite{yoshida2026prooffinitethresholdunionfind}.

Let $N\subseteq E$ denote the set of faulty edges in the detector graph $G=(V,E)$.
The hierarchy of level-$k$ error clusters is defined as follows (see Definition~\ref{def:level-k_cluster_for_edges} for details).
For chosen increasing parameter sequences $\{d_k^E\}_{k\geq 1}$ and $\{b_k^E\}_{k\geq 1}$, we define the residual edge sets $N_k$ inductively.
We set $N_1=N$, and for $k\geq 1$, $N_{k+1}$ is obtained from $N_k$ by removing all edges contained in $(d_k^E,b_k^E)_E$-clusters in $N_k$.
A $(d_k^E,b_k^E)_E$-cluster is a subset of $N_k$ with diameter at most $d_k^E$ and separated from the remaining edges in $N_k$ by a buffer of width greater than $b_k^E$.
Our indexing differs from that of Ref.~\cite{yoshida2026prooffinitethresholdunionfind} in that the level index starts from $k=1$.

We apply the error-clustering threshold theorem of Ref.~\cite{yoshida2026prooffinitethresholdunionfind} using the parameter sequences
\begin{align}
    b_k^E=\beta \lambda^{f(k+1)}+w_{\max},\quad
    d_k^E=\gamma \lambda^{f(k)}-w_{\max},\label{eq:b_k_d_k}
\end{align}
where $\lambda>1$ is a constant, $\beta,\gamma>0$ scale linearly with $w_{\max} := \max\{w(e) \mid e \in E \}$, and $f:\mathbb{Z}_{>0}\to\mathbb{R}_{\geq0}$ is a monotonically increasing function satisfying
\begin{align}
\sum_{j=0}^{k-2} f(k-j)2^j \le c\,2^{k-1}.\label{eq:f_k}
\end{align}
In the present work, we use the locality condition
\begin{align}
    |B_e(r)| \le A(1+\Lambda r)^\Delta,
\end{align}
which is a slightly modified form of the locality condition used in Ref.~\cite{yoshida2026prooffinitethresholdunionfind}. 
With this locality condition, the same clustering argument gives the following consequence. 
For $p<p_{\mathrm{th}}$, the probability $p_k$ that a fixed edge belongs to a level-$k$ error cluster satisfies
\begin{align}
    p_k
    =
    O\!\left((p/p_{\mathrm{th}})^{2^{k-1}}\right),\label{eq:doubly_exponential_bound}
\end{align}
where $p_{\mathrm{th}}$ is given by
\begin{align}
    p_{\mathrm{th}}
    &=
    \frac{1}{
    \xi A \lambda^{c\Delta}
    \left(
    \lambda^{-f(1)}
    +
    \Lambda
    \left(
    \beta+\dfrac{\gamma+w_{\max}\lambda^{-f(1)}}{2}
    \right)
    \right)^{\Delta}
    }.
    \label{eq:threshold}
\end{align}
The proof proceeds as in Ref.~\cite{yoshida2026prooffinitethresholdunionfind}, with the above locality condition substituted for the original one and the constants modified accordingly.

\paragraph*{Hierarchical execution.}
We first define the hierarchy of clusters constructed solely from observed detection events, which we call \textit{processing clusters}, and the hierarchical execution rule.

\begin{definition}[Hierarchy of processing clusters]
\label{def:processing_cluster_hierarchy}
Let $D \subseteq V$ be the set of detection events.
For each level $k$, let
\begin{align}
\mathcal{C}_k := \{ C_{k,i} \}_{i=1}^{I_k}
\end{align}
be the collection of level-$k$ processing clusters, where $I_k$ is the number of such clusters.
Each processing cluster $C_{k,i}$ is required to satisfy the parity condition
\begin{align}
C_{k,i} \cap V_{\mathrm{bd}} \neq \varnothing
\quad \text{or} \quad
|C_{k,i} \cap D| \equiv 0 \pmod{2},
\end{align}
where $V_{\mathrm{bd}}\subseteq V$ denotes the set of boundary vertices.
The index $i$ labels the processing clusters at level $k$.
\end{definition}

\begin{theorem}[Construction of the processing-cluster hierarchy]
Algorithm~\ref{alg:clustering} constructs the hierarchy of processing clusters defined above.
In particular, every level-$k$ processing cluster satisfies the $(d_k^V,b_k^V)_V$-cluster condition, which is the vertex analog of the $(d_k^E,b_k^E)_E$-cluster condition (see Definition~\ref{def:level-k_cluster_for_vertices} for details).
\end{theorem}

\begin{proof}
We prove the statement by induction on the level $k$.
At each level $k$, the connected-component construction ensures that every accepted component is separated from the remaining residual detection events by more than $b_k^V$, while the acceptance condition guarantees that its diameter is at most $d_k^V$ and that it satisfies the parity condition.
Thus, every accepted component is a valid level-$k$ processing cluster.
The remaining detection events form the residual set for level $k+1$, to which the same construction is applied.
Therefore, Algorithm~\ref{alg:clustering} constructs the hierarchy of processing clusters defined in Definition~\ref{def:processing_cluster_hierarchy}.
\end{proof}

\begin{figure}[t]
\begin{algorithm}[H]
\caption{Clustering}
\label{alg:clustering}
\begin{algorithmic}[1]
\REQUIRE Detector graph, nonempty detection-event set, parameters for level-$k$ processing clusters
\ENSURE Processing clusters organized by level

\STATE $k \leftarrow 1$
\STATE Let the residual set be the input detection-event set.
\WHILE{residual set is non-empty}
    \STATE Grow an exploring region of radius $b_k^V/2$ around each residual detection event in parallel.
    \STATE Construct the connected components induced by mutually touching exploring regions.
    \FORALL{connected components in parallel}
        \IF{the detection-event set in the component has diameter at most $d_k^V$ and satisfies the parity condition}
            \STATE Declare the component to be a level-$k$ processing cluster.
            \STATE Remove the detection events in this component from the residual set.
        \ENDIF
    \ENDFOR
    \STATE $k \leftarrow k+1$
\ENDWHILE
\RETURN The processing clusters organized by level.
\end{algorithmic}
\end{algorithm}
\end{figure}

Given the hierarchy of processing clusters constructed by Algorithm~\ref{alg:clustering} according to Definition~\ref{def:processing_cluster_hierarchy}, the hierarchical execution of parallel sparse blossom proceeds as follows.
Algorithm~\ref{alg:parallel_sparse_blossom} summarizes the complete decoding procedure implementing this hierarchical execution.
The decoding output is represented by a binary vector specifying the predicted flips of the logical observables, which we refer to as the predicted observable flip vector.

\begin{enumerate}
    \item \textbf{Parallel execution within processing clusters.}
    All active processing clusters at levels above the current stopping level are executed in parallel.
    Matching configurations produced by stopped lower-level processing clusters are retained and inherited through higher levels.
    \item \textbf{Inheritance of stopped matching configurations.}
    A matching configuration produced by a stopped processing cluster is referred to as an \emph{inherited matching configuration} at subsequent levels.
    During the execution of a processing cluster, an inherited matching configuration is said to be \emph{touched} if an event in the sparse-blossom execution modifies the intermediate state associated with any active detector belonging to that configuration.
    For each inherited matching configuration, its subsequent execution is determined solely by whether it is touched by a processing-cluster execution.
    \begin{itemize}
        \item If an inherited matching configuration is touched by a processing cluster, the sparse-blossom dynamics of that processing cluster continue from and modify the corresponding intermediate state.
        When the processing cluster stops, the active detectors belonging to the touched inherited configuration are included in the stopped matching configuration produced by that processing cluster.
        This resulting matching configuration is then retained through subsequent higher levels.
        \item If an inherited matching configuration is not touched, it remains unchanged and continues to be retained.
    \end{itemize}
    Thus, when a processing cluster stops, its matching configuration consists of the state resulting from the sparse-blossom dynamics involving both its own active detectors and those belonging to all inherited matching configurations touched during its execution.
    Inherited matching configurations that remain untouched are not included in this stopped configuration and are instead inherited separately without modification.
    \item \textbf{Terminal configurations.}
    A stopped matching configuration is called a \emph{terminal matching configuration} if it is not touched by any processing-cluster execution after it is produced.
    After all processing clusters have stopped, the terminal matching configurations are combined to obtain the final matching result.
\end{enumerate}

We refer to these procedures collectively as the \textit{hierarchical execution rule}.

\begin{figure}[t]
\begin{algorithm}[H]
\caption{Parallel sparse blossom}
\label{alg:parallel_sparse_blossom}
\begin{algorithmic}[1]
\REQUIRE Detector graph, nonempty detection-event set, and parameters for level-$k$ processing clusters
\ENSURE Predicted observable flip vector

\STATE Construct processing clusters by Algorithm~\ref{alg:clustering}.
\STATE $k_{\max} \leftarrow \max\{k \mid {\text{level-}}k\text{ clusters exist}\}$
\STATE $k_{\mathrm{stop}} \leftarrow 0$
\WHILE{$k_{\mathrm{stop}}<k_{\max}$}
    \STATE Execute all active processing clusters at levels $k>k_{\mathrm{stop}}$ in parallel.
    \IF{all level-$(k_{\mathrm{stop}}+1)$ processing clusters have stopped}
        \STATE Update the retained matching configurations according to the hierarchical execution rule.
        \STATE $k_{\mathrm{stop}}\leftarrow k_{\mathrm{stop}}+1$
    \ENDIF
\ENDWHILE
\STATE Combine all terminal matching configurations.
\STATE Compute the predicted observable flip vector.
\end{algorithmic}
\end{algorithm}
\end{figure}

\paragraph*{Correctness of the parallel sparse-blossom algorithm.}
We prove that the hierarchy of processing clusters constructed solely from the observed detection events according to the $(d_k^V,b_k^V)_V$-cluster condition yields the same decoding output as global sparse blossom.

\noindent
\textbf{Informal statement 1: Correctness of error-cluster-based execution.}
The parallel execution based on the error-cluster hierarchy produces the same decoding result as global sparse blossom.

\noindent
\textit{Proof sketch.}
The argument has two steps.

First, we establish the stopping property of extended SB clusters by induction on the cluster level $k$ under appropriate conditions on the clustering parameters.
At $k=1$, the diameter bound ensures that a persistent growing region explores all active detectors associated with the level-$1$ error cluster within a bounded time. The parity condition then makes the extended SB cluster stable, while the surrounding error-free buffer is sufficiently large to prevent it from reaching any distinct extended SB cluster of the same or a higher level before becoming stable.
For the inductive step, fix $k\geq2$ and assume that every level-$k'$ extended SB cluster with $k'<k$ becomes stable before interacting with any distinct extended SB cluster of the same or a higher level.
Although a level-$k$ extended SB cluster can effectively spread faster by merging with such lower-level extended SB clusters, the induction hypothesis limits how much of the separating buffer can be effectively traversed through them.
The parameter conditions then guarantee that a sufficiently large unexplored region remains between the level-$k$ extended SB cluster and any distinct extended SB cluster of the same or a higher level.
This region cannot be traversed before the level-$k$ extended SB cluster becomes stable, and hence the cluster stops before interacting with any distinct extended SB cluster of the same or a higher level.

Second, using this stopping property, we show by induction on $k$ that the parallel execution based on the error-cluster hierarchy produces the same decoding result as global sparse blossom.
At $k=1$, the stopping property ensures that the sparse-blossom dynamics associated with each extended SB cluster remain independent until the cluster stops, so its stopped matching configuration coincides with the corresponding configuration in the global execution.
For the inductive step, fix $k\geq2$ and assume that the sparse-blossom dynamics associated with every level-$k'$ extended SB cluster with $k'<k$ agree with those in the global execution.
The stopping property excludes interactions with any distinct extended SB cluster of the same or a higher level before the level-$k$ cluster stops, while the induction hypothesis guarantees that all possible interactions with lower-level extended SB clusters reproduce the corresponding dynamics in the global execution.
Thus, the dynamics associated with each level-$k$ extended SB cluster coincide with the corresponding part of the global sparse-blossom execution until the cluster stops.
Repeating this argument over all levels yields the same final decoding result.

See Supplementary Information~\ref{sup:equivalence_from_error_cluster} for the formal statements and proofs.

\noindent
\textbf{Informal statement 2: Reduction to algorithmic processing clusters.}
The hierarchical execution based on the processing clusters constructed from the observed detection events produces the same decoding result as the execution based on the underlying error-cluster hierarchy.

\noindent
\textit{Proof sketch.}
The argument has three steps.

First, we establish a parameter conversion under which the active detectors associated with each error cluster form a corresponding processing cluster.
By converting the edge-based distances into vertex-based distances, we show that this cluster satisfies a processing-cluster condition induced by the corresponding error-cluster condition.
We call this cluster the \textit{associated processing cluster}.

Second, we establish the relationship between the associated processing clusters and the processing clusters sequentially constructed from the observed detection events under the same parameter conditions.
Specifically, we show by induction on $k$ that the active detectors in each level-$k$ associated processing cluster are partitioned into processing clusters of levels at most $k$.
At $k=1$, the residual active-detector sets in the two constructions coincide. 
The diameter and buffer conditions ensure that the active detectors in each associated processing cluster form an isolated connected component in the processing-cluster construction, and the parity condition ensures that this component is accepted as a level-$1$ processing cluster. 
Hence, each level-$1$ associated processing cluster and a corresponding level-$1$ processing cluster contain exactly the same active detectors. 
Note that additional level-$1$ processing clusters may also be formed from active detectors that do not belong to any level-$1$ associated processing cluster.
For the inductive step, fix a level $k\geq2$ and assume that, for every level-$k'$ associated processing cluster with $k'<k$, its active detectors are partitioned into processing clusters of levels at most $k'$.
Consequently, some active detectors in a level-$k$ associated processing cluster may already have been removed by lower-level processing clusters before level $k$ is reached.
The buffer condition ensures that no such lower-level processing cluster can contain active detectors belonging to two distinct level-$k$ associated processing clusters, or to both a level-$k$ associated processing cluster and the remaining residual set.
Therefore, the active detectors removed at lower levels can be assigned to individual level-$k$ associated processing clusters, while the remaining active detectors in each associated processing cluster form a level-$k$ processing cluster.
Hence, the active detectors in every level-$k$ associated processing cluster are partitioned into processing clusters of levels at most $k$.

Finally, we prove that the decoding result obtained by executing sparse blossom directly on each associated processing cluster coincides with that obtained by the hierarchical parallel execution of its constituent processing clusters.
We first show that the hierarchy of processing clusters satisfies an analogous stopping property to that established for the error-cluster hierarchy.
This stopping property ensures that each constituent processing cluster can be executed independently until it stops, without interacting with any processing cluster of the same or a higher level.
When a constituent processing cluster stops, its matching configuration is retained through subsequent higher levels and remains unchanged until it is touched by a higher-level processing-cluster execution.
If it is touched, the corresponding sparse-blossom dynamics continue from that stopped configuration, and the active detectors belonging to the touched configuration participate in the resulting higher-level dynamics.
Thus, the higher-level execution continues from exactly the same intermediate state as the direct execution on the associated processing cluster.
Consequently, the hierarchical execution differs from the direct execution only in the order in which mutually independent parts of the sparse-blossom dynamics are executed, and therefore produces the same final matching result.
Applying this argument to all associated processing clusters shows that the hierarchical execution of the full processing-cluster hierarchy produces the same decoding result as the execution based on the full hierarchy of associated processing clusters.

See Supplementary Information~\ref{sup:reduction_to_processing_cluster} for the formal statement and proof.

\noindent
\textbf{Proof of theorem~\ref{thm:parallel_correctness}.}
Combining the two statements above, the execution on the algorithmically constructed processing-cluster hierarchy is equivalent to the execution on the hierarchy associated with the underlying faulty edges, which in turn is equivalent to global sparse blossom. 
This proves Theorem~\ref{thm:parallel_correctness}.
See Supplementary Information~\ref{sup:correctness_of_theorem_1} for the formal statement and proof.

\paragraph*{Quasi-polylogarithmic average runtime.}
We now prove the quasi-polylogarithmic average-runtime bound of Theorem~\ref{thm:parallel_ave_time} for the rotated surface code.
Throughout this analysis, parallel runtime is measured in the classical circuit-depth model defined in the main text.
We use the parameters $\{d_k^E,b_k^E\}_{k\geq 1}$ for error clusters and $\{d_k^V,b_k^V\}_{k\geq 1}$ for processing clusters.

\noindent
\textbf{Proof of theorem~\ref{thm:parallel_ave_time}.}
We first choose the parameters for the edge-based error-clustering framework.
To satisfy both the error-clustering bound and the stopping guarantee, we set
\begin{align}
b_k^{E}
&:=
\beta\lambda^{(k+1)\log(k+1)}+w_{\max},\\
d_k^{E}
&:=
\gamma\lambda^{k\log k}-w_{\max}.
\end{align}
This choice corresponds to
\begin{align}
f(k)=k\log k
\end{align}
in Eq.~\eqref{eq:b_k_d_k}, which satisfies Eq.~\eqref{eq:f_k} for some finite constant $c$.
We then choose the ratios $\beta/w_{\max}$ and $\gamma/w_{\max}$, together with the constant $\lambda$, so that the conditions required for the error-clustering bound and the stopping guarantee are satisfied simultaneously.
See Supplementary Information~\ref{sup:parameter_verification} for the formal verification of this parameter choice.

The corresponding vertex-based parameters for processing clusters have the same asymptotic scaling,
\begin{align}
d_k^{V},b_k^{V}=\exp[\Theta(k\log k)].
\end{align}
Under this parameter choice, the fraction of the separating distance that remains unoccupied by lower-level extended SB clusters is bounded below by a positive constant.
Consequently, the sparse-blossom dynamics associated with a level-$k$ processing cluster can extend only over a region of linear scale $O(d_k^V)$ before stopping.
By the locality of the detector graph, the runtime of processing such a cluster is therefore upper bounded by
\begin{align}
O\!\left(\poly(d_k^V)\right)=\exp[O(k\log k)].
\end{align}

While the computational cost of a level-$k$ processing cluster grows super-polynomially with $k$, the probability that a fixed edge remains in the residual error set up to level $k$ decreases doubly exponentially with $k$, as shown in Eq.~\eqref{eq:doubly_exponential_bound}, for a physical error rate $p<p_{\mathrm{th}}$ with $p_{\mathrm{th}}$ given in Eq.~\eqref{eq:threshold}.
Since the decoding window for a distance-$d$ code contains $O(d^3)$ edges, the probability that the clustering reaches level $k$ is bounded by an $O(d^3)$ factor times this doubly exponentially decaying probability.
The crossover level at which this probability becomes smaller than order one therefore satisfies
\begin{align}
k=O(\log\log d).
\end{align}
Up to this level, the largest per-level processing cost is
\begin{align}
\exp[O(k\log k)]=\exp[O\!\left((\log\log d)(\log\log\log d)\right)].
\end{align}
Beyond this crossover level, the doubly exponential suppression of the probability dominates the increasing per-level processing cost, so the contribution from the tail obeys the same asymptotic bound.
The processing-cluster construction can be analyzed in the same way, because the probability that residual active detectors remain to level $k$ obeys the same doubly exponential suppression, while the additional parallel overhead for constructing the clusters is only polylogarithmic in $d$.
Therefore, the total average parallel decoding time is upper bounded by
\begin{align}
\exp[O\!\left((\log\log d)(\log\log\log d)\right)].
\end{align}
See Supplementary Information~\ref{sup:correctness_of_theorem_2} for the formal proof, including the control of the high-level tail and the runtime of the processing-cluster construction.

\paragraph*{Finite-size parameter schedule.}
The parameter schedule used above is chosen to establish the asymptotic runtime bound. 
However, its constants are too general and conservative for the finite code distances considered in our numerical experiments. 
The resulting parameters tend to group the active detectors into only a few excessively large processing clusters, thereby obscuring the practical parallelism of the decoder.

Algorithm~\ref{alg:experimental_parameter_schedule} defines the resulting recursive parameter schedule.
In the algorithm, we use the following auxiliary quantities.
For fixed parameters $0<\phi_{\min}<1$ and $0<q<1$, we define the target sequence
\begin{align}
    \overline{\phi}_k
    :=
    \phi_{\min}
    +
    (1-\phi_{\min})q^{k-1}.\label{eq:phi_sequence_q}
\end{align}
For given $d_k^V$, $b>0$, and $\phi_k$, we also define
\begin{align}
    \Phi_{k+1}(b)
    :=
    \phi_k
    -
    \frac{(\phi_k+2)d_k^V}
    {(\phi_k+1)d_k^V+\phi_k b}.
\end{align}

\begin{figure}[t]
\begin{algorithm}[H]
\caption{Parameter schedule}
\label{alg:experimental_parameter_schedule}
\begin{algorithmic}[1]
\REQUIRE Initial diameter $d_1^V=w_{\max} + 1$, target floor $\phi_{\min}\in(0,1)$, decay factor $q\in(0,1)$
\ENSURE Parameter sequences $\{d_k^V\}_{k\ge1}$ and $\{b_k^V\}_{k\ge1}$
\STATE $\phi_1 \gets 1$
\FOR{$k=1,2, \ldots$}
    \STATE Compute $\overline{\phi}_{k+1}$
    \STATE Choose the smallest integer $b_k^V$ satisfying
    \begin{align}
        b_k^V &\ge d_k^V,\\
        \phi_k
        &>
        2\frac{d_k^V}{b_k^V},\\
        \Phi_{k+1}(b_k^V)
        &\ge
        \overline{\phi}_{k+1}.
    \end{align}
    \STATE $\phi_{k+1}\gets\Phi_{k+1}(b_k^V)$
    \STATE $d_{k+1}^V\gets3d_k^V+4b_k^V+2w_{\max}$
\ENDFOR
\RETURN $\{d_k^V\}_{k\ge1}$ and $\{b_k^V\}_{k\ge1}$
\end{algorithmic}
\end{algorithm}
\end{figure}

The average runtime of the resulting parallel decoder is
upper bounded by
\begin{align}
    \exp\left[O\left((\log\log d)^2\right)\right].
\end{align}
Thus, although this bound is weaker than that obtained using the closed-form parameter schedule, the average parallel runtime remains quasi-polylogarithmic in $d$.

The formal justification for the existence of such parameters, the verification that they satisfy the conditions required for the correctness guarantee, and the corresponding average-runtime analysis are given in Supplementary Information~\ref{sup:analysis_of_scheduled_parameters}.

\paragraph*{Numerical evaluation}
We next numerically evaluate the finite-size performance of parallel sparse blossom using the parameters generated by Algorithm~\ref{alg:experimental_parameter_schedule}.
We consider the $[[n=d^2,k=1,d]]$ two-dimensional rotated surface code~\cite{bombin2007optimal} under two circuit-level i.i.d. depolarizing noise models.
In the first model, which we call the uniform noise model, every circuit location has the same physical error rate $p$.
After preparation, an error that flips the prepared state occurs with probability $p$.
After each single-qubit gate, including an idle operation, each of the Pauli errors $X$, $Y$, and $Z$ occurs with a probability $p/3$.
After each two-qubit CNOT gate, each of the 15 nonidentity two-qubit Pauli errors occurs with probability $p/15$.
Before measurement, an error that flips the measurement outcome occurs with probability $p$.
The second is a physically motivated noise model from Ref.~\cite{sahay2026foldtransversalsurfacecodecultivation}, also parameterized by $p$, in which single-qubit gates have a lower error rate than two-qubit gates.
After preparation, an error that flips the prepared state occurs with probability $p$.
After each non-idle single-qubit gate, each of the Pauli errors $X$, $Y$, and $Z$ occurs with probability $p/30$, giving a total single-qubit-gate error probability of $p/10$.
Idle operations are assumed to be noiseless.
After each two-qubit CNOT gate, each of the 15 nonidentity two-qubit Pauli errors occurs with probability $p/15$.
Before measurement, an error that flips the measurement outcome occurs with probability $p$.
For both noise models, we simulate memory experiments initialized in the logical state $\ket{+}_L$.
Each experiment consists of $d$ rounds of syndrome extraction followed by destructive measurement of the data qubits in the $X$ basis.

Under this setting, we estimate the parallel runtime using the number of sparse-blossom events on the critical execution path.
We use Stim~\cite{gidney2021stim} to generate the syndrome-extraction circuits and the corresponding detector error models.
The observed active detectors are assigned to processing clusters using the parameter sequences generated by Algorithm~\ref{alg:experimental_parameter_schedule}.
We then use PyMatching~\cite{Higgott2025sparseblossom,pymatching2022} to execute sparse blossom on each processing cluster.
In the simulator, the processing clusters are executed sequentially in order to reproduce the inheritance of matching configurations.
After all processing clusters at a given level have stopped, their stopped matching configurations are inherited through the subsequent higher levels.
The parallel event count is then reconstructed from the stopping times and event histories by treating mutually independent cluster executions as concurrent.

For a processing cluster $C$ at level $k$, let $s(C)$ denote its stopping time in the algorithmic time of sparse blossom, and let $E_C(t)$ denote the number of sparse-blossom events generated by $C$ up to time $t$.
We define $E_C(t):=E_C(s(C))$ for $t\ge s(C)$.
Let $t_{<k}$ denote the critical-path stopping time accumulated through levels strictly below $k$, and let $E_{<k}$ denote the corresponding critical-path event count.
For $k=1$, we set
\begin{align}
    t_{<1}=E_{<1}=0.
\end{align}
For a level-$k$ processing cluster $C$, its critical-path stopping time is defined by
\begin{align}
    t(C)
    :=
    \max\left\{
        t_{<k},
        s(C)
    \right\}.
\end{align}
The corresponding critical-path event count is
\begin{align}
    E(C)
    &:=
    \max\left\{
        E_{<k},
        E_C\left(t_{<k}\right)
    \right\}
    \nonumber\\
    &\quad+
    E_C\left(t(C)\right)
    -
    E_C\left(t_{<k}\right).
    \label{eq:ideal_parallel_event_count}
\end{align}
After all level-$k$ processing clusters have stopped, we update
\begin{align}
    t_{<k+1}
    &:=
    \max_{C\in\mathcal{C}_k} t(C),
    \\
    E_{<k+1}
    &:=
    \max_{C\in\mathcal{C}_k} E(C).
\end{align}
If level $k$ contains no processing cluster, we set
\begin{align}
    t_{<k+1}:=t_{<k},
    \qquad
    E_{<k+1}:=E_{<k}.
\end{align}
Finally, if $k_{\max}$ is the highest level containing a processing cluster, the parallel event count for one decoding instance is
\begin{align}
    E_{\mathrm{ideal}}
    :=
    E_{<k_{\max}+1}.
\end{align}
For an instance with no detection event, we set $E_{\mathrm{ideal}}=0$.
We normalize $E_{\mathrm{ideal}}$ by the number of syndrome-extraction rounds.

Equation~\eqref{eq:ideal_parallel_event_count} accounts for the overlap between lower-level and level-$k$ executions.
Up to the critical-path stopping time $t_{<k}$ accumulated from the lower levels, the lower-level critical execution path and the execution of $C$ proceed concurrently.
We therefore retain only the larger of their event counts.
Events generated by $C$ after $t_{<k}$ are then added because they extend beyond the lower-level execution.

In the analytical runtime bound, time is measured in sparse-blossom algorithmic time, corresponding to the growth time of graph fill regions.
Sparse blossom, however, is implemented as an event-driven algorithm: the scheduler executes directly to the time of the next event and then processes the corresponding update.
We therefore use the number of events as a proxy for parallel runtime, under the idealized assumption that advancing to the next event and processing each event both require $O(1)$ time.
The theoretical stopping guarantee bounds the stopping time of a processing cluster using its diameter and margin.
In the numerical evaluation, we instead use the actual stopping time recorded during the sparse-blossom execution.
The measured stopping time is no larger than the theoretical upper bound, and no additional sparse-blossom event is generated after a cluster has stopped.
Therefore, extending the event history from the measured stopping time to the theoretical upper bound would add no events.
Using the measured stopping time consequently does not omit any event contributing to the critical execution path.
Under the parallel model defined above, the resulting event count is therefore a proxy for the parallel runtime rather than a measurement of wall-clock time.

For each noise model and each combination of physical error rate $p\in\{10^{-3},10^{-4},10^{-5}\}$ and code distance $d\in\{9,13,\ldots,49\}$, we generated a batch of $256$ shots using Stim and decoded each shot using both global sparse blossom and parallel sparse blossom.
For parallel sparse blossom, we set $q=0.1$ and $\phi_{\min}=0.01$ in Eq.~\eqref{eq:phi_sequence_q}.
This parameter choice showed favorable performance among the choices examined in our numerical experiments.
For each shot, we compared the parallel event count with the event count of global sparse blossom.
We also verified that the decoding result produced by parallel sparse blossom coincided with that of global sparse blossom for all simulated samples.
The resulting finite-size comparison is presented in Fig.~\ref{fig:fixed_vs_global}.

\section*{DATA AVAILABILITY}
The numerical data supporting the findings of this study, including the data underlying the figures, are available at \url{https://github.com/Hayata-Yamasaki-Group/parallel-sparse-blossom/tree/main/data}.

\section*{CODE AVAILABILITY}
The code used to implement the parallel sparse-blossom decoder, perform the numerical simulations, and generate the numerical results reported in this study is available at \url{https://github.com/Hayata-Yamasaki-Group/parallel-sparse-blossom}.
The numerical simulations use Stim and PyMatching, as described in Methods.

\bibliography{citation} 

\section*{ACKNOWLEDGEMENTS}
This work was supported by JST PRESTO Grant No. JPMJPR23FC,
JST CREST Grant No. JPMJCR25I5, JST Moonshot R\&D Grant No. JPMJMS256J, and Faculty Research Funding
from Google Quantum AI.

\section*{Author Contributions}
Both authors contributed to the conception of the work, the analysis and interpretation of the work, and the preparation of the manuscripts.

\section*{COMPETING INTERESTS}
The authors declare no competing interests.

\clearpage
\onecolumngrid
\section*{Supplementary Information}
\beginsupplement

The Supplemental Materials are organized as follows.
In Sec.~\ref{sup:sparse_blossom}, we review the sparse-blossom algorithm and its standard primal-dual formulation.
In Sec.~\ref{sup:metric}, we introduce the edge and vertex metrics and define the notions used for hierarchical clustering.
In Sec.~\ref{sup:property_of_error_cluster}, we establish basic properties of the error-clustering hierarchy.
In Sec.~\ref{sup:error_clustering}, we derive probabilistic bounds for the occurrence of high-level error clusters.
In Sec.~\ref{sup:equivalence_from_error_cluster}, we establish the localization and stopping properties of sparse-blossom dynamics induced by the error-cluster hierarchy.
In Sec.~\ref{sup:reduction_to_processing_cluster}, we relate the error-cluster hierarchy, which depends on the hidden faulty-edge configuration, to the processing-cluster hierarchy constructed solely from the observed detection events.
In Sec.~\ref{sup:correctness_of_theorem_1}, we combine these results to prove the correctness of the parallel sparse-blossom decoder.
In Sec.~\ref{sup:parameter_verification}, we verify that the asymptotic clustering parameters satisfy the conditions required for the correctness and error-clustering arguments.
In Sec.~\ref{sup:correctness_of_theorem_2}, we prove the quasi-polylogarithmic upper bound on the average runtime of the parallel sparse-blossom decoder.
In Sec.~\ref{sup:nonparallel_average_runtime}, we derive an average-runtime bound for non-parallel sparse blossom.
Finally, in Sec.~\ref{sup:analysis_of_scheduled_parameters}, we analyze the finite-size parameter schedule used in the numerical evaluation and verify its required correctness and runtime properties.

\paragraph*{Notation and setup.}
Throughout the Supplementary Information, let $G=(V,E)$ denote the weighted detector graph, let $V_{\mathrm{bd}}\subseteq V$ denote the set of boundary vertices, and let $D\subseteq V$ denote the observed detection events, or the set of active detectors.
Let $N\subseteq E$ denote the faulty-edge set that induces $D$.
For each edge $e\in E$, let $w(e)>0$ denote its weight, and define
\begin{align}
    w_{\min}:=\min_{e\in E}w(e),
    \qquad
    w_{\max}:=\max_{e\in E}w(e).
\end{align}
We assume that the induced faulty-edge model satisfies
\begin{align}
    \Pr[N\supseteq S]\le \tilde p^{|S|}
\end{align}
for every $S\subseteq E$.
Under the circuit-level local stochastic error model considered in the main text, $\tilde p$ satisfies
\begin{align}
    \tilde p\le \xi p
\end{align}
for some constant $\xi\ge1$, where $p$ is the physical error rate.

\section{Preliminaries on sparse blossom}\label{sup:sparse_blossom}
We briefly recall the standard primal-dual view of the blossom algorithm for MWPM.
The blossom algorithm can be formulated in terms of the following primal-dual pair on a weighted graph $G=(V,E)$.
\begin{itemize}
    \item The primal problem:
    \begin{align}
        \text{Minimize} \quad&\sum_{e \in E} w(e)\; x_e, \\
        \text{subject to} \quad&\sum_{e\in \delta(v)} x_e = 1 \quad ^\forall v \in V, \\
        &\sum_{e \in \delta(S)} x_e \ge 1 \quad ^\forall S \in \mathcal{O}, \\
        &x_e \ge 0 \quad ^\forall e \in E,
    \end{align}
    where
    \begin{align}
        \delta(S) &\coloneqq \left\{(u, v) \in E \mid u \in S, v \in V \setminus S \right\}, \\
        \mathcal{O} &\coloneqq \left\{o \subseteq V \mid |o| > 1, |o| \equiv 1 \;(\mathrm{mod} \; 2) \right\}.
    \end{align}
    \item The dual problem:
    \begin{align}
        \text{Maximize} \quad &\sum_{v \in V}y_v + \sum_{S \in \mathcal{O}}y_S, \\
        \text{subject to} \quad &\operatorname{slack}(e) \ge 0 \quad ^\forall e \in E, \\
        &y_S \ge 0 \quad ^\forall S \in \mathcal{O},
    \end{align}
    where
    \begin{align}
        \operatorname{slack}(e) \coloneqq w(e )- \sum_{u \in e}y_u - \sum_{S \in \mathcal{O}: \; e \in \delta(S)}y_S.
    \end{align}
\end{itemize}

The blossom algorithm solves the MWPM problem by updating the primal and dual variables.
The primal updates consist of four operations: GROW, AUGMENT, SHRINK, and EXPAND.
Those four operations are described as follows.
\begin{itemize}
    \item \textbf{GROW:} Add a matched pair of nodes to an alternating tree.
    \item \textbf{AUGMENT:} Augment the path between the roots of two alternating trees when they become connected.
    \item \textbf{SHRINK:} Form a blossom when an odd-length cycle is found.
    \item \textbf{EXPAND:} Dissolve a blossom and add the odd-length path through it to the alternating tree, while matching the nodes along the even-length path to their neighbors.
\end{itemize}

Sparse blossom solves the MWPM problem on the detector graph by representing the dual-variable updates geometrically through graph-fill regions.
Each node, either a regular node or a blossom, is associated with a graph-fill region in one of three growth states: growing $(+1)$, frozen $(0)$, or shrinking $(-1)$.
The objects maintained by sparse blossom are organized into three structures: matches, blossoms, and alternating trees.
A match consists of a pair of frozen regions, both with growth rate $0$.
A blossom is a region containing an odd-length cycle of regions, called a blossom cycle, formed when two growing regions in the same alternating tree collide.
The regions in a blossom cycle are called blossom children, and the blossom containing them is called their blossom parent.
A graph-fill region with no blossom parent is called active.
An alternating tree is a tree whose nodes correspond to active graph-fill regions.
It contains at least one growing region and may also contain shrinking regions.
Growing and shrinking regions correspond, respectively, to outer and inner nodes in the general blossom algorithm.

During a dual update, only the dual variables associated with nodes in alternating trees are modified:
\begin{align}
    y_u \leftarrow y_u \pm \delta_T,
\end{align}
where the plus sign applies to outer nodes and the minus sign to inner nodes.
In sparse blossom, the same update magnitude $\delta_T=\delta$ is used for every alternating tree $T$.
If the node is a blossom rather than a regular node, only the dual variable associated with the blossom is updated, while those of its constituent nodes remain unchanged.
In our implementation, we set $\delta=1$ by scaling all edge weights to even integers.

Because only nodes in alternating trees have changing dual variables, and hence changing graph-fill regions, the sparse-blossom dynamics are driven by events generated by the growth and shrinkage of these regions.
We classify these events as follows.
\begin{itemize}
\item[(a)] \textbf{Tree hits match:} The matched pair is added to the alternating tree.
\item[(b)] \textbf{Tree hits tree:} The two alternating trees become matched.
\item[(c)] \textbf{Tree hits itself:} The collision forms a blossom.
\item[(d)] \textbf{Blossom shrinks to zero:} The blossom shatters.
\item[(e)] \textbf{Node shrinks to zero:} The associated nodes form a blossom.
\item[(f)] \textbf{Tree hits boundary:} The tree becomes matched to the boundary.
\item[(g)] \textbf{Tree hits boundary match:} The tree becomes matched to the boundary match.
\end{itemize}

We refer to the component that manages these events according to the growth, shrinkage, and collisions of graph-fill regions in the detector graph as the \emph{flooder}.
Its operations are organized into four types of flooder events:
\begin{itemize}
    \item \textbf{ARRIVE:} A growing region enters an empty detector node.
    \item \textbf{LEAVE:} A shrinking region leaves an empty detector node.
    \item \textbf{COLLIDE:} A growing region collides with another region or with the boundary.
    \item \textbf{IMPLODE:} A shrinking region reaches a radius of zero.
\end{itemize}

\section{Metric and clustering definitions}\label{sup:metric}
Throughout this section, the distance between a nonempty set and the empty set is defined to be $+\infty$.

\begin{definition}[Distance in a line graph]
Let $G=(V, E)$ be a graph.
The edge distance $d_E(e,e')$ in the line graph $L(G)$ is defined as 
\begin{align}
        d_E(e, e') \coloneqq \min_{v \in e, v' \in e'} d_V(v, v') + \frac{w(e) + w(e')}{2}(1 - \delta_{e, e'}),
\end{align}
where $d_V(v, v')$ denotes the shortest path (distance) connecting $v, v'$ in $V$, and $\delta_{e, e'}$ denotes the Kronecker delta.
\end{definition}

\begin{proposition}[Metric property of the edge distance]
    The distance for edges $d_E(e, e')$ is a metric on $E$.
\end{proposition}

\begin{proof}
    We show that $d_E(e,e')$ satisfies the axioms of a metric.
    \begin{itemize}
        \item $d_E(e,e') = 0 \Leftrightarrow e=e'$,
        \item $d_E(e,e') = d_E(e',e)$,
        \item $d_E(e,e'') \le d_E(e,e') + d_E(e',e'')$.
    \end{itemize}

    First, we prove the identity of indiscernibles.
    Since $d_V(v,v') \ge 0$, $w(e)>0$, and
    $1-\delta_{e,e'} \ge 0$ hold for all $v, v' \in V$ and $e, e' \in E$, we have
    \begin{align}
        d_E(e,e')=0
        &\Leftrightarrow
        \min_{v\in e,\; v'\in e'} d_V(v,v')
        + \frac{w(e)+w(e')}{2}(1-\delta_{e,e'}) = 0 \\
        &\Leftrightarrow
        \min_{v\in e,\; v'\in e'} d_V(v,v') = 0
        \quad\land\quad
        \frac{w(e)+w(e')}{2}(1-\delta_{e,e'})=0 \\
        &\Leftrightarrow
        \delta_{e,e'}=1 \\
        &\Leftrightarrow
        e=e'.
    \end{align}
    Therefore,
    \begin{align}
        d_E(e,e')=0 \Leftrightarrow e=e'.
    \end{align}

    Next, we prove the symmetry.
    By the symmetry of $d_V$ and $\delta_{e,e'}$, we obtain
    \begin{align}
        d_E(e,e')
        &= \min_{v\in e,\; v'\in e'} d_V(v,v')
           + \frac{w(e)+w(e')}{2}(1-\delta_{e,e'}) \\
        &= \min_{v'\in e',\; v\in e} d_V(v',v)
           + \frac{w(e')+w(e)}{2}(1-\delta_{e',e}) \\
        &= d_E(e',e).
    \end{align}

    Finally, we prove the triangle inequality. 
    If $e=e'$ or $e'=e''$, then the claim is immediate from the identity. If $e=e''$, then
    \begin{align}
        d_E(e,e'') = 0 \le d_E(e,e') + d_E(e',e'').
    \end{align}
    Hence, it suffices to consider the case where $e,e',e''$ are pairwise distinct.
    Choose vertices $u\in e$, $v\in e'$, $x\in e'$, and $y\in e''$ such that
    \begin{align}
        d_E(e,e') &= d_V(u,v) + \frac{w(e)+w(e')}{2}, \\
        d_E(e',e'') &= d_V(x,y) + \frac{w(e')+w(e'')}{2}.
    \end{align}
    Since $v,x\in e'$, the two vertices $v$ and $x$ are joined by the edge $e'$.
    Therefore, we have
    \begin{align}
        d_V(v,x) \le w(e').
    \end{align}
    By the triangle inequality for $d_V$, we have
    \begin{align}
        \min_{p\in e,\; q\in e''} d_V(p,q)
        \le d_V(u,y)
        \le d_V(u,v)+d_V(v,x)+d_V(x,y)
        \le d_V(u,v)+w(e')+d_V(x,y).
    \end{align}
    Hence, we obtain
    \begin{align}
        d_E(e,e'')
        &= \min_{p\in e,\; q\in e''} d_V(p,q) + \frac{w(e)+w(e'')}{2} \\
        &\le d_V(u,v)+w(e')+d_V(x,y) + \frac{w(e)+w(e'')}{2} \\
        &= \left(d_V(u,v)+\frac{w(e)+w(e')}{2}\right)
         + \left(d_V(x,y)+\frac{w(e')+w(e'')}{2}\right) \\
        &= d_E(e,e') + d_E(e',e'').
    \end{align}
    Therefore, $d_E$ satisfies the triangle inequality as well. 
    
    Thus $d_E$ is a metric on $E$.
\end{proof}

\begin{definition}[Metric notions for vertices]
We define the auxiliary notions.
    \begin{itemize}
        \item For subsets $C,C' \subseteq V$, we define
        \begin{align}
        d_{C_V}(C,C') := \min_{v\in C,\ v'\in C'} d_V(v,v').
        \end{align}
        \item We define the ball of radius $r$ around a vertex $v$ by
        \begin{align}
        B_v(r) := \{v'\in V \mid d_V(v,v') \le r\}.
        \end{align}
        \item The diameter of a set of vertices $C$ is defined as
        \begin{align}
        \diam_{V}(C) := \max_{v,v'\in C} d_V(v,v').
        \end{align}
    \end{itemize}
\end{definition}

\begin{definition}[Isolation and clustering of a vertex set]
For $D \subseteq V$,
\begin{itemize}
    \item Two vertices $v,v'\in D$ are called $(r,R)_V$-isolated if
    $v'\notin B_v(R)\setminus B_v(r)$.
    Otherwise, they are called $(r,R)_V$-linked.
    \item A vertex $v\in D$ is $(r,R)_V$-isolated in $D$ if $v$ and $v'$ are $(r,R)_V$-isolated for all $v'\in D$.
    \item A subset $C\subseteq D$ is a $(d,b)_V$-cluster in $D$ if
    \begin{align}
    \diam_V(C)\le d,\quad d_{C_V}(D\setminus C,C)>b.
    \end{align}
    \item A vertex $v$ is $(d,b)_V$-clustered if there is a $(d,b)_V$-cluster $C$ such that $v \in C$.
\end{itemize}
\end{definition}

\begin{definition}[Clustered sets and isolated sets for vertices]\label{def:level-k_cluster_for_vertices}
Let $\{d_k^V\}_{k\ge 1}$ and $\{b_k^V\}_{k\ge 1}$ denote monotonically increasing sequences satisfying
\begin{align}
b_k^V \ge d_k^V > 0.
\end{align}
We define the level-$k$ clustered set $\mathrm{D}_k$ and the level-$k$ isolated set $\mathcal{D}_k$ of vertices inductively as
\begin{align}
\mathrm{D}_1 &= \mathcal{D}_1 = D, \\
\mathrm{D}_{k+1} &:= \mathrm{D}_k\setminus V_k,\\
\mathcal{D}_{k+1} &:= \mathcal{D}_k\setminus \mathcal{V}_k,
\end{align}
where
\begin{align}
V_k
&:= \{v\in \mathrm{D}_k \mid v \text{ is } (d_k^V,b_k^V)_V\text{-clustered in } \mathrm{D}_k\}, \\
\mathcal{V}_k
&:= \{v\in \mathcal{D}_k \mid v \text{ is } (d_k^V/2,b_k^V+d_k^V/2)_V\text{-isolated in } \mathcal{D}_k\}.
\end{align}
\end{definition}

\begin{definition}[Metric notions for edges]
We define the auxiliary notions.
    \begin{itemize}
        \item For subsets $C,C' \subseteq E$, we define
        \begin{align}
        d_{C_E}(C,C') := \min_{e\in C,\ e'\in C'} d_E(e,e').
        \end{align}
        \item We define the ball of radius $r$ around an edge $e$ by
        \begin{align}
        B_e(r) := \{e'\in E \mid d_E(e,e') \le r\}.
        \end{align}
        \item The diameter of a set of edges $C$ is defined as
        \begin{align}
        \diam_{E}(C) := \max_{e,e'\in C} d_E(e,e').
        \end{align}
    \end{itemize}
\end{definition}

\begin{definition}[Isolation and clustering of an edge set]
For $N \subseteq E$,
\begin{itemize}
    \item Two edges $e,e'\in N$ are called $(r,R)_E$-isolated if
    $e'\notin B_e(R)\setminus B_e(r)$.
    Otherwise, they are called $(r,R)_E$-linked.
    \item An edge $e\in N$ is $(r,R)_E$-isolated in $N$ if $e$ and $e'$ are $(r,R)_E$-isolated for all $e'\in N$.
    \item A subset $C\subseteq N$ is a $(d,b)_E$-cluster in $N$ if
    \begin{align}
    \diam_E(C)\le d,\qquad d_{C_E}(N\setminus C,C)>b.
    \end{align}
    \item An edge $e$ is $(d,b)_E$-clustered if there is a $(d,b)_E$-cluster $C$ such that $e \in C$.
\end{itemize}
\end{definition}

\begin{definition}[Clustered sets and isolated sets of edges]\label{def:level-k_cluster_for_edges}
Let $\{d_k^E\}_{k\ge 1}$ and $\{b_k^E\}_{k\ge 1}$ denote monotonically increasing sequences satisfying
\begin{align}
b_k^E \ge d_k^E > 0.
\end{align}
We define the level-$k$ clustered set $\mathrm{N}_k$ and the level-$k$ isolated set $\mathcal{N}_k$ inductively as
\begin{align}
\mathrm{N}_1 &= \mathcal{N}_1 = N, \\
\mathrm{N}_{k+1} &:= \mathrm{N}_k\setminus E_k,\\
\mathcal{N}_{k+1} &:= \mathcal{N}_k\setminus \mathcal{E}_k,
\end{align}
where
\begin{align}
E_k
&:= \{e\in \mathrm{N}_k \mid e \text{ is } (d_k^E,b_k^E)_E\text{-clustered in } \mathrm{N}_k\}, \\
\mathcal{E}_k
&:= \{e\in \mathcal{N}_k \mid e \text{ is } (d_k^E/2,b_k^E+d_k^E/2)_E\text{-isolated in } \mathcal{N}_k\}.
\end{align}
\end{definition}

For simplicity, when the distinction is clear from the context, we write $(d,b)$-cluster instead of $(d,b)_V$-cluster or $(d,b)_E$-cluster.

\section{Basic properties of the error-clustering hierarchy}\label{sup:property_of_error_cluster}
\begin{proposition}[Isolation implying clustering for edges]\label{prop:isolation_to_cluster}
Let $N\subseteq N'\subseteq E$ and $e\in N$.
If $e$ is $(r,R)$-isolated in $N'$, then $e$ is $(2r,R-r)$-clustered in $N$.
\end{proposition}

\begin{proof}
Define
\begin{align}
    C:=N\cap B_e(r).
\end{align}
Since $e\in N$, we have $e\in C$.
For any $f,f'\in C$, the triangle inequality gives
\begin{align}
    d_E(f,f')
    \le
    d_E(f,e)+d_E(e,f')
    \le
    2r.
\end{align}
Hence,
\begin{align}
    \operatorname{diam}_E(C)\le 2r.
\end{align}

Now take $f\in C$ and $g\in N\setminus C$.
Since $g\notin B_e(r)$, we have $d_E(e,g)>r$.
Moreover, $g\in N\subseteq N'$, and $e$ is $(r,R)$-isolated in $N'$.
Therefore,
\begin{align}
    d_E(e,g)>R.
\end{align}
Using the triangle inequality,
\begin{align}
    d_E(f,g)
    \ge
    d_E(e,g)-d_E(e,f)
    >
    R-r.
\end{align}
Thus,
\begin{align}
    d_{C_E}(N\setminus C,C)>R-r.
\end{align}
Therefore, $C$ is a $(2r,R-r)_E$-cluster in $N$ containing $e$, and hence $e$ is $(2r,R-r)$-clustered in $N$.
\end{proof}

\begin{proposition}[Separation of vertex-induced clusters]\label{prop:cluster_separation}
If $C,C'$ are $(d,b)_V$-clusters in the same set $D$ with $b\ge d$ and $C\neq C'$, then
\begin{align}
C\cap C' = \varnothing.
\end{align}
\end{proposition}

\begin{proof}
Suppose that there exists a vertex $v\in C\cap C'$. Since $C \neq C'$, either $C \setminus C'$ or $C' \setminus C$ is nonempty, and we assume that $v'\in C\setminus C'$ exists by symmetry and choose $v'\in C\setminus C'$.
Then, we have
\begin{align}
d_V(v,v')\le \diam_V(C)\le d\le b.
\end{align}
Since $v\in C'$ and $v'\in D\setminus C'$, we also have
\begin{align}
d_V(v,v')\ge d_{C_V}(D\setminus C',C')>b.
\end{align}
This is a contradiction.
\end{proof}

\begin{proposition}[Separation of edge-induced clusters]
If $C,C'$ are $(d,b)$-clusters in the same set $N$ with $b\ge d$ and $C\neq C'$, then
\begin{align}
C\cap C' = \varnothing.
\end{align}
\end{proposition}

\begin{proof}
Suppose that there exists an edge $e\in C\cap C'$. Since $C \neq C'$, either $C \setminus C'$ or $C' \setminus C$ is nonempty, and we assume that $e'\in C\setminus C'$ exists by symmetry and choose $e'\in C\setminus C'$.
Then, we have
\begin{align}
d_E(e,e')\le \diam_E(C)\le d\le b.
\end{align}
Since $e\in C'$ and $e'\in N\setminus C'$, we also have
\begin{align}
d_E(e,e')\ge d_{C_E}(N\setminus C',C')>b.
\end{align}
This is a contradiction.
\end{proof}

\begin{proposition}[$\mathrm{N}_k\subseteq \mathcal{N}_k$]\label{prop:N_k-subset}
For every $k$,
\begin{align}
\mathrm{N}_k\subseteq \mathcal{N}_k.
\end{align}
\end{proposition}

\begin{proof}
The case $k=1$ is trivial.
Assume that
\begin{align}
    \mathrm{N}_k\subseteq\mathcal{N}_k.
\end{align}
Take any
\begin{align}
    e\in\mathrm{N}_k\cap\mathcal{E}_k.
\end{align}
Then $e$ is
$(d_k^E/2,b_k^E+d_k^E/2)_E$-isolated in $\mathcal{N}_k$.
Applying Proposition~\ref{prop:isolation_to_cluster} with
\begin{align}
    N=\mathrm{N}_k,\quad
    N'=\mathcal{N}_k,\quad
    r=d_k^E/2,\quad
    R=b_k^E+d_k^E/2,
\end{align}
we conclude that $e$ is $(d_k^E,b_k^E)_E$-clustered in $\mathrm{N}_k$, and hence $e\in E_k$.
Therefore,
\begin{align}
    \mathrm{N}_k\cap\mathcal{E}_k\subseteq E_k.
\end{align}
It follows that
\begin{align}
    \mathrm{N}_{k+1}=
    \mathrm{N}_k\setminus E_k
    \subseteq
    \mathrm{N}_k\setminus\mathcal{E}_k
    \subseteq
    \mathcal{N}_k\setminus\mathcal{E}_k
    =
    \mathcal{N}_{k+1}.
\end{align}
\end{proof}

\begin{lemma}[Upper bound on the set of inclusion-minimal subsets for isolated sets of edges]\label{lem:size_of_M_k}
For each edge $e\in E$, let $M_k(e)$ denote the family of inclusion-minimal subsets of $E$ implying $e\in \mathcal{N}_k$.
That is, $N\in M_k(e)$ if and only if
\begin{enumerate}
    \item $e\in \mathcal{N}_k\subseteq N$,
    \item for any subset $N'\subsetneq N$, we have $e\notin\mathcal{N}'_k$, where $\mathcal{N}'_k$ denotes the level-$k$ isolated set obtained by constructing the hierarchy from $N'$.
\end{enumerate}
If
\begin{align}
d_{k+1}^E\ge 4b_k^E+3d_k^E
\end{align}
holds for all $k$, the following relations hold:
\begin{align}
N &\subseteq B_e(d_k^E/4)\quad (\forall N\in M_k(e)), \label{eq:N_ball_size}\\
|N|&=2^{k-1} \quad (\forall N\in M_k(e)). \label{eq:N_set_size}
\end{align}
Moreover, for $k > 1$, if
\begin{align}
|B_e(r)| \le A(1 + \Lambda r)^\Delta
\end{align}
is satisfied, the cardinality of $M_k(e)$ is bounded by
\begin{align}
    |M_k(e)|
    \le
    \prod_{j=0}^{k-2}
    \left[
    A\left(1 + \Lambda\left(b_{k-j-1}^E+\frac{d_{k-j-1}^E}{2}\right) \right)^\Delta
\right]^{2^j}. \label{eq:M_set_size}
\end{align}
\end{lemma}

\begin{proof}
We argue by induction to prove Eq.~\eqref{eq:N_ball_size}.
For $k=1$, we have $e \in \mathcal{N}_1$, $N = \{e\}$, and $M_1(e)=\{\{e\}\}$, so the statement is obvious.
If $e\in \mathcal{N}_{k+1}$, then by the definition of
\begin{align}
\mathcal{N}_{k+1}
:= \mathcal{N}_k\setminus \{e\in \mathcal{N}_k \mid e \text{ is } (d_k^E/2,b_k^E+d_k^E/2)\text{-isolated in } \mathcal{N}_k\},
\end{align}
there exists some $e'\in \mathcal{N}_k$ such that
$e$ and $e'$ are $(d_k^E/2,b_k^E+d_k^E/2)$-linked, indicating there also exist $N_1$ and $N_2$ that imply $e, e' \in \mathcal{N}_k$.
Then, by minimality, any $N\in M_{k+1}(e)$ can be written as
\begin{align}
N=N_1\cup N_2,\quad N_1\in M_k(e),\; N_2\in M_k(e').
\end{align}
By the induction hypothesis, we have
\begin{align}
N_1\subseteq B_e(d_k^E/4),\quad N_2\subseteq B_{e'}(d_k^E/4).
\end{align}
Since $e$ and $e'$ are $(d_k^E/2,b_k^E+d_k^E/2)$-linked, $d_E(e,e')\le b_k^E+d_k^E/2$. 
Hence, we obtain
\begin{align}
N
\subseteq B_e\bigl((b_k^E+d_k^E/2) + (d_k^E/4)\bigr)
= B_e(b_k^E+3d_k^E/4)
\subseteq B_e(d_{k+1}^E/4),
\end{align}
where the last inclusion uses $d_{k+1}^E\ge 4b_k^E+3d_k^E$.

We use the same decomposition
\begin{align}
N=N_1\cup N_2
\end{align}
to prove Eq.~\eqref{eq:N_set_size}.
Since $N_1\subseteq B_e(d_k^E/4)$ and $N_2\subseteq B_{e'}(d_k^E/4)$ hold, together with $d_E(e,e')>d_k^E/2$ induced by the fact that $e$ and $e'$ are $(d_k^E/2,b_k^E+d_k^E/2)$-linked, we see that $N_1$ and $N_2$ are disjoint.
Therefore, we obtain
\begin{align}
|N|=|N_1|+|N_2|=2^{k-1}+2^{k-1}=2^k.
\end{align}

For Eq.~\eqref{eq:M_set_size}, we consider the same decomposition.
Since $N=N_1\cup N_2$ with $N_1\in M_{k-1}(e)$ and $N_2\in M_{k-1}(e')$ for an edge $e'$ such that $e$ and $e'$ are $(d_{k-1}^E/2,b_{k-1}^E+d_{k-1}^E/2)$-linked, we obtain
\begin{align}
    |M_k(e)|
    &\le
    |M_{k-1}(e)| \sum_{e' \in B_e(b_{k-1}^E + d_{k-1}^E/2)\setminus B_e(d_{k-1}^E/2)}
    |M_{k-1}(e')| \\
    &\le
    |M_{k-1}(e)| |B_e(b_{k-1}^E+d_{k-1}^E/2)| \max_{e'\in E} |M_{k-1}(e')| \\
    &\le
    |B_e(b_{k-1}^E+d_{k-1}^E/2)|
    \left(\max_{e\in E} |M_{k-1}(e)|\right)^2 \\
    &\le
    A\left(1 + \Lambda\left(b_{k-1}^E+\frac{d_{k-1}^E}{2}\right) \right)^\Delta
    \left(\max_{e\in E} |M_{k-1}(e)|\right)^2 \\
    &\le
    \prod_{j=0}^{k-2}
    \left[
    A\left(1 + \Lambda\left(b_{k-j-1}^E+\frac{d_{k-j-1}^E}{2}\right) \right)^\Delta
    \right]^{2^j}, 
\end{align}
where we use $|M_1(e)| = 1$.
\end{proof}

\section{Probabilistic bounds for error clustering}\label{sup:error_clustering}

\begin{lemma}[Locality condition of the detector graph from the surface code]\label{lem:locality_condition}
For the detector graph of the surface code, the following conditions hold:
\begin{itemize}
    \item the degree of the detector graph is at most $12$,
    \item each edge has length at least $w_{\min}$.
\end{itemize}
Therefore, for every edge $e \in E$, we have
\begin{align}
    |B_e(r)|
    \le 12\sqrt3\pi\left(1+\frac{2r}{w_{\min}}\right)^3.
\end{align}
Hence we may take
\begin{align}
A = 12\sqrt3\pi,\quad \Lambda = \frac{2}{w_{\min}},\quad \Delta=3\label{eq:parameters}
\end{align}
for Lemma~\ref{lem:super_exponential_bound_p_k}.
\end{lemma}

\begin{proof}
Let $v_1$ and $v_2$ be the two endpoints of $e$.
If $d_E(e,e')\leq r$, then there exist $v_i\in\{v_1,v_2\}$ and an endpoint $u$ of $e'$ such that
\begin{align}
d_V(v_i,u)\leq r.
\end{align}
Therefore, every edge in $B_e(r)$ is incident to a vertex in $B_{v_1}(r)\cup B_{v_2}(r)$, and hence
\begin{align}
|B_e(r)|
\leq
12|B_{v_1}(r)|+12|B_{v_2}(r)|.
\end{align}

Now embed the unweighted detector graph into the cubic lattice $\mathbb Z^3$.
Since we are considering the surface code, each detector can be assigned an integer lattice coordinate, and adjacent vertices differ by at most $1$ in each coordinate.
Hence, traversing one detector-graph edge changes the Euclidean position by at most $\sqrt3$.

Since each edge has length at least $w_{\min}$ in a weighted detector graph, any path of length at most $r$ uses at most $r/w_{\min}$ edges.
Therefore, every vertex in $B_v(r)$ lies in the Euclidean ball of radius
\begin{align}
    R := \frac{\sqrt3\,r}{w_{\min}}
\end{align}
centered at the lattice point corresponding to $v$.

To bound the number of lattice points in this ball, assign to each lattice point $x\in \mathbb Z^3$ the unit cube $x+[-1/2,1/2]^3$.
These cubes are disjoint, and every cube corresponding to a lattice point in the ball of radius $R$ is contained in the ball of radius $R+\sqrt3/2$.
Hence, we have
\begin{align}
    |B_v(r)|
    \le \frac{4\pi}{3}\left(R+\frac{\sqrt3}{2}\right)^3
    = \frac{\sqrt3\pi}{2}\left(1+\frac{2r}{w_{\min}}\right)^3,
\end{align}
and we obtain
\begin{align}
    |B_e(r)|
    \le 12|B_v(r)| + 12|B_{v'}(r)|
    \le 12\sqrt3\pi\left(1+\frac{2r}{w_{\min}}\right)^3.
\end{align}
\end{proof}

The following clustering argument follows that of Ref.~\cite{yoshida2026prooffinitethresholdunionfind}, with the graph-locality bound replaced by Lemma~\ref{lem:locality_condition}.

\begin{lemma}[Super-exponential bound on clustering of errors on edges of graphs]\label{lem:super_exponential_bound_p_k}
Let $G=(V,E)$ be a graph, and assume that the set of faulty locations $N\subseteq E$ is randomly specified such a way that there exists $\tilde p$ satisfying
\begin{align}
    \Prob[N\supseteq S]\le \tilde p^{|S|}
\end{align}
for any $S\subseteq E$.
Suppose that there exist constants $A, \Lambda, \Delta$ such that the locality condition
\begin{align}
&|B_e(r)| \le A(1 + \Lambda r)^\Delta,\\
&b_k^E\ge d_k^E>0,\quad
d_{k+1}^E\ge 4b_k^E+3d_k^E
\end{align}
for all $e \in E$, $r \in \mathbb{Z}_{\ge 0}$, and $k \ge 1$.
Then, the probability
\begin{align}
p_k:=\max_{e\in E}\Prob[e\in \mathrm{N}_k]
\end{align}
is upper bounded by
\begin{align}
p_k
\le
\tilde p^{2^{k-1}}
\prod_{j=0}^{k-2}
\left[
A\left(1 + \Lambda\left(b_{k-j-1}^E+\frac{d_{k-j-1}^E}{2}\right) \right)^\Delta
\right]^{2^j}.
\end{align}
\end{lemma}

\begin{proof}
By Proposition~\ref{prop:N_k-subset}, we have $\mathrm{N}_k\subseteq \mathcal{N}_k$, so the inequality 
\begin{align}
    \Prob[e\in \mathrm{N}_k]\le \Prob[e\in \mathcal{N}_k]
\end{align}
holds.
If $e\in \mathcal{N}_k$ occurs, then there exists a minimal subset $M\in M_k(e)$ such that $M\subseteq N$.
Hence, by the union bound and the local-stochastic assumption, we have
\begin{align}
    \Prob[e\in \mathcal{N}_k]
    \le
    \sum_{M\in M_k(e)} \Prob[M\subseteq N]
    \le
    \sum_{M\in M_k(e)} \tilde p^{|M|}.
\end{align}
Using Eqs.~\eqref{eq:N_set_size} and \eqref{eq:M_set_size} from Lemma~\ref{lem:size_of_M_k}, we obtain
\begin{align}
    p_k
    &= \max_{e\in E} \Prob[e\in \mathrm{N}_k] \\
    &\le \max_{e\in E} \Prob[e\in \mathcal{N}_k] \\
    &\le \max_{e\in E} \sum_{M\in M_k(e)} \Prob[M\subseteq N] \\
    &\le \max_{e\in E} \sum_{M\in M_k(e)} \tilde p^{|M|} \\
    &= \max_{e\in E} |M_k(e)|\,\tilde p^{2^{k-1}} \\
    &\le \tilde p^{2^{k-1}}
    \prod_{j=0}^{k-2}
    \left[
    A\left(1 + \Lambda\left(b_{k-j-1}^E+\frac{d_{k-j-1}^E}{2}\right) \right)^\Delta
    \right]^{2^j}.
\end{align}
\end{proof}

\begin{lemma}[Threshold theorem on error clustering]\label{lem:threshold_theorem}
Suppose that $b_k^E$ and $d_k^E$ are given by
\begin{align}
b_k^E \leq \beta \lambda^{f(k+1)}+w_{\max},\quad
d_k^E \leq \gamma \lambda^{f(k)}-w_{\max},
\end{align}
where $\lambda > 1$ is a constant scale parameter, $\beta,\gamma>0$ are parameters that scale linearly with $w_{\max}$, and $f:\mathbb{Z}_{>0}\to \mathbb{R}_{\ge0}$ is monotonically increasing and satisfies
\begin{align}
\sum_{j=0}^{k-2} f(k-j)2^j \le c\,2^{k-1}.
\label{eq:supp_f_condition}
\end{align}
Assume also that $b_k^E$ and $d_k^E$ satisfy the assumptions of Lemma~\ref{lem:super_exponential_bound_p_k}.
Then, we have
\begin{align}
    p_k=O\!\left((p/p_{\mathrm{th}})^{2^{k-1}}\right),
\end{align}
where $p_\mathrm{th}$ is given by
\begin{align}
    p_{\mathrm{th}} &=
    \frac{1}{
    \xi A \lambda^{c\Delta}
    \left(
    \lambda^{-f(1)} + 
    \Lambda
    \left(
    \beta+\dfrac{\gamma+w_{\max}\lambda^{-f(1)}}{2}
    \right)
    \right)^{\Delta}}.
    \label{eq:supp_threshold}
\end{align}
\end{lemma}

\begin{proof}
Substituting
\begin{align}
b_{k-j-1}^E \leq \beta \lambda^{f(k-j)}+w_{\max},\quad
d_{k-j-1}^E \leq \gamma \lambda^{f(k-j-1)}-w_{\max}
\end{align}
into the upper bound from Lemma~\ref{lem:super_exponential_bound_p_k} gives
\begin{align}
p_k
&\le \tilde p^{2^{k-1}}
\prod_{j=0}^{k-2}
\left[
A\left(1 + \Lambda\left(b_{k-j-1}^E+\frac{d_{k-j-1}^E}{2}\right) \right)^\Delta
\right]^{2^j} \\
&\le
(\xi p)^{2^{k-1}}
\prod_{j=0}^{k-2}
\left[
A\left(
1 + 
\Lambda\left(
\beta \lambda^{f(k-j)}+w_{\max}+\frac{\gamma \lambda^{f(k-j-1)}-w_{\max}}{2}
\right)
\right)^\Delta
\right]^{2^j} \\
&=
(\xi p)^{2^{k-1}}
\prod_{j=0}^{k-2}
\left[
A\left(
\lambda^{-f(k-j)} + 
\Lambda\left(
\beta+\frac{\gamma \lambda^{f(k-j-1) - f(k-j)}+w_{\max}\lambda^{-f(k-j)}}{2}
\right)
\right)^{\Delta}
\lambda^{\Delta f(k-j)}
\right]^{2^j}.
\end{align}

Since $f$ is monotonically increasing, we have $f(k-j-1)-f(k-j)\le 0$ and $f(k-j)\ge f(1)$.
Together with $\lambda>1$, this implies
\begin{align}
\lambda^{f(k-j-1)-f(k-j)}\le 1,\qquad \lambda^{-f(k-j)}\le \lambda^{-f(1)}.
\end{align}

Then we have
\begin{align}
p_k
&\le
(\xi p)^{2^{k-1}}
\prod_{j=0}^{k-2}
\left[
A\left(
\lambda^{-f(1)} + 
\Lambda
\left(
\beta+\frac{\gamma+w_{\max}\lambda^{-f(1)}}{2}
\right)
\right)^\Delta
\lambda^{\Delta f(k-j)}
\right]^{2^j}\\
&=
(\xi p)^{2^{k-1}}
\left[
A\left(
\lambda^{-f(1)} + 
\Lambda
\left(
\beta+\frac{\gamma+w_{\max}\lambda^{-f(1)}}{2}
\right)
\right)^\Delta
\right]^{2^{k-1}-1}
\lambda^{\Delta\sum_{j=0}^{k-2} f(k-j)2^j} \\
&\le
(\xi p)^{2^{k-1}}
\left[
A\left(
\lambda^{-f(1)} + 
\Lambda
\left(
\beta+\frac{\gamma+w_{\max}\lambda^{-f(1)}}{2}
\right)
\right)^\Delta
\right]^{2^{k-1}-1}
\lambda^{\Delta c\,2^{k-1}} \\
&=
O(1)
\left[
(\xi p)
A\left(
\lambda^{-f(1)} + 
\Lambda
\left(
\beta+\frac{\gamma+w_{\max}\lambda^{-f(1)}}{2}
\right)
\right)^{\Delta}
\lambda^{\Delta c}
\right]^{2^{k-1}}.
\end{align}

Finally, we obtain
\begin{align}
    p_k &= O\!\left((p/p_{\mathrm{th}})^{2^{k-1}}\right), \\
    p_{\mathrm{th}} &=
    \frac{1}{
    \xi A \lambda^{c\Delta}
    \left(
    \lambda^{-f(1)} + 
    \Lambda
    \left(
    \beta+\dfrac{\gamma+w_{\max}\lambda^{-f(1)}}{2}
    \right)
    \right)^{\Delta}}.
\end{align}
\end{proof}

\section{Localization of Sparse-Blossom Dynamics by Error Clusters}\label{sup:equivalence_from_error_cluster}

\begin{definition}[Mono-edge]
Recall that each detector-graph edge has an even integer length. 
For an edge of length $2m$, we subdivide it into $2m$ unit-length segments. 
Each such unit segment is called a mono-edge. 
Equivalently, a mono-edge is the amount of edge length traversed by a growing graph fill region during one unit of time, since growing graph fill regions move at unit speed in sparse blossom.
\end{definition}

A mono-edge is called \emph{highlighted} once it has been explored by a graph-fill region during the sparse-blossom execution.

\begin{definition}[SB cluster]
We initialize a singleton SB cluster for each active detector. 
At a given algorithmic time of $t > 0$, an SB cluster is the edge-induced subgraph of highlighted mono-edges. 
Only invalid SB clusters grow. 
An SB cluster is called invalid if it contains at least one unresolved alternating tree; equivalently, if the sparse-blossom process associated with it contains non-frozen graph fill regions; otherwise, it is called valid.
By saying that an SB cluster grows, we mean that the internal sparse-blossom dynamics associated with that cluster execute in time. 
Equivalently, the graph fill regions contained in the SB cluster grow with rates $+1$, $0$, or $-1$, corresponding respectively to growing, frozen, or shrinking in accordance with sparse-blossom dynamics. 
When a COLLIDE event is generated by two graph fill regions associated with active detectors in different SB clusters, and at least one of these graph fill regions is growing, the corresponding SB clusters are merged into a single SB cluster.
\end{definition}

\begin{lemma}[Persistent root for an invalid SB cluster]\label{lem:persistent_root}
Let $S_t$ be an SB cluster that is invalid at time $t$.
Then there exists an alternating tree $T_t \subset S_t$ and an active detector $v \in S_t$ such that, when tracing the ancestry of $T_t$ backward in time, the root associated with $v$ is preserved throughout the whole time interval $[0, t]$.
In particular, at every moment in that interval, the top-level active region containing $v$ is an outer node and hence has a growth rate $+1$.
\end{lemma}

\begin{proof}
We first prove the existence of a persistent alternating tree.
Since $S_t$ is invalid, it contains at least one alternating tree at time $t$. 
Fix one such tree and denote it by $T_t$.

We trace the execution history backward in time. 
We consider the event preceding the current state of the tree. 
Events of type (a), (c), (d), and (e) do not terminate an alternating tree. 
More precisely, event (a) attaches a matched pair to an already existing alternating tree; event (c) forms a blossom inside an already existing alternating tree; event (d) shatters an inner blossom; and event (e) forms a blossom when an inner node shrinks to zero. 
In all of these cases, if an alternating tree is present after the event, then it has a well-defined predecessor alternating tree before the event.
On the other hand, events of type (b), (f), and (g) terminate only the alternating trees that are directly involved in the corresponding augmentation or boundary match. 
Therefore, if the chosen tree is still present after such an event, then it was not one of the terminated trees. 
Hence, the same alternating tree already existed immediately before the event. 
In particular, when such an event also merges SB clusters, the chosen tree after the merge has a predecessor in one of the pre-merge SB clusters.

Thus, by induction backward over the event history, the chosen tree $T_t$ has an ancestral alternating tree at every earlier time until its initial singleton tree is reached. 
We call this sequence of ancestral trees a persistent alternating tree.

It remains to show that this persistent alternating tree has a persistent root.
Let $v$ be the root node of the initial singleton tree in the ancestry of $T_t$. 
We show that the root associated with $v$ is preserved under all events on the ancestry of the persistent tree.

For event (a), a matched pair is attached below an already existing outer node of the tree. 
The root of the tree is not modified. 
Therefore, the root associated with $v$ remains the root.
For event (c), two outer regions in the same alternating tree collide and form a blossom. 
If the root node associated with $v$ is not involved in the collision, then it is unchanged. 
If it is involved in the collision, then the newly formed blossom becomes the top-level outer representative containing $v$. 
In either case, the top-level active region containing $v$ remains outer.
For event (d), an inner blossom shrinks and is shattered. Since the event acts on an inner object, it cannot eliminate the  root of the alternating tree. 
Thus, the root associated with $v$ is preserved.
For event (e), an inner object shrinks to zero and the representation of the surrounding node is updated. 
This changes only the internal representation of the alternating tree. 
It does not remove the root, and the top-level active region containing $v$ remains outer.
Finally, events (b), (f), and (g) terminate the alternating trees that are directly involved in the event. 
Since the persistent tree is present after the event, its ancestry cannot be one of those terminated trees. 
Hence, these events are not on the ancestry in a way that changes or destroys the root of the persistent tree.

Consequently, throughout the entire ancestry of $T_t$, the root associated with $v$ is preserved. 
Moreover, at every time at which the ancestral tree exists, the top-level active region containing $v$ is an outer node. 
By the definition of the sparse-blossom growth rule, every outer top-level active region has growth rate $+1$. 
This proves the claim.
\end{proof}

\begin{definition}[Extended SB clusters]
At each step $t$, in order to treat SB clusters and level-$k$ error clusters together, we define the smallest equivalence relation $\sim_t$ on the set of mono-edges satisfying the following conditions:
\begin{itemize}
    \item mono-edges belonging to the same SB cluster are equivalent,
    \item mono-edges belonging to the same level-$k$ error cluster are also equivalent.
\end{itemize}
The equivalence classes are called extended SB clusters.
An extended SB cluster that contains a level-$k$ cluster but no cluster of a higher level is called a level-$k$ extended SB cluster.
If all SB clusters contained within it are valid and have stopped growing, we call the extended SB cluster stable.
\end{definition}

We let $\mathrm{N}_{k,t}$ denote the set of level-$k$ extended SB clusters at step $t$ and define the margin of a level-$k$ extended SB cluster.

\begin{definition}[Margin]
For $\tilde C_t\in \mathrm{N}_{k,t}$, define its margin $m_E(\tilde C_t)$ as the maximum distance from the boundary of the highest-level error cluster inside $\tilde C_t$ to the boundary of $\tilde C_t$.
Then, we define the margin as
\begin{align}
m_k^E := \max_t \max_{\tilde C_t\in \mathrm{N}_{k,t}} m_E(\tilde C_t).
\end{align}
\end{definition}

\begin{lemma}[Stopping guarantee on the growth of extended SB clusters]\label{lem:stopping_guarantee}
Assume that level-$k$ error clusters satisfy the $(d_k^E,b_k^E)_E$-cluster condition.
Define
\begin{align}
\phi_1:=1,\quad 
\phi_k:=1-\sum_{k''=1}^{k-1}
\frac{(\phi_{k''}+2)(d_{k''}^E+w_{\max})}
     {(\phi_{k''}+1)(d_{k''}^E+w_{\max})+\phi_{k''}(b_{k''}^E-w_{\max})}
\qquad (k\ge 2).
\end{align}
If, for all $k$,
\begin{align}
\phi_k > 2\frac{d_k^E+w_{\max}}{b_k^E-w_{\max}} > 0,
\end{align}
then:
\begin{enumerate}
    \item any level-$k$ extended SB cluster stops growing before merging with any level-$k'$ extended SB cluster with $k'\ge k$;
    \item
    \begin{align}
    m_k^E \le \frac{d_k^E+w_{\max}}{\phi_k}.
    \end{align}
\end{enumerate}
\end{lemma}

\begin{proof}
We prove the statement by induction.
The geometric part of the induction follows the argument of Ref.~\cite{yoshida2026prooffinitethresholdunionfind}, while Lemma~\ref{lem:persistent_root} is used to control the sparse-blossom dynamics.
By the definition of the edge metric $d_E$, if an error cluster has edge diameter at most $d_k^E$, then any two detector endpoints incident to edges in that cluster are separated by vertex distance at most $d_k^E+w_{\max}$.

For $k=1$, let $\tilde C_t$ be a level-$1$ extended SB cluster.
We show that $\tilde C_t$ becomes stable by time $d_1^E+w_{\max}$.

Suppose, for contradiction, that $\tilde C_t$ is not stable at some time $t>d_1^E+w_{\max}$.
Then $\tilde C_t$ contains an invalid SB cluster.
By Lemma~\ref{lem:persistent_root}, there exists an active detector $v$ such that the top-level region rooted at $v$ remains an outer node throughout the time interval $[0,t]$.
We recall the following property of the graph fill region of sparse blossom.
An active outer node has growth rate $+1$, and its graph fill region expands through the detector graph at unit speed from its source.
If this outer node is absorbed into a blossom, the already explored region is not discarded. 
The new top-level blossom has as its source the surface formed by its blossom children, and the top-level outer blossom continues to grow from this surface with growth rate $+1$.
Equivalently, the set of detector nodes owned by the persistent top-level outer region, including active detectors owned by its blossom descendants, is monotonically increasing in time and expands by one mono-edge unit per unit time.
Since $\tilde C_t$ is a level-$1$ extended SB cluster, there are no lower-level extended SB clusters inside it.
Let $C$ be the level-$1$ error cluster contained in $\tilde C_t$.
For any active detector $v'\in C$, the observation above gives a mono-edge path $P$ from $v$ to $v'$ with $|P|\le d_1^E+w_{\max}$.
By the property of the graph fill region above, the persistent top-level outer region rooted at $v$ reaches every node at mono-edge distance at most $d_1^E+w_{\max}$ from $v$ by time $d_1^E+w_{\max}$.
Hence, by time $d_1^E+w_{\max}$, this persistent outer growth has explored all active detectors in the level-$1$ error cluster $C$.

It remains to show that no unresolved alternating tree can persist after this exploration has occurred.
First consider the boundary-free case.
The level-$1$ error cluster $C$ induces an even number of active detectors, since the set of detector vertices incident to an odd number of faulty edges has even cardinality in the absence of a boundary.
Each matched object contains an even number of active detectors whereas each unresolved alternating tree contains an odd number of active detectors.
Therefore, the number of unresolved alternating trees contained in the explored cluster is even.
If at least two unresolved alternating trees remained after the whole cluster had been explored, the explored graph fill regions would provide a collision event between two unresolved trees, unless that collision had already been processed.
Thus unresolved alternating trees inside the fully explored level-$1$ cluster are eliminated in pairs.
Consequently, no unresolved alternating tree remains in the boundary-free case.

Now consider the case where $C$ touches the boundary.
As discussed above, unresolved alternating trees in the explored cluster can be eliminated in pairs by collision events.
If one unresolved alternating tree remains, because the whole level-$1$ error cluster has already been explored and $C$ touches the boundary, this remaining tree has an explored path to the boundary or to a boundary match.
The boundary event associated with this path turns the remaining unresolved tree into a boundary match.
Hence no unresolved alternating tree remains in the boundary case either.
Therefore, after time $d_1^E+w_{\max}$, every component of $\tilde C_t$ is a union of matched objects, or is matched to the boundary in the boundary case.
Thus the SB cluster is valid, and $\tilde C_t$ is stable.
This contradicts the assumption that $\tilde C_t$ is not stable at some time $t>d_1^E+w_{\max}$.

Consequently, every level-$1$ extended SB cluster becomes stable by time $d_1^E+w_{\max}$.

Moreover, since $\phi_1=1$ and there are no lower-level clusters for
$k=1$, the margin satisfies
\begin{align}
    m_1^E \le d_1^E+w_{\max} = \frac{d_1^E+w_{\max}}{\phi_1}.
\end{align}

For $k > 1$, we assume the claim holds for all smaller levels $k'' < k$.
Let $\tilde C_t$ be a level-$k$ extended SB cluster, and let $C$ be its level-$k$ error cluster.

First, we show that, by time $d_k^E+w_{\max}$, $\tilde C_t$ becomes stable unless it merges with another level-$k'$ extended SB cluster with $k'\ge k$.
Suppose, for contradiction, that $\tilde C_t$ has not merged with any such cluster and remains unstable until time $d_k^E+w_{\max}$.
Then it contains an invalid SB cluster $S\subset \tilde C_t$.
By Lemma~\ref{lem:persistent_root}, there exists an active detector $v$ such that the top-level region rooted at $v$ remains outer throughout the interval $[0,t]$.

The same persistent-root growth argument as in the case $k=1$ applies: if this region is absorbed into a blossom, the explored radius is not reset, and lower-level stopped clusters or matched objects do not slow down the growth of the top-level outer region containing $v$.
Thus, the persistent outer region explores mono-edge distance at rate $+1$.
Since any two active detectors induced by the level-$k$ error cluster $C$ are separated by at most $d_k^E+w_{\max}$, the persistent root explores all active detectors in $C$ by time $d_k^E+w_{\max}$.

The parity and boundary argument from the case $k=1$ then implies that no unresolved alternating tree can remain inside $\tilde C_t$, contradicting the assumption that it is unstable. 
Therefore, $\tilde C_t$ is stable by time $d_k^E+w_{\max}$ unless such a merge occurs.

Next, suppose for contradiction that $\tilde C_t$ merges with $\tilde C_t'$ before stopping.
Let $C'$ be the level-$k'$ cluster inside $\tilde C_t'$.
We take $t$ satisfying $t \le d_k^E + w_{\max}$ and show that no such merge can occur.
If $\tilde C_t$ merges with $\tilde C_t'$, then the region between $\tilde C_t$ and $\tilde C_t'$ must be filled by a chain of smaller-level extended SB clusters
\begin{align}
\tilde C_t^{(1)}, \dots, \tilde C_t^{(n)},
\end{align}
and the total distance among clusters must satisfy
\begin{align}
&d_{C_E}(C, \tilde C_t^{(1)}) + \sum_{i=1}^{n-1} d_{C_E}(\tilde C_t^{(i)}, \tilde C_t^{(i+1)}) + d_{C_E}(\tilde C_t^{(n)}, C') \\
&\le 2t \le 2(d_k^E + w_{\max}).
\end{align}
On the other hand, the level-$k$ clusters are separated by a buffer, so we have
\begin{align}
d_{C_E}(C, C') \ge b_k^E - w_{\max}.
\end{align}
Since, along a path between $C$ and $C'$, the fraction of the path length occupied by level-$k''$ extended SB clusters is at most (See Fig.~3(a) of Ref.~\cite{yoshida2026prooffinitethresholdunionfind} for an illustration of the same geometric configuration)
\begin{align}
&\frac{(m_{k''}^E) + (d_{k''}^E + w_{\max}) + (m_{k''}^E)}{(b_{k''}^E - w_{\max}) + (d_{k''}^E + w_{\max}) + (m_{k''}^E)} \\
&= \frac{d_{k''}^E + w_{\max} + 2m_{k''}^E}{d_{k''}^E + w_{\max} + m_{k''}^E + b_{k''}^E - w_{\max}}, \label{eq:fraction}
\end{align}
the fraction of edges not occupied by any level-$k''$ extended SB clusters for all $k'' < k$ is at least
\begin{align}
    &1 - \sum_{k''=1}^{k-1} \frac{d_{k''}^E + w_{\max} + 2m_{k''}^E}{d_{k''}^E + w_{\max} + m_{k''}^E + b_{k''}^E - w_{\max}} \\
    &\ge 1 - \sum_{k''=1}^{k-1} \frac{d_{k''}^E + w_{\max} + 2\frac{d_{k''}^E + w_{\max}}{\phi_{k''}}}{d_{k''}^E + w_{\max} + \frac{d_{k''}^E + w_{\max}}{\phi_{k''}} + b_{k''}^E - w_{\max}} \\
    &= 1 - \sum_{k''=1}^{k-1} \frac{(\phi_{k''} + 2)(d_{k''}^E + w_{\max})}{(\phi_{k''} + 1)(d_{k''}^E + w_{\max}) + \phi_{k''}(b_{k''}^E - w_{\max})} \\
    &= \phi_k.
\end{align}

Therefore, the distance among clusters, that is, the distance of edges not occupied by any level-$k''$ extended SB clusters for all $k'' < k$, is bounded below by
\begin{align}
&d_{C_E}(C, \tilde C_t^{(1)}) + \sum_{i=1}^{n-1} d_{C_E}(\tilde C_t^{(i)}, \tilde C_t^{(i+1)}) + d_{C_E}(\tilde C_t^{(n)}, C') \\
&\ge (b_k^E - w_{\max})\phi_k.
\end{align}
Thus, if
\begin{align}
(b_k^E - w_{\max})\phi_k > 2(d_k^E + w_{\max}),
\end{align}
then the free distance that must be traversed before $\tilde C_t$ can merge with $\tilde C_t'$ is larger than the available growth length $2(d_k^E + w_{\max}) \ge 2t$.
Hence, such a merge is impossible.

Thus, we have shown that for $t \le d_k^E + w_{\max}$, a level-$k$ extended SB cluster stops before merging with another level-$k'$ extended SB cluster $\tilde C_t'$ with $k' \ge k$. 
Here, since the cluster stops within this time bound, we do not need to consider the case where $t > d_k^E + w_{\max}$.

Next, let $v$ be an active detector contained in an invalid SB cluster inside $\tilde C_t$, chosen as in Lemma~\ref{lem:persistent_root}.
Along any path from $v$, the fraction of edges not occupied by smaller-level clusters is at least $\phi_k$.
Therefore, in order for the boundary of $\tilde C_t$ to be at a distance $m_E(\tilde C_t)$ from the cluster $C$, the root node growing from $v$ must traverse at least
\begin{align}
\phi_k m_E(\tilde C_t)
\end{align}
units of free distance.
Since this root node grows with rate $1$, traversing that amount of free distance requires at least the same amount of time.
Hence
\begin{align}
\phi_k m_E(\tilde C_t) \le t,
\end{align}
and therefore
\begin{align}
m_E(\tilde C_t) \le \frac{t}{\phi_k}.
\end{align}
Since $t \le d_k^E + w_{\max}$, we obtain
\begin{align}
m_k^E \le \frac{d_k^E + w_{\max}}{\phi_k}.
\end{align}
\end{proof}

\section{Reduction to Algorithmic Processing Clusters}\label{sup:reduction_to_processing_cluster}

Unlike the error-cluster hierarchy used for the analysis in Ref.~\cite{yoshida2026prooffinitethresholdunionfind}, the processing clusters considered below are constructed solely from the observed active detectors.

We first recall the construction of the processing-cluster hierarchy used by the parallel decoder.

\begin{definition}[Hierarchy of processing clusters]
\label{def:sup_processing_cluster_hierarchy}
Let $D\subseteq V$ be the set of detection events.
For each level $k$, let
\begin{align}
\mathcal{C}_k:=\{C_{k,i}\}_{i=1}^{I_k}
\end{align}
denote the collection of level-$k$ processing clusters.
Each processing cluster $C_{k,i}$ satisfies the parity condition
\begin{align}
C_{k,i}\cap V_{\mathrm{bd}}\neq\varnothing
\quad\text{or}\quad
|C_{k,i}\cap D|\equiv0\pmod 2.
\end{align}
\end{definition}

The processing-cluster hierarchy is constructed level by level from the observed detection events.
Set $D_1:=D$, and let $D_k$ denote the residual detection-event set at the beginning of level $k$.
At level $k$, an exploring region of radius $b_k^V/2$ is grown around every detection event in $D_k$, and connected components are formed from mutually touching exploring regions.
A connected component is declared to be a level-$k$ processing cluster if its detection-event set has a diameter of at most $d_k^V$ and satisfies the parity condition.
The detection events belonging to the declared processing clusters are removed from the residual set to obtain $D_{k+1}$.
This procedure is repeated until the residual set becomes empty.

When a processing cluster stops, its stopped matching configuration is retained through subsequent levels and is referred to as an \emph{inherited matching configuration}.
An inherited matching configuration is said to be \emph{touched} by a processing cluster if an event in its sparse-blossom execution modifies the intermediate state associated with any active detector belonging to that configuration.

The subsequent execution of each stopped matching configuration is determined solely by whether it is touched by a processing-cluster execution.
If an inherited matching configuration is touched by a processing cluster, the sparse-blossom dynamics of that processing cluster continue from and modify the corresponding intermediate state.
When the processing cluster stops, the active detectors belonging to the touched inherited matching configuration are included in the stopped matching configuration produced by that processing cluster, which is then retained through subsequent levels.
If an inherited matching configuration is not touched, it remains unchanged and continues to be retained.

A stopped matching configuration is called a \emph{terminal matching configuration} if it is not touched by any processing-cluster execution after it is produced.
After all processing-cluster executions have stopped, the terminal matching configurations are combined to obtain the final matching result.

We refer to these procedures collectively as the \textit{hierarchical execution rule}.

\begin{definition}[Executing SB clusters]
At each step $t$, in order to treat SB clusters and level-$k$ processing clusters together, we define the smallest equivalence relation $\sim_t$ on the set of detectors as follows:
\begin{itemize}
    \item detectors belonging to the same SB cluster are equivalent,
    \item detectors belonging to the same processing cluster, at any level, are also equivalent.
\end{itemize}
The equivalence classes are called executing SB clusters.
An executing SB cluster that contains a level-$k$ processing cluster but no processing cluster of a higher level is called a level-$k$ executing SB cluster.
If all SB clusters contained within it are valid and have stopped growing, we call the executing SB cluster stable.
\end{definition}

We let $\mathrm{D}_{k,t}$ denote the set of level-$k$ executing SB clusters at step $t$ and define the margin of a level-$k$ executing SB cluster.

\begin{definition}[Margin]
For $\tilde C_t\in \mathrm{D}_{k,t}$, define its margin $m_V(\tilde C_t)$ as the maximum distance from the boundary of the highest-level processing cluster inside $\tilde C_t$ to the boundary of $\tilde C_t$.
Then, we define the margin as
\begin{align}
m_k^V := \max_t \max_{\tilde C_t\in \mathrm{D}_{k,t}} m_V(\tilde C_t).
\end{align}
\end{definition}

Then, the level-$k$ executing SB clusters satisfy the following lemma.

\begin{lemma}[Stopping guarantee on processing clusters]
\label{lem:stopping_guarantee_for_processing_cluster}
Assume that level-$k$ processing clusters satisfy the $(d_k^V,b_k^V)_V$-cluster condition and the parity condition.
Define
\begin{align}
\phi_1:=1,\quad 
\phi_k:=1-\sum_{k''=1}^{k-1}
\frac{(\phi_{k''}+2)d_{k''}^V}
     {(\phi_{k''}+1)d_{k''}^V+\phi_{k''}b_{k''}^V}
\qquad (k\ge 2).
\end{align}
If, for all $k$,
\begin{align}
\phi_k > 2\frac{d_k^V}{b_k^V} > 0,
\end{align}
then:
\begin{enumerate}
    \item any level-$k$ executing SB cluster stops growing before merging with any level-$k'$ executing SB cluster with $k'\ge k$;
    \item
    \begin{align}
    m^V_k \le \frac{d_k^V}{\phi_k}.
    \end{align}
\end{enumerate}
\end{lemma}

\begin{proof} The proof is the same as that of Lemma~\ref{lem:stopping_guarantee}. 
The essential ingredient of that proof is the existence of a persistent root in every invalid SB cluster, as established in Lemma~\ref{lem:persistent_root}. 
Such a root grows at unit speed and therefore explores all relevant detection events within the diameter bound of the processing cluster. 
The parity condition then guarantees that the execution stops, while the buffer condition guarantees that it stops before interacting with any distinct processing cluster of the same or a higher level. 
The same free-distance argument also gives the corresponding margin bound. 
None of these arguments depends on the processing cluster being induced by an actual faulty-edge configuration. \end{proof}

\begin{lemma}[Unique touching of inherited matching configurations]
\label{lem:disjointness_of_inherited_state}
Let $P=C_{k,i}$ and $Q=C_{k,j}$, with $i\neq j$, be distinct level-$k$ processing clusters.
Suppose that the clustering parameters satisfy the conditions of Lemma~\ref{lem:stopping_guarantee_for_processing_cluster}.
Then any inherited matching configuration can be touched by at most one of $P$ and $Q$ before they stop.
\end{lemma}

\begin{proof}
Suppose, for contradiction, that both $P$ and $Q$ touch the same inherited matching configuration.
Consider the first such touching event in the joint parallel execution.
Without loss of generality, suppose that this event is generated by the execution of $P$.

Before this event, the inherited matching configuration consists of frozen matched objects contained in the same executing SB cluster.
The first touching event generated by $P$ connects the sparse-blossom configuration generated by $P$ to one of these frozen matched objects, and hence places them in the same executing SB cluster.

If the execution of $Q$ subsequently touches the state involving the same inherited matching configuration before stopping, possibly through a different matched object belonging to that configuration, the sparse-blossom configuration generated by $Q$ also becomes part of the same executing SB cluster.
Thus, the executing SB clusters associated with the distinct level-$k$ processing clusters $P$ and $Q$ merge before stopping, contradicting Lemma~\ref{lem:stopping_guarantee_for_processing_cluster}.
\end{proof}

\begin{definition}[Associated processing cluster with error cluster]\label{def:associated_processing_cluster}
Let $D \subseteq V$ be the set of detection events, and let $V_{\mathrm{bd}} \subseteq V$ be the set of boundary vertices. 
Let $N_{k,i} \subseteq E$ be a level-$k$ error cluster, where the index $i$ distinguishes each level-$k$ error cluster. 
The processing cluster associated with $N_{k,i}$ is defined by
\begin{align}
    \mathsf{P}(N_{k,i})
    :=
    (D \cup V_{\mathrm{bd}})
    \cap
    \bigcup_{e \in N_{k,i}} \partial e,
\end{align}
where $\partial e$ denotes the set of endpoints of the edge $e$.
\end{definition}

\begin{lemma}[Conversion between edge clusters and vertex clusters]\label{lem:associated_processing_cluster_is_processing_cluster}
Let $N_{k,i}\subseteq E$ be a level-$k$ error cluster satisfying the $(d_k^E,b_k^E)_E$-cluster condition, and define
\begin{align}
\hat D_k
:=
(D\cup V_{\mathrm{bd}})
\cap
\bigcup_{e\in N_k}\partial e.
\end{align}
Then the associated processing cluster $\mathsf{P}(N_{k,i})$ is a $(d_k^E+w_{\max},b_k^E-w_{\max})_V$-cluster in $\hat D_k$.
\end{lemma}
    
\begin{proof}
Since $N_{k,i}$ satisfies the $(d_k^E,b_k^E)_E$-cluster condition, we have
\begin{align}
    \diam_E(N_{k,i}) \leq d_k^E
\end{align}
and
\begin{align}
    d_{C_E}(N_k \setminus N_{k,i}, N_{k,i}) > b_k^E,
\end{align}
where $N_k$ denotes the residual edge set at level $k$.

We first bound the vertex diameter of the associated processing cluster. 
Take arbitrary vertices $v,v' \in \mathsf{P}(N_{k,i})$. 
By the definition of $\mathsf{P}$, there exist edges $e,e' \in N_{k,i}$ such that $v \in \partial e$ and $v' \in \partial e'$. 
By the definition of the edge distance $d_E$, the distance between endpoints of $e$ and $e'$ is bounded by the edge distance plus the two half-edge corrections. 
Hence,
\begin{align}
    d_V(v,v')
    &\leq d_E(e,e') + \frac{w(e)}{2} + \frac{w(e')}{2} \\
    &\leq d_k^E + w_{\max}.
\end{align}
Therefore,
\begin{align}
    \diam_V(\mathsf{P}(N_{k,i})) \leq d_k^E + w_{\max}.
\end{align}

We next bound the separation from the remaining vertices at the same residual level.
Define
\begin{align}
    \hat D_k
    :=
    (D\cup V_{\mathrm{bd}})
    \cap
    \bigcup_{e\in N_k}\partial e.
\end{align}
Let $u\in\hat D_k\setminus\mathsf{P}(N_{k,i})$ and $v\in\mathsf{P}(N_{k,i})$.
By definition, there exist $f\in N_k\setminus N_{k,i}$ and $e\in N_{k,i}$ such that $u\in\partial f$ and $v\in\partial e$.
By the $(d_k^E,b_k^E)_E$-cluster condition,
\begin{align}
    d_E(e,f) > b_k^E.
\end{align}
Again using the relation between edge distance and endpoint distance, we obtain
\begin{align}
    d_V(u,v)
    &\geq d_E(e,f) - \frac{w(e)}{2} - \frac{w(f)}{2} \\
    &> b_k^E - w_{\max}.
\end{align}
Thus,
\begin{align}
    d_{C_V}(\hat D_k \setminus \mathsf{P}(N_{k,i}), \mathsf{P}(N_{k,i}))
    > b_k^E - w_{\max}.
\end{align}

Then, the associated processing cluster $\mathsf{P}(N_{k,i})$ satisfies the vertex-based cluster condition with parameters
\begin{align}
    (d_k^E+w_{\max}, b_k^E-w_{\max})_V.
\end{align}
\end{proof}

\begin{lemma}[Decomposition into processing clusters]
\label{lem:covering_processing_clusters}
Let $D \subseteq V$ be the set of detection events, and let $N \subseteq E$ be a faulty-edge set for $D$.
For each level-$k$ error cluster $N_{k,i}$, define
\begin{align}
    A_{k,i}
    :=
    \mathsf{P}(N_{k,i}) \cap D.
\end{align}
Assume that each associated processing cluster $\mathsf{P}(N_{k,i})$ satisfies the corresponding $(d_k^V,b_k^V)_V$-cluster condition and the parameters satisfy the conditions of Lemma~\ref{lem:stopping_guarantee} and Lemma~\ref{lem:stopping_guarantee_for_processing_cluster}.

Then, for every level-$k$ error cluster $N_{k,i}$, there exists an index set
\begin{align}
    \mathcal{I}_{k,i}
    \subseteq
    \left\{
        (k',j)
        \,\middle|\,
        1 \leq k' \leq k
    \right\}
\end{align}
such that
\begin{align}
    A_{k,i}
    =
    \bigsqcup_{(k',j)\in\mathcal{I}_{k,i}}
    \left(C_{k',j}\cap D\right).
\end{align}
In particular, every active detector in $\mathsf{P}(N_{k,i})$ belongs to a processing cluster of level at most $k$.
\end{lemma}

\begin{proof}
We prove the claim by induction on $k$.
Let $D_k$ denote the residual detection-event set at the beginning of level $k$ of the processing-cluster construction.

For $k=1$, we have $D_1=D$.
By assumption, $\mathsf{P}(N_{1,i})$ satisfies the $(d_1^V,b_1^V)_V$-cluster condition.
Hence, the exploring regions grown from the active detectors in $A_{1,i}$ form one connected component and do not touch any active detector outside $A_{1,i}$.

It remains to verify the parity condition.
If $\mathsf{P}(N_{1,i})$ contains no boundary vertex, compatibility of $N$, together with the separation of $N_{1,i}$ from the remaining edges, implies that $A_{1,i}$ contains an even number of active detectors.
If $\mathsf{P}(N_{1,i})$ contains a boundary vertex, every active detector in $A_{1,i}$ is at distance at most $d_1^V$ from this boundary vertex.
Since $\phi_1=1$ and $\phi_1>2d_1^V/b_1^V$, we have $d_1^V<b_1^V/2$, so the exploring regions reach the boundary.
Thus, in either case, the connected component satisfies the parity condition and is declared to be a level-$1$ processing cluster $C_{1,j}$ satisfying
\begin{align}
    C_{1,j}\cap D
    =
    A_{1,i}.
\end{align}
This proves the claim for $k=1$.

Now let $k\geq 2$, and assume that the claim holds for every level $k'<k$.
Fix a level-$k$ error cluster $N_{k,i}$, and write
\begin{align}
    A
    &:=
    A_{k,i}
    =
    \mathsf{P}(N_{k,i})\cap D,\\
    R
    &:=
    A\cap D_k.
\end{align}
The active detectors in $A\setminus R$ have already been removed by processing clusters of levels strictly smaller than $k$.

By the induction hypothesis, for every level-$k'$ error cluster $N_{k',h}$ with $k'<k$, the active detectors in $\mathsf{P}(N_{k',h})$ are partitioned into processing clusters of levels at most $k'$.
We call these processing clusters the processing-cluster constituents of $\mathsf{P}(N_{k',h})$.

Consider a lower-level processing cluster $C_{\ell,j}$ with $\ell<k$ that is not already a constituent of any $\mathsf{P}(N_{k',h})$ with $k'<k$.
We claim that $C_{\ell,j}$ cannot intersect two distinct level-$k$ associated processing clusters, nor can it intersect both a level-$k$ associated processing cluster and the remaining level-$k$ residual detection-event set.

Suppose otherwise.
Then there exist active detectors $v,v'\in C_{\ell,j}\cap D$ that belong to distinct parts of the level-$k$ residual decomposition.
Since $C_{\ell,j}$ is not a constituent of any $\mathsf{P}(N_{k',h})$ with $k'<k$, none of its active detectors is incident to an edge removed before level $k$.
Indeed, otherwise the induction hypothesis would imply that
$C_{\ell,j}$ is a constituent of the associated processing cluster containing that detector.
Hence, every active detector in $C_{\ell,j}$ is incident to an edge in the residual faulty-edge set $N_k$.
By the level-$k$ buffer condition,
\begin{align}
    d_V(v,v')
    >
    b_k^V.
\end{align}
On the other hand, since $v,v'\in C_{\ell,j}$,
\begin{align}
    d_V(v,v')
    &\leq
    \operatorname{diam}_V(C_{\ell,j})\\
    &\leq
    d_\ell^V\\
    &\leq
    d_k^V\\
    &\leq
    b_k^V,
\end{align}
which is a contradiction.
Thus, every such lower-level processing cluster that intersects $A$ is contained entirely in $A$.

Therefore, the active detectors in $A\setminus R$ are partitioned into processing clusters of levels strictly smaller than $k$.
That is, there exists an index set $\mathcal{I}^{<k}_{k,i}$ such that
\begin{align}
    A\setminus R
    =
    \bigsqcup_{(\ell,j)\in\mathcal{I}^{<k}_{k,i}}
    \left(C_{\ell,j}\cap D\right),
    \qquad
    \ell<k.
\end{align}

If $R=\varnothing$, this decomposition proves the claim.
We therefore assume that $R\neq\varnothing$.

Since $R\subseteq A$ and $\mathsf{P}(N_{k,i})$ satisfy the diameter condition,
\begin{align}
    \operatorname{diam}_V(R)
    \leq
    \operatorname{diam}_V\left(\mathsf{P}(N_{k,i})\right)
    \leq
    d_k^V.
\end{align}
Moreover, let $v\in R$ and $u\in D_k\setminus R$.
Since all active detectors belonging to lower-level associated processing clusters have already been removed, both $u$ and $v$ belong to the level-$k$ residual detection-event set.
Since $v\in A$ and $u\notin A$, the level-$k$ buffer condition gives
\begin{align}
    d_V(u,v)
    >
    b_k^V.
\end{align}
Therefore,
\begin{align}
    d_{C_V}\left(D_k\setminus R,R\right)
    >
    b_k^V.
\end{align}
Thus, $R$ satisfies the geometric part of the $(d_k^V,b_k^V)_V$-cluster condition in $D_k$.

We finally verify the parity condition.
First suppose that $\mathsf{P}(N_{k,i})$ contains no boundary vertex.
The separation of $N_{k,i}$ from the remaining edges implies that
\begin{align}
    \lvert A\rvert
    \equiv
    0
    \pmod 2.
\end{align}
Every lower-level processing cluster contained in $A\setminus R$ also contains no boundary vertex and therefore contains an even number of active detectors.
Hence,
\begin{align}
    \lvert R\rvert
    &=
    \lvert A\rvert
    -
    \sum_{(\ell,j)\in\mathcal{I}^{<k}_{k,i}}
    \left\lvert C_{\ell,j}\cap D\right\rvert\\
    &\equiv
    0
    \pmod 2.
\end{align}

Next suppose that $\mathsf{P}(N_{k,i})$ contains a boundary vertex.
Since $R$ is nonempty and the diameter of $\mathsf{P}(N_{k,i})$ is at most $d_k^V$, every active detector in $R$ is at distance at most $d_k^V$ from the boundary vertex contained in $\mathsf{P}(N_{k,i})$.
The condition $\phi_k>2d_k^V/b_k^V$ in Lemma~\ref{lem:stopping_guarantee_for_processing_cluster}, together with $\phi_k\le 1$, implies
\begin{align}
d_k^V<\frac{b_k^V}{2}.
\end{align}
Hence, the exploring regions of radius $b_k^V/2$ grown from the active detectors in $R$ reach the boundary.
Therefore, the connected component generated from $R$ contains a boundary vertex and satisfies the parity condition.

Thus, $R$ satisfies both the $(d_k^V,b_k^V)_V$-cluster condition and the parity condition.
It is therefore declared to be a level-$k$ processing cluster $C_{k,j}$ satisfying
\begin{align}
    C_{k,j}\cap D
    =
    R.
\end{align}
Combining this processing cluster with the decomposition of $A\setminus R$, we obtain
\begin{align}
    A
    =
    \left(C_{k,j}\cap D\right)
    \sqcup
    \bigsqcup_{(\ell,j')\in\mathcal{I}^{<k}_{k,i}}
    \left(C_{\ell,j'}\cap D\right).
\end{align}
All processing clusters appearing on the right-hand side have level at most $k$.
This completes the induction.
\end{proof}

\begin{lemma}[Correctness of parallel sparse blossom for faulty-edge-based error clusters]
\label{lem:correctness_for_edge_cluster}
Let $D \subseteq V$ be the set of detection events.
For each level-$k$ error cluster $N_{k,i} \subseteq E$, let $\tilde{C}_{k,i}$ be the level-$k$ extended SB cluster generated by $N_{k,i}$ at time $t=0$.
Define the associated processing clusters by $C_{k,i} := \mathsf P(N_{k,i})$.
Assume that the parameters for level-$k$ clusters satisfy the conditions for Lemma~\ref{lem:stopping_guarantee} and Lemma~\ref{lem:stopping_guarantee_for_processing_cluster}.
Then the parallel sparse-blossom execution of the hierarchical execution rule on these processing clusters produces the same matching result, and hence the same decoding result, as the original global sparse-blossom decoder.
\end{lemma}

\begin{proof}
We prove the claim by induction over the level $k$ of the processing clusters.
Throughout the proof, we compare the parallel execution with the original global sparse-blossom execution.

First, consider the case $k=1$.
There are no lower-level processing clusters.
By Lemma~\ref{lem:stopping_guarantee}, each level-$1$ extended SB cluster stops growing before it merges with any level-$k'$ extended SB cluster with $k' \geq 1$.
Therefore, until the level-$1$ extended SB cluster $\tilde{C}_{1,i}$ stops, all sparse-blossom events generated from the active detectors in the associated processing cluster $C_{1,i}=\mathsf P(N_{1,i})$ are internal to this cluster.
Hence, executing sparse blossom on each $C_{1,i}$ independently gives the same sequence of local sparse-blossom updates as the restriction of the global sparse-blossom execution to $\tilde{C}_{1,i}$.
Thus, when $C_{1,i}$ stops, the matching configuration obtained by the parallel execution agrees with the corresponding stopped configuration in the global execution.
This stopped configuration is retained unchanged until it is touched by a subsequent processing-cluster execution.

Now assume that the claim holds for all levels strictly smaller than $k$.
Consider a level-$k$ error cluster $N_{k,i}$ and the level-$k$ extended SB cluster $\tilde{C}_{k,i}$ generated by it.
By Lemma~\ref{lem:stopping_guarantee}, $\tilde{C}_{k,i}$ stops growing before it merges with any level-$k'$ extended SB cluster with $k' \geq k$.
Consequently, before $\tilde{C}_{k,i}$ stops, the sparse-blossom dynamics generated from $C_{k,i}=\mathsf P(N_{k,i})$ cannot depend on the growth of any distinct same-level or higher-level cluster.
Thus, any interaction with another cluster before stopping can only involve a lower-level stopped state.

By the induction hypothesis, every such lower-level stopped matching configuration agrees with the corresponding state in the global sparse-blossom execution.
Moreover, Lemma~\ref{lem:stopping_guarantee} guarantees that the corresponding lower-level extended SB cluster stops before it can merge with the level-$k$ extended SB cluster.
Hence, whenever the level-$k$ execution first touches such a lower-level matching configuration, that configuration has already stopped and has been retained unchanged.
The level-$k$ sparse-blossom dynamics therefore continue from exactly the same intermediate state as in the global execution.

By Lemma~\ref{lem:disjointness_of_inherited_state}, each inherited matching configuration can be touched by at most one level-$k$ processing cluster before stopping.
Therefore, no lower-level stopped state is independently modified by two distinct level-$k$ executions.
If an inherited matching configuration is not touched at level $k$, it remains unchanged.
If it is touched, its active detectors participate in the sparse-blossom dynamics of the unique level-$k$ processing cluster that touches it and are included in the stopped matching configuration produced by that execution.

It follows that the sparse-blossom dynamics associated with $C_{k,i}$ coincide with the corresponding dynamics in the global sparse-blossom execution until $C_{k,i}$ stops.
In particular, the stopped matching configuration produced by $C_{k,i}$ agrees with the corresponding stopped state in the global execution.

By induction, this agreement holds for all levels and all processing clusters.
Throughout the hierarchical execution, every stopped matching configuration is either retained unchanged or, when touched, incorporated into the stopped matching configuration produced by the unique processing-cluster execution that touches it.
Hence no stopped state is lost or duplicated.
After all processing-cluster executions have stopped, the remaining terminal matching configurations therefore collectively represent the same final matching state as the original global sparse-blossom execution.
Combining them gives the same matching result, and hence the same decoding result, as the original global sparse-blossom decoder.
\end{proof}

\begin{lemma}[Equivalence under hierarchical decomposition]
\label{lem:equivalence_under_decomposition}
Let $N_{k,i}$ be a level-$k$ error cluster, and consider its associated processing cluster $\mathsf{P}(N_{k,i})$.
Let $\mathcal{I}_{k,i}$ denote the index set defined in Lemma~\ref{lem:covering_processing_clusters}, and define
\begin{align}
    \operatorname{Dec}(\mathsf{P}(N_{k,i}))
    :=
    \left\{
        C_{k',j}
        \,\middle|\,
        (k',j)\in\mathcal{I}_{k,i}
    \right\}.
\end{align}

Assume that the clustering parameters satisfy the conditions of Lemma~\ref{lem:stopping_guarantee_for_processing_cluster}.
Then the direct sparse-blossom execution on $\mathsf{P}(N_{k,i})$ produces the same final matching result as the hierarchical execution of $\operatorname{Dec}(\mathsf{P}(N_{k,i}))$.
\end{lemma}

\begin{proof}
By Lemma~\ref{lem:covering_processing_clusters}, the active detectors in $\mathsf{P}(N_{k,i})$ are partitioned into the active-detector sets of the processing clusters in $\operatorname{Dec}(\mathsf{P}(N_{k,i}))$.
Hence, the execution based on the decomposition contains exactly the same initial active detectors as the direct execution on $\mathsf{P}(N_{k,i})$, without omission or duplication.

Take a processing cluster in $\operatorname{Dec}(\mathsf{P}(N_{k,i}))$.
By Lemma~\ref{lem:stopping_guarantee_for_processing_cluster}, its execution stops before interacting with any processing cluster of the same or a higher level.
Therefore, its execution up to its stopping time is independent of the remaining dynamics in $\mathsf{P}(N_{k,i})$.
Executing this processing cluster independently to its stopping time is thus equivalent to executing the corresponding part of the direct execution on $\mathsf{P}(N_{k,i})$ to the same stopped matching configuration.

After the processing cluster stops, this matching configuration is retained unchanged until it is touched by a subsequent higher-level processing-cluster execution.
If it is not touched, it continues to represent exactly the same stopped state as in the direct execution.
If it is touched, the subsequent sparse-blossom dynamics continue from that stopped state, and the active detectors belonging to the touched configuration participate in the resulting dynamics.
Thus, whenever two previously independent parts of the execution begin to interact, the hierarchical execution starts that interaction from exactly the same intermediate states as the direct execution.

Moreover, by Lemma~\ref{lem:disjointness_of_inherited_state}, an inherited matching configuration can be touched by at most one processing cluster at any given level before stopping.
Hence, no stopped state is independently modified by two distinct same-level executions.

It follows that the hierarchical execution differs from the direct execution on $\mathsf{P}(N_{k,i})$ only in the order in which mutually independent parts of the sparse-blossom dynamics are executed.
All interactions occur from the same intermediate matching configurations as in the direct execution.
Therefore, after all processing clusters in $\operatorname{Dec}(\mathsf{P}(N_{k,i}))$ have completed their hierarchical execution, combining the resulting terminal matching configurations gives exactly the same final matching result as executing sparse blossom directly on $\mathsf{P}(N_{k,i})$.
\end{proof}

\section{Correctness of parallel sparse blossom}\label{sup:correctness_of_theorem_1}

Here, we recall the correctness theorem for our parallel sparse-blossom algorithm.
\begin{theorem}[Correctness of parallel sparse blossom]\label{thm:supp_parallel_correctness}
    Fix a detector graph $G=(V,E)$, edge weights $w$, and a nonempty set of observed detection events $D\subseteq V$.
    Then, there is an algorithm that constructs a hierarchical clustering of $D$ such that sparse blossom can be run independently and in parallel on the clusters at each level, and the resulting decoder produces the same decoding result as the original sparse-blossom decoder run globally.
\end{theorem}

Choose the clustering parameters so that the conditions of Lemmas~\ref{lem:stopping_guarantee} and~\ref{lem:stopping_guarantee_for_processing_cluster} are satisfied; the existence of such a parameter choice is verified in the next section.
We now combine Lemmas~\ref{lem:correctness_for_edge_cluster} and~\ref{lem:equivalence_under_decomposition} to prove Theorem~\ref{thm:supp_parallel_correctness}.
The key point is that the processing-cluster hierarchy used by the algorithm is constructed solely from the observed detection events, without access to the underlying faulty-edge configuration.

\begin{proof}[Proof of Theorem~\ref{thm:supp_parallel_correctness}]
Let $D\subseteq V$ be the observed detection-event set, and let
\begin{align}
    \mathcal{C}:=\{C_{k,i}\}_{k,i}
\end{align}
denote the hierarchy of processing clusters constructed from $D$ by the processing-cluster construction described above.
Let $N\subseteq E$ be the unobserved faulty-edge set that induces $D$, and consider its error-cluster hierarchy.
Let
\begin{align}
    \mathsf{P}(N):=\{\mathsf{P}(N_{k,i})\}_{k,i}
\end{align}
denote the corresponding hierarchy of associated processing clusters.

For a hierarchy of processing clusters $\mathcal H$, let $\operatorname{Exec}(\mathcal H)$ denote the final matching result obtained by combining the terminal matching configurations produced by its hierarchical sparse-blossom execution, and let $\operatorname{Exec}_{\mathrm{global}}(D)$ denote the matching result produced by global sparse blossom on $D$.

Applying Lemma~\ref{lem:equivalence_under_decomposition} to each associated processing cluster $\mathsf{P}(N_{k,i})$, we can replace its direct execution by the hierarchical execution of its decomposition $\operatorname{Dec}(\mathsf{P}(N_{k,i}))$ using processing clusters from $\mathcal{C}$, without changing the resulting matching result.
Therefore,
\begin{align}
    \operatorname{Exec}(\mathcal{C})
    =
    \operatorname{Exec}\left(\mathsf{P}(N)\right).
\end{align}

By Lemma~\ref{lem:correctness_for_edge_cluster} and Definition~\ref{def:associated_processing_cluster}, the parallel sparse-blossom execution based on the associated processing clusters of $N$ produces the same matching result as the original global sparse-blossom execution.
Hence,
\begin{align}
    \operatorname{Exec}\left(\mathsf{P}(N)\right)
    =
    \operatorname{Exec}_{\mathrm{global}}(D).
\end{align}

Combining the two equalities gives
\begin{align}
    \operatorname{Exec}(\mathcal{C})
    =
    \operatorname{Exec}\left(\mathsf{P}(N)\right)
    =
    \operatorname{Exec}_{\mathrm{global}}(D).
\end{align}
Therefore, although the faulty-edge set $N$ is not available to the decoder, the processing-cluster hierarchy constructed solely from the observed detection events $D$ produces exactly the same matching result, and hence the same decoding result, as global sparse blossom.
\end{proof}

\section{Verification of the asymptotic parameter choice}\label{sup:parameter_verification}
To satisfy the conditions of Lemma~\ref{lem:stopping_guarantee} and Lemma~\ref{lem:stopping_guarantee_for_processing_cluster}, we take the following parameter sequences with super-exponential separation:
\begin{align}
b_k^E = \beta \lambda^{(k+1)\log(k+1)} + w_{\max},\quad
d_k^E = \gamma \lambda^{k\log k} - w_{\max}, 
\end{align}
or, by the conversion in Lemma~\ref{lem:associated_processing_cluster_is_processing_cluster},
\begin{align}
b_k^V = \beta \lambda^{(k+1)\log(k+1)},\quad
d_k^V = \gamma \lambda^{k\log k}.
\end{align}
This corresponds to
\begin{align}
f(k)=k\log k
\end{align}
in Lemma~\ref{lem:threshold_theorem}, which satisfies Eq.~\eqref{eq:supp_f_condition} with the parameter choice of $c$ and its derivation described below.
For $f(k)=k\log k$, we have
\begin{align}
\sum_{j=0}^{k-2} f(k-j)2^j
&=
2^{k-1}
\sum_{n=2}^{k}\frac{f(n)}{2^{n-1}}
\\
&\leq
c\,2^{k-1},
\end{align}
where
\begin{align}
c
:=
\sum_{n=2}^{\infty}
\frac{n\log n}{2^{n-1}}
\approx 3.57257.
\end{align}

Assume also that
\begin{align}
\lambda>e,\quad
\frac{12\gamma}{\beta}\bigl[\zeta(\log\lambda)-1\bigr]\le 1.\label{eq:constraint_0}
\end{align}
Then $\phi_k$ satisfies $1 \ge \phi_k \ge \frac12$, which is shown by induction on $k$ as follows.

For $k = 1$, $\phi_1 = 1$.
For $k > 1$, if $\phi_{k'} \ge \frac12$ holds for all $k' < k$, we have
\begin{align}
    \phi_k
    &=
    1 - \sum_{k''=1}^{k-1}\frac{(\phi_{k''} + 2)(d_{k''}^E + w_{\max})}{(\phi_{k''} + 1)(d_{k''}^E + w_{\max}) + \phi_{k''}(b_{k''}^E - w_{\max})} \\
    &\ge
    1 - \sum_{k''=1}^{k-1}\frac{3(d_{k''}^E + w_{\max})}{\phi_{k''}(b_{k''}^E - w_{\max})} \\
    &\ge
    1 - \sum_{k''=1}^{k-1}\frac{6\gamma \lambda^{k''\log k''}}{\beta \lambda^{(k''+1)\log(k''+1)}} \\
    &=
    1 - \frac{6\gamma}{\beta}
    \left(
    \sum_{k''=1}^{k-1}\lambda^{-\left((k'' + 1)\log (k'' + 1) - k''\log k''\right)}
    \right).
\end{align}
Here,
\begin{align}
    -\left((k'' + 1)\log (k'' + 1) - k''\log k''\right)
    &= -\left(\log (k'' + 1) + k''\left(\log\left(1 + \frac{1}{k''}\right)\right)\right) \\
    &\le -\log (k'' + 1)
\end{align}
holds.
Then we have
\begin{align}
    \phi_k
    &\ge
    1 - \frac{6\gamma}{\beta}
    \left(
    \sum_{k''=1}^{k-1}\lambda^{-\left((k'' + 1)\log (k'' + 1) - k''\log k''\right)}
    \right)\\
    &\ge
    1 - \frac{6\gamma}{\beta}
    \left(
    \sum_{k''=1}^{k-1}\lambda^{-\log (k'' + 1)}
    \right)\\
    &=
    1 - \frac{6\gamma}{\beta}
    \left(
    \sum_{k''=2}^{k}\lambda^{-\log k''}
    \right)\\
    &\ge
    1 - \frac{6\gamma}{\beta}
    \left(
    \sum_{k''=2}^{\infty}\lambda^{-\log k''}
    \right)\\
    &=
    1 - \frac{6\gamma}{\beta}
    \left(
    \sum_{k''=2}^{\infty}k''^{-\log \lambda}
    \right) \\
    &=
    1 - \frac{6\gamma}{\beta}[\zeta(\log \lambda) - 1] \\
    &\ge \frac12.
\end{align}
As a consequence,
\begin{align}
m_k^E \le \frac{d_k^E+w_{\max}}{\phi_k} = O(d_k^E).
\end{align}

In addition, to satisfy the assumptions of Lemma~\ref{lem:stopping_guarantee}, Lemma~\ref{lem:stopping_guarantee_for_processing_cluster}, and Lemma~\ref{lem:super_exponential_bound_p_k}, we consider the following formulas:
\begin{align}
b_k^E \ge d_k^E > 0, \quad d_{k+1}^E \ge 4b_k^E + 3d_k^E, \quad \phi_k > 2\frac{d_k^E+w_{\max}}{b_k^E-w_{\max}} > 0.
\end{align}
Since $\phi_k \ge \frac12$, it is enough to consider the following conditions:
\begin{align}
    &b_k^E \ge d_k^E > 0, \quad d_{k+1}^E \ge 4b_k^E + 3d_k^E,\quad \frac{1}{4} > \frac{d_k^E+w_{\max}}{b_k^E-w_{\max}} > 0 \\
    \Leftarrow \quad 
    &d_k^E > 0, \quad d_{k+1}^E \ge 4b_k^E + 3d_k^E,\quad b_k^E - w_{\max} >4(d_k^E + w_{\max}).\label{eq:all_constraints}
\end{align}
A sufficient set of conditions for the constraints above is
\begin{align}
    \gamma &> w_{\max},\label{eq:constraint_1} \\
    \gamma &\ge \frac{3\gamma}{\lambda} + 4\beta + 2w_{\max}\lambda^{-2\log 2},\label{eq:constraint_2} \\
    \beta\lambda &> 4\gamma,\label{eq:constraint_3}
\end{align}
where we use the fact that $\left((k + 1)\log (k + 1) - k\log k\right)$ is monotonically increasing.

We derive an explicit lower bound on the threshold $p_{\mathrm{th}}$ in Eq.~\eqref{eq:supp_threshold}.
We numerically search over the ratios $\beta/w_{\max}$ and $\gamma/w_{\max}$ and over $\lambda$, subject to the constraints in Eqs.~\eqref{eq:constraint_0}, \eqref{eq:constraint_1}, \eqref{eq:constraint_2}, and \eqref{eq:constraint_3}, and obtain the parameter choice of
\begin{align}
    \beta = 72  w_{\max},\quad \gamma = 291 w_{\max},\quad \lambda =320.\label{eq:choice_of_parameters}
\end{align}
These values satisfy all the required constraints simultaneously.

Since we take $f(k) = k \log k$, we have
\begin{align}
    p_{\mathrm{th}} 
    &=
    \frac{1}{
    \xi A \lambda^{c\Delta}
    \left(
    1 + 
    \Lambda
    \left(
    \beta+\dfrac{\gamma+w_{\max}}{2}
    \right)
    \right)^{\Delta}}\\
    &\ge
    \frac{1}{
    \xi A \lambda^{c\Delta}
    \left(
    1 + 
    \Lambda
    \left(
    \beta+\gamma
    \right)
    \right)^{\Delta}} \quad (\because \gamma > w_{\max}).
\end{align}
Recall that $\beta=72 w_{\max}$, $\gamma=291 w_{\max}$ and $\Lambda = 2/w_{\min}$, and define the constant ratio $\theta := w_{\max}/w_{\min}$.
Then we have
\begin{align}
    p_{\mathrm{th}}
    &\ge
    \frac{1}{
    \xi A \lambda^{c\Delta}
    \left(
    1 + 
    2\theta(\beta'+\gamma')
    \right)^{\Delta}},
\end{align}
where $\beta'=72, \gamma'=291$.

Using
\begin{align}
    A=12\sqrt{3}\pi,\qquad
    \Delta=3,\qquad
    c\simeq3.57257,\qquad
    \lambda=320,
\end{align}
and, for illustration, taking $\xi\leq10$ and $\theta=2$, we obtain
\begin{align}
    p_{\mathrm{th}} \gtrsim 7 \times 10^{-40}.
\end{align}
This bound is not intended to provide a practical estimate of the decoding threshold.
Rather, it demonstrates that the constants in the clustering argument can be chosen consistently so that the rigorous lower bound on the threshold is strictly positive.

\section{Quasi-polylogarithmic average runtime}\label{sup:correctness_of_theorem_2}
Here, we recall the theorem that guarantees the quasi-polylogarithmic average runtime of our parallel sparse-blossom algorithm.

\begin{theorem}[Quasi-polylogarithmic average runtime of parallel sparse blossom]\label{thm:supp_parallel_ave_time}
For the $[[n=d^2,1,d]]$ rotated surface code under the local stochastic error model with a physical error rate $p$ below a finite threshold $p_{\mathrm{th}}$, the average runtime of the parallel sparse-blossom decoder for $O(d)$ rounds of syndrome extraction is upper bounded by
\begin{align}
\exp\bigl[O((\log\log d)(\log\log\log d))\bigr].
\end{align}
\end{theorem}

\begin{proof}
Throughout this proof, parallel runtime is measured in the classical circuit-depth model in which each elementary operation acts on $O(1)$ bits, independent operations are executed concurrently, and a polynomial amount of parallel computational resources is available.
We use the asymptotic clustering parameters verified in Sec.~\ref{sup:parameter_verification}.
Let
\begin{align}
    \rho
    :=
    \frac{p}{p_{\mathrm{th}}}
    <
    1.
\end{align}
The error-clustering bound gives
\begin{align}
    p_k
    =
    O\!\left(
        \rho^{2^{k-1}}
    \right),
\end{align}
where
\begin{align}
p_k:=\max_{e\in E}\Pr[e\in\mathrm{N}_k]
\end{align}
is the maximum probability that a fixed edge remains in the residual edge set at level $k$.

Let $\mathcal{P}_k$ denote the event that at least one level-$k$ processing cluster appears in the decoding window.
Suppose that $\mathcal{P}_k$ occurs, and choose an active detector $v$ in a level-$k$ processing cluster.
Because $v$ is induced by the underlying faulty-edge set $N$, it is incident to at least one faulty edge belonging to some level-$k'$ error cluster $N_{k',i}$.
Therefore,
\begin{align}
    v
    \in
    \mathsf{P}(N_{k',i}).
\end{align}
By Lemma~\ref{lem:covering_processing_clusters}, every active detector in $\mathsf{P}(N_{k',i})$ belongs to a processing cluster of level at most $k'$.
Since each active detector is assigned to a unique processing cluster, we have
\begin{align}
    k'\geq k.
\end{align}
Since $N_{k',i}\subseteq\mathrm{N}_{k'}$ and the residual edge sets are nested as $\mathrm{N}_{k'}\subseteq\mathrm{N}_k$ for $k'\geq k$, the occurrence of $\mathcal{P}_k$ implies that $\mathrm{N}_k$ is nonempty.

Since a decoding window for $O(d)$ rounds contains
\begin{align}
    |E|=O(d^3)
\end{align}
edges, the union bound gives
\begin{align}
    \Pr[\mathcal{P}_k]
    &\leq
    \Pr[\mathrm{N}_k\neq\varnothing] \\
    &\leq
    \sum_{e\in E}\Pr[e\in\mathrm{N}_k] \\
    &\leq
    |E|p_k \\
    &\leq
    Cd^3\rho^{2^{k-1}}
\end{align}
for some constant $C>0$ independent of $k$ and $d$.
Therefore,
\begin{align}
    \Pr[\mathcal{P}_k]
    \leq
    \min\left\{
        Cd^3\rho^{2^{k-1}},
        1
    \right\}.
    \label{eq:processing_cluster_probability_bound}
\end{align}

By Lemma~\ref{lem:stopping_guarantee_for_processing_cluster}, a level-$k$ executing SB cluster has diameter $O(d_k^V)$ and grows for at most $O(d_k^V)$ steps.
Together with the locality bound of the surface-code detector graph, this implies that the number of detector-graph vertices and edges contained in such a cluster is $O((d_k^V)^3)$.
Therefore, the time required to process a level-$k$ processing cluster is bounded by
\begin{align}
O\!\left(\operatorname{poly}(d_k^V)\right)
\leq
\lambda^{\tau_d k\log k}
\end{align}
for some sufficiently large constant $\tau_d$.

All level-$k$ processing clusters are executed in parallel.
Bounding the contributions of different levels sequentially can only overestimate the actual execution time, since clusters at different levels may also overlap in time.
Therefore,
\begin{align}
    T_{\mathrm{exec}}
    \leq
    \sum_{k=1}^{\infty}
    \lambda^{\tau_d k\log k}
    \mathbf{1}_{\mathcal{P}_k},
\end{align}
where $\mathbf{1}_{\mathcal{P}_k}$ denotes the indicator function of
$\mathcal{P}_k$.
Taking the expectation gives
\begin{align}
    \mathbb{E}[T_{\mathrm{exec}}]
    &\leq
    \sum_{k=1}^{\infty}
    \lambda^{\tau_d k\log k}
    \Pr[\mathcal{P}_k]\\
    &\leq
    \sum_{k=1}^{\infty}
    \lambda^{\tau_d k\log k}
    \min\left\{
        Cd^3\rho^{2^{k-1}},
        1
    \right\}.
    \label{eq:parallel_execution_runtime_sum}
\end{align}

We next analyze the processing-cluster construction.
Let $\mathcal{R}_k$ denote the event that the residual detection-event set is nonempty at the beginning of level $k$.
Any such residual active detector must eventually be assigned to a processing cluster of some level $k'\geq k$.
By the argument above, the occurrence of such a processing cluster implies that $\mathrm{N}_{k'}$ is nonempty.
Since $\mathrm{N}_{k'}\subseteq\mathrm{N}_k$ for $k'\geq k$, the event $\mathcal{R}_k$ implies that $\mathrm{N}_k$ is nonempty.
Therefore, by the same union-bound argument,
\begin{align}
    \Pr[\mathcal{R}_k]
    \leq
    \min\left\{
        Cd^3\rho^{2^{k-1}},
        1
    \right\}.
    \label{eq:clustering_reaches_level_probability}
\end{align}

At level $k$, the exploring regions are grown only up to distance $O(b_k^V)$.
The connected-component construction and the remaining componentwise operations can be performed in polylogarithmic parallel depth in the size of the detector graph using a polynomial amount of parallel computational resources.
Since the detector graph for an $O(d)$-round decoding window has polynomial size in $d$, this contributes a factor of $\operatorname{polylog}(d)$.
Thus, the parallel time required for the level-$k$ clustering step is bounded by
\begin{align}
    O\!\left(
        \operatorname{poly}(b_k^V)\operatorname{polylog}(d)
    \right)
    \leq
    \operatorname{polylog}(d)\lambda^{\tau_b k\log k}
\end{align}
for some sufficiently large constant $\tau_b$.
It follows that
\begin{align}
    \mathbb{E}[T_{\mathrm{cluster}}]
    &\leq
    \operatorname{polylog}(d)
    \sum_{k=1}^{\infty}
    \lambda^{\tau_b k\log k}
    \Pr[\mathcal{R}_k]\\
    &\leq
    \operatorname{polylog}(d)
    \sum_{k=1}^{\infty}
    \lambda^{\tau_b k\log k}
    \min\left\{
        Cd^3\rho^{2^{k-1}},
        1
    \right\}.
    \label{eq:parallel_clustering_runtime_sum}
\end{align}

It remains to evaluate the sums in
Eqs.~\eqref{eq:parallel_execution_runtime_sum} and
\eqref{eq:parallel_clustering_runtime_sum}.
Let
\begin{align}
    \tau
    :=
    \max\left\{
        \tau_d,
        \tau_b
    \right\}
\end{align}
and define
\begin{align}
    S_{\tau}(d)
    :=
    \sum_{k=1}^{\infty}
    \lambda^{\tau k\log k}
    \min\left\{
        Cd^3\rho^{2^{k-1}},
        1
    \right\}.
\end{align}
Define
\begin{align}
    \overline{k}
    :=
    \max\left(
        \left\{
            k\geq 1
            \,\middle|\,
            Cd^3\rho^{2^{k-1}}
            \geq
            1
        \right\}
        \cup
        \{0\}
    \right).
\end{align}
Taking logarithms twice shows that
\begin{align}
    \overline{k}
    =
    O(\log\log d).
\end{align}

We split the sum as
\begin{align}
    S_{\tau}(d)
    &\leq
    \sum_{k=1}^{\overline{k}}
    \lambda^{\tau k\log k}+
    Cd^3
    \sum_{k=\overline{k}+1}^{\infty}
    \lambda^{\tau k\log k}
    \rho^{2^{k-1}}.
    \label{eq:split_parallel_runtime_sum}
\end{align}
The first sum is bounded by
\begin{align}
    \sum_{k=1}^{\overline{k}}
    \lambda^{\tau k\log k}
    &\leq
    \overline{k}
    \lambda^{\tau\overline{k}\log\overline{k}}\\
    &=
    \exp\!\left[
        O\!\left(
            (\log\log d)(\log\log\log d)
        \right)
    \right].
    \label{eq:first_parallel_runtime_sum}
\end{align}

For the second sum, define
\begin{align}
    a_k
    :=
    \lambda^{\tau k\log k}
    \rho^{2^{k-1}}.
\end{align}
We have
\begin{align}
    \frac{a_{k+1}}{a_k}
    =
    \lambda^{
        \tau\left(
            (k+1)\log(k+1)-k\log k
        \right)
    }
    \rho^{2^{k-1}},
\end{align}
and therefore
\begin{align}
    \lim_{k\to\infty}
    \frac{a_{k+1}}{a_k}
    =
    0.
\end{align}
For any fixed $0<\varepsilon<1$, there consequently exists a constant $k_0$ such that
\begin{align}
    \frac{a_{k+1}}{a_k}
    \leq
    1-\varepsilon
\end{align}
for all $k\geq k_0$.
The constant $k_0$ is independent of $d$.
For sufficiently large $d$, we have
$\overline{k}+1\geq k_0$, and hence
\begin{align}
    Cd^3
    \sum_{k=\overline{k}+1}^{\infty}a_k
    &\leq
    Cd^3a_{\overline{k}+1}
    \sum_{j=0}^{\infty}(1-\varepsilon)^j\\
    &=
    \frac{
        Cd^3
        \rho^{2^{\overline{k}}}
    }{\varepsilon}
    \lambda^{
        \tau(\overline{k}+1)
        \log(\overline{k}+1)
    }.
\end{align}
By the definition of $\overline{k}$,
\begin{align}
    Cd^3\rho^{2^{\overline{k}}}
    <
    1.
\end{align}
Therefore,
\begin{align}
    Cd^3
    \sum_{k=\overline{k}+1}^{\infty}a_k
    &\leq
    \frac{1}{\varepsilon}
    \lambda^{
        \tau(\overline{k}+1)
        \log(\overline{k}+1)
    }\\
    &=
    \exp\!\left[
        O\!\left(
            (\log\log d)(\log\log\log d)
        \right)
    \right].
    \label{eq:second_parallel_runtime_sum}
\end{align}
For the finitely many smaller values of $d$, the same bound follows by increasing the constant implicit in the $O$ notation.

Combining Eqs.~\eqref{eq:first_parallel_runtime_sum} and
\eqref{eq:second_parallel_runtime_sum}, we obtain
\begin{align}
    S_{\tau}(d)
    \leq
    \exp\!\left[
        O\!\left(
            (\log\log d) (\log\log\log d)
        \right)
    \right].
\end{align}

Since
\begin{align}
    \mathbb{E}[T_{\mathrm{exec}}]
    &\leq S_\tau(d),\\
    \mathbb{E}[T_{\mathrm{cluster}}]
    &\leq \operatorname{polylog}(d)S_\tau(d),
\end{align}
and
\begin{align}
    \operatorname{polylog}(d)
    =
    \exp[O(\log\log d)],
\end{align}
the additional polylogarithmic factor is absorbed into the quasi-polylogarithmic bound.
Consequently,
\begin{align}
    \mathbb{E}[T_{\mathrm{cluster}}+T_{\mathrm{exec}}]
    \leq
    \exp\!\left[
        O\!\left(
            (\log\log d)(\log\log\log d)
        \right)
    \right].
\end{align}

This proves Theorem~\ref{thm:supp_parallel_ave_time}.
\end{proof}

\section{Average runtime of non-parallel sparse blossom}
\label{sup:nonparallel_average_runtime}
We also analyze the average runtime of non-parallel sparse blossom.

\begin{theorem}[$O(d^3)$ average runtime of non-parallel sparse blossom]
\label{thm:nonparallel_average_runtime}
For the $[[n=d^2,1,d]]$ rotated surface code under the local stochastic error model with a physical error rate $p$ below a finite threshold $p_{\mathrm{th}}$, the average runtime of the non-parallel sparse-blossom decoder for $O(d)$ rounds of syndrome extraction is upper bounded by
\begin{align}
    O(d^3).
\end{align}
\end{theorem}

\begin{proof}
We use the same asymptotic parameter choice as in Sec.~\ref{sup:parameter_verification}.
A level-$k$ extended SB cluster has diameter
\begin{align}
    O\!\left(
        d_k^E+w_{\max}
    \right)
    =
    O(d_k^V)
\end{align}
and grows for at most $O(d_k^V)$ steps.
Together with the locality bound of the surface-code detector graph, its volume is $O((d_k^V)^3)$.
Therefore, the time required to process each level-$k$ extended SB cluster is bounded by
\begin{align}
    O\!\left(
        \operatorname{poly}(d_k^V)
    \right)
    \leq
    \lambda^{\tau k\log k}
\end{align}
for some sufficiently large constant $\tau$.

Let $Z_k$ denote the number of level-$k$ extended SB clusters that appear in the decoding window.
Each level-$k$ extended SB cluster contains at least one level-$k$ error cluster, and distinct extended SB clusters contain distinct level-$k$ error clusters.
The number of level-$k$ error clusters is at most the number of edges belonging to such clusters.
Moreover, every edge belonging to a level-$k$ error cluster lies in the residual edge set $\mathrm{N}_k$.
Consequently,
\begin{align}
    \mathbb{E}[Z_k]
    &\leq
    \sum_{e\in E}
    \Pr[
        e\text{ belongs to a level-$k$ error cluster}
    ]\\
    &\leq
    \sum_{e\in E}\Pr[e\in\mathrm{N}_k]\\
    &\leq
    |E|p_k\\
    &=
    O(d^3)
    O\!\left(
        \rho^{2^{k-1}}
    \right).
\end{align}
where
\begin{align}
    \rho
    :=
    \frac{p}{p_{\mathrm{th}}}
    <
    1.
\end{align}

Since the clusters are processed sequentially in the non-parallel setting, the total runtime is bounded by
\begin{align}
    T_{\mathrm{nonparallel}}
    \leq
    \sum_{k=1}^{\infty}
    Z_k
    \lambda^{\tau k\log k}.
\end{align}
Taking the expectation gives
\begin{align}
    \mathbb{E}[T_{\mathrm{nonparallel}}]
    &\leq
    \sum_{k=1}^{\infty}
    \lambda^{\tau k\log k}
    \mathbb{E}[Z_k]\\
    &\leq
    O(|E|)
    \sum_{k=1}^{\infty}
    \lambda^{\tau k\log k}
    \rho^{2^{k-1}}.
\end{align}

Define
\begin{align}
    a_k
    :=
    \lambda^{\tau k\log k}
    \rho^{2^{k-1}}.
\end{align}
As in the proof above,
\begin{align}
    \frac{a_{k+1}}{a_k}
    =
    \lambda^{
        \tau\left(
            (k+1)\log(k+1)-k\log k
        \right)
    }
    \rho^{2^{k-1}}
    \longrightarrow
    0.
\end{align}
Therefore, for any fixed $0<\varepsilon<1$, there exists a constant $k_0$, independent of $d$, such that
\begin{align}
    \frac{a_{k+1}}{a_k}
    \leq
    1-\varepsilon
\end{align}
for all $k\geq k_0$.
It follows that
\begin{align}
    \sum_{k=1}^{\infty}a_k
    &=
    \sum_{k=1}^{k_0}a_k
    +
    \sum_{k=k_0+1}^{\infty}a_k\\
    &\leq
    \sum_{k=1}^{k_0}a_k
    +
    a_{k_0}
    \sum_{j=1}^{\infty}(1-\varepsilon)^j\\
    &=
    \sum_{k=1}^{k_0}a_k
    +
    a_{k_0}
    \frac{1-\varepsilon}{\varepsilon}\\
    &=
    O(1).
\end{align}
Here, $k_0$ and every $a_k$ are independent of $d$.
Therefore, the average runtime is
\begin{align}
    \mathbb{E}[T_{\mathrm{nonparallel}}]
    &=
    O(|E|)\\
    &=
    O(d^3).
\end{align}
\end{proof}

When the decoding window contains $d$ rounds, the average runtime per round is upper bounded by 
\begin{align}
    \frac{O(d^3)}{d}
    &=
    O(d^2).
\end{align}

\section{Analysis of the finite-size parameter schedule}\label{sup:analysis_of_scheduled_parameters}

We first recall the finite-size parameter schedule used in the numerical evaluation.
Fix $0<\phi_{\min}<1$ and $0<q<1$, and define
\begin{align}
    \overline{\phi}_k
    :=
    \phi_{\min}
    +(1-\phi_{\min})q^{k-1}.
\end{align}
Set
\begin{align}
    d_1^V:=w_{\max}+1,
    \qquad
    \phi_1:=1.
\end{align}
For each level $k$, define
\begin{align}
    \Phi_{k+1}(b)
    :=
    \phi_k
    -
    \frac{(\phi_k+2)d_k^V}
    {(\phi_k+1)d_k^V+\phi_k b},
\end{align}
and choose $b_k^V$ as the smallest positive integer satisfying
\begin{align}
    b_k^V&\ge d_k^V,\\
    \phi_k&>2\frac{d_k^V}{b_k^V},\\
    \Phi_{k+1}(b_k^V)&\ge\overline{\phi}_{k+1}.
\end{align}
Then set
\begin{align}
    \phi_{k+1}
    &:=
    \Phi_{k+1}(b_k^V),\\
    d_{k+1}^V
    &:=
    3d_k^V+4b_k^V+2w_{\max}.
\end{align}

\begin{proposition}[Well-definedness of the finite-size parameter schedule]
The recursively generated parameter sequences satisfy
\begin{align}
    \phi_k
    \ge
    \overline{\phi}_k
    \ge
    \phi_{\min}
\end{align}
for every $k\ge1$.
Moreover, the set of admissible integer values of $b_k^V$ is nonempty at every level, and hence the smallest admissible value exists.
\end{proposition}

\begin{proof}
We first prove $\phi_k\ge\overline{\phi}_k$ by induction on $k$.
For $k=1$, we have
\begin{align}
    \phi_1
    =
    1
    =
    \overline{\phi}_1.
\end{align}

Suppose that
$\phi_k\ge\overline{\phi}_k$.
By construction, $b_k^V$ is chosen so that
\begin{align}
    \phi_{k+1}
    =
    \Phi_{k+1}(b_k^V)
    \ge
    \overline{\phi}_{k+1}.
\end{align}
Thus, the claim follows by induction.
Since $\overline{\phi}_k\ge\phi_{\min}$, we obtain
\begin{align}
    \phi_k\ge\phi_{\min}>0
\end{align}
for every $k$.

It remains to show that such a finite $b_k^V$ exists.
For fixed $d_k^V$ and $\phi_k>0$, the function $\Phi_{k+1}(b)$ is monotonically increasing in $b$ and satisfies
\begin{align}
    \lim_{b\to\infty}\Phi_{k+1}(b)
    =
    \phi_k.
\end{align}
By the induction hypothesis,
\begin{align}
    \phi_k
    \ge
    \overline{\phi}_k
    >
    \overline{\phi}_{k+1}.
\end{align}
Therefore, a sufficiently large finite $b$ satisfies
\begin{align}
    \Phi_{k+1}(b)
    \ge
    \overline{\phi}_{k+1}.
\end{align}
The remaining conditions $b\ge d_k^V$ and $\phi_k>2d_k^V/b$ also hold for all sufficiently large $b$.
Hence, a finite admissible value of $b_k^V$ exists at every level.
\end{proof}

\begin{lemma}[Upper bound on the scheduled parameters]
    The recursively constructed parameter sequences of $\{d_k^V\}_{k\ge1}$ and $\{b_k^V\}_{k\ge1}$ generated by the finite-size parameter schedule defined above grow at most as $\exp[O(k^2)]$.\label{lem:order_of_scheduled_parameters}
\end{lemma}

\begin{proof}
Recall that the target sequence is defined by
\begin{align}
    \overline{\phi}_k
    :=
    \phi_{\min}
    +
    (1-\phi_{\min})q^{k-1},
    \quad
    0<\phi_{\min}<1,
    \quad
    0<q<1.
\end{align}
The difference between two consecutive target values is
\begin{align}
    \delta_k
    &:=
    \overline{\phi}_k
    -
    \overline{\phi}_{k+1}
    \\
    &=
    (1-\phi_{\min})(1-q)q^{k-1}.
\end{align}

By the recursive definition of $\phi_k$, we have
\begin{align}
    \phi_{k+1}
    =
    \phi_k
    -
    \frac{(\phi_k+2)d_k^V}
    {(\phi_k+1)d_k^V+\phi_k b_k^V}.
    \label{eq:phi_recursive_experimental}
\end{align}
Since
\begin{align}
    \phi_{\min}
    \le
    \phi_k
    \le
    1,
\end{align}
the decrease of $\phi_k$ at level $k$ is bounded by
\begin{align}
    \phi_k-\phi_{k+1}
    &=
    \frac{(\phi_k+2)d_k^V}
    {(\phi_k+1)d_k^V+\phi_k b_k^V}
    \\
    &\le
    \frac{3d_k^V}
    {\phi_{\min}b_k^V}.
    \label{eq:phi_decrease_upper_bound}
\end{align}

Consider a candidate buffer parameter of the form
\begin{align}
    \hat b_k^V
    :=
    \left\lceil
    Cq^{-(k-1)}d_k^V
    \right\rceil,
    \label{eq:candidate_buffer}
\end{align}
where $\lceil \cdot\rceil$ denotes the ceiling function and $C>0$ is a sufficiently large constant independent of $k$.
In particular, choose $C$ so that
\begin{align}
    C
    >
    \max
    \left\{
        1,\,
        \frac{2}{\phi_{\min}},\,
        \frac{3}{\phi_{\min}(1-\phi_{\min})(1-q)}
    \right\}.
    \label{eq:C_condition}
\end{align}
Then, we have
\begin{align}
    \hat b_k^V > q^{-(k-1)}d_k^V, \quad 
    \hat b_k^V > \frac{2}{\phi_{\min}}q^{-(k-1)}d_k^V, \quad 
    \hat b_k^V > \frac{3}{\phi_{\min}(1-\phi_{\min})(1-q)}q^{-(k-1)}d_k^V. 
\end{align}
By using the first and second inequalities described above, we have
\begin{align}
    \hat{b}_k^V
    \ge
    d_k^V, \quad
    \phi_k
    \ge
    \phi_{\min}
    >
    2\frac{d_k^V}{\hat{b}_k^V}.
\end{align}
Thus, the diameter-buffer condition and the stopping condition are satisfied.

Moreover, substituting Eq.~\eqref{eq:candidate_buffer} into
Eq.~\eqref{eq:phi_decrease_upper_bound} gives
\begin{align}
    \phi_k-\phi_{k+1}
    &\le
    \frac{3}
    {\phi_{\min}C}
    q^{k-1}
    \\
    &\le
    (1-\phi_{\min})(1-q)q^{k-1}
    \\
    &=
    \delta_k.
\end{align}
Since $\phi_k\ge\overline{\phi}_k$, it follows that
\begin{align}
    \phi_{k+1}
    &\ge
    \phi_k-\delta_k
    \\
    &\ge
    \overline{\phi}_k-\delta_k
    \\
    &=
    \overline{\phi}_{k+1}.
\end{align}
Therefore, $\hat{b}_k^V$ is an admissible choice in the
minimal-parameter construction. Since $b_k^V$ is chosen as the smallest admissible buffer parameter, there exists a constant $C_b>0$, independent of $k$, such that
\begin{align}
    b_k^V
    \le
    C_b q^{-(k-1)}d_k^V.
    \label{eq:b_growth_bound}
\end{align}

Using the recursive choice
\begin{align}
    d_{k+1}^V
    =
    3d_k^V+4b_k^V+2w_{\max},
\end{align}
together with Eq.~\eqref{eq:b_growth_bound}, we obtain
\begin{align}
    d_{k+1}^V
    &\le
    \left(
        3+4C_b q^{-(k-1)}
    \right)d_k^V+2w_{\max}
    \\
    &\le
    C_d q^{-k}d_k^V
\end{align}
for some constant $C_d>0$ independent of $k$.
Iterating this inequality yields
\begin{align}
    d_k^V
    &\le
    d_1^V
    \prod_{j=1}^{k-1}
    \left(
        C_d q^{-j}
    \right)
    \\
    &=
    d_1^V
    C_d^{k-1}
    q^{-\frac{k(k-1)}{2}}.
\end{align}
Taking the logarithm, we find
\begin{align}
    \log d_k^V
    &\le
    \log d_1^V
    +(k-1)\log C_d
    +\frac{k(k-1)}{2}\log\frac{1}{q}
    \\
    &=
    O(k^2).
\end{align}
Consequently,
\begin{align}
    d_k^V
    =
    \exp[O(k^2)].
\end{align}

Finally, Eq.~\eqref{eq:b_growth_bound} gives
\begin{align}
    \log b_k^V
    &\le
    \log C_b
    +(k-1)\log\frac{1}{q}
    +\log d_k^V
    \\
    &=
    O(k^2),
\end{align}
and hence
\begin{align}
    b_k^V
    =
    \exp[O(k^2)].
\end{align}
Therefore, both parameter sequences satisfy
\begin{align}
    d_k^V,b_k^V
    =
    \exp[O(k^2)].
\end{align}
\end{proof}

Define the corresponding edge-based parameters by
\begin{align}
    d_k^E
    &:=
    d_k^V-w_{\max},\\
    b_k^E
    &:=
    b_k^V+w_{\max}.
\end{align}
Then $d_1^E=1>0$, and the recursive relation for $d_k^V$ gives
\begin{align}
    d_{k+1}^E
    &=
    d_{k+1}^V-w_{\max}\\
    &=
    3d_k^V+4b_k^V+w_{\max}\\
    &=
    3d_k^E+4b_k^E.
\end{align}
Hence, the edge-based parameters satisfy the separation condition required by Lemma~\ref{lem:threshold_theorem}.

By Lemma~\ref{lem:order_of_scheduled_parameters}, there exist constants $\beta,\gamma>0$ and $\lambda>1$ such that the parameter sequences generated by the finite-size parameter schedule defined above satisfy the upper bounds required by Lemma~\ref{lem:threshold_theorem} with $f(k)=k^2$.
Since
\begin{align}
    \sum_{n=2}^{\infty}\frac{n^2}{2^{n-1}}=11,
\end{align}
Lemma~\ref{lem:threshold_theorem} therefore implies the existence of a strictly positive error-clustering threshold $p_{\mathrm{th}}>0$.

\begin{theorem}[Average parallel runtime with the scheduled parameters]
For every fixed physical error rate $p<p_{\mathrm{th}}$, the average runtime of parallel sparse blossom using the finite-size parameter schedule defined above is upper bounded by
\begin{align}
    \exp\left[O\left((\log\log d)^2\right)\right].
\end{align}
\end{theorem}

\begin{proof}
The proof follows the same argument as that of Theorem~\ref{thm:supp_parallel_ave_time}.
The only difference is the growth rate of the clustering parameters.

By Lemma~\ref{lem:order_of_scheduled_parameters}, there exists a constant $\lambda>1$ such that
\begin{align}
    d_k^V,b_k^V
    =
    \lambda^{O(k^2)}.
\end{align}
Hence, for some constant $\tau>0$, the execution time of a level-$k$ processing cluster is bounded by
\begin{align}
    \lambda^{\tau k^2},
\end{align}
while the corresponding clustering step has parallel runtime at most
\begin{align}
    \operatorname{polylog}(d)\lambda^{\tau k^2}.
\end{align}

Moreover, by the same residual-edge argument as in the proof of Theorem~\ref{thm:supp_parallel_ave_time}, for every fixed $p<p_{\mathrm{th}}$,
\begin{align}
    \Pr[\mathcal{P}_k]
    \le
    \min\left\{
        Cd^3\rho^{2^{k-1}},
        1
    \right\},
    \qquad
    \rho:=\frac{p}{p_{\mathrm{th}}}<1,
\end{align}
for some constant $C>0$ independent of $k$ and $d$.

Therefore, the same splitting argument as in the proof of Theorem~\ref{thm:supp_parallel_ave_time}, now with the factor $\lambda^{\tau k^2}$, gives
\begin{align}
    \mathbb{E}[T_{\mathrm{exec}}]
    \le
    \exp\left[
        O\left((\log\log d)^2\right)
    \right].
\end{align}
The same bound holds for the processing-cluster construction up to an additional $\operatorname{polylog}(d)$ factor, which is absorbed into the same asymptotic bound.
Hence,
\begin{align}
    \mathbb{E}[T_{\mathrm{cluster}}+T_{\mathrm{exec}}]
    \le
    \exp\left[
        O\left((\log\log d)^2\right)
    \right].
\end{align}
\end{proof}

\end{document}